\pdfoutput=1
\documentclass[12pt,a4paper]{amsart}
\usepackage[a4paper,inner=2.5cm,outer=2.5cm,top=2.5cm,bottom=2.5cm]{geometry}
\usepackage{float}
\usepackage{amsmath,amssymb,amsthm,enumerate,mathtools,stmaryrd}
\usepackage{hyperref}
\hypersetup{colorlinks=true,linkcolor=blue,citecolor=teal,filecolor=magenta,urlcolor=cyan}

\usepackage{makecell}
\usepackage{graphicx}

\usepackage{textcomp}

\usepackage{extarrows}

\usepackage{xcolor}

\newcommand{\C}{\mathbb{C}}

\theoremstyle{plain}
\newtheorem{theorem}{Theorem}[section]
\newtheorem{proposition}[theorem]{Proposition}
\newtheorem{corollary}[theorem]{Corollary}
\newtheorem{lemma}[theorem]{Lemma}
\newtheorem*{genprin}{General Principle}
\theoremstyle{definition}
\newtheorem{definition}[theorem]{Definition}

\newtheorem{notation}[theorem]{Notation}
\theoremstyle{remark}
\newtheorem{remark}[theorem]{Remark}

\newcommand{\Z}{\mathbb{Z}}
\newcommand{\cD}{\mathcal{D}}
\newcommand{\cS}{\mathcal{S}}
\newcommand{\cA}{\mathcal{A}}
\newcommand{\cE}{\mathcal{E}}

\newcommand{\cP}{\mathcal{P}}
\newcommand{\cM}{\mathcal{M}}
\newcommand{\cW}{\mathcal{W}}
\newcommand{\cT}{\mathcal{T}}

\newcommand{\VEV}[1]{{\big\langle 0 \big| {#1} \big| 0 \big\rangle}}
\newcommand{\VEVc}[1]{{\big\langle 0 \big| {#1} \big| 0 \big\rangle^\circ}}

\newcommand{\res}{\mathop{\rm res}}

\newcommand{\Hc}{H}

\newcommand{\euro}{\mathbb{J}}
\newcommand{\Pt}{\widetilde{P}}

\begin{document}

\title[Topological recursion for KP tau functions of hypergeometric type]{Topological recursion for Kadomtsev--Petviashvili  tau functions of hypergeometric type}

\author[B.~Bychkov]{Boris~Bychkov}
\address{B.~B.: Faculty of Mathematics, HSE University, Usacheva 6, 119048 Moscow, Russia; and Center of Integrable Systems, P.G. Demidov Yaroslavl State University, Sovetskaya 14, 150003,Yaroslavl, Russia; and current affiliation: Department of Mathematics, University of Haifa, Mount Carmel, 3498838, Haifa, Israel}
\email{bbychkov@hse.ru}

\author[P.~Dunin-Barkowski]{Petr~Dunin-Barkowski}
\address{P.~D.-B.: Faculty of Mathematics, HSE University, Usacheva 6, 119048 Moscow, Russia; HSE--Skoltech International Laboratory of Representation Theory and Mathematical Physics, Skoltech, Bolshoy Boulevard 30 bld. 1, 121205 Moscow, Russia; and NRC “Kurchatov Institute” -- ITEP, 117218 Moscow, Russia}
\email{ptdunin@hse.ru}

\author[M.~Kazarian]{Maxim~Kazarian}
\address{M.~K.: Faculty of Mathematics, HSE University, Usacheva 6, 119048 Moscow, Russia; and Igor Krichever Center for Advanced Studies, Skoltech, Bolshoy Boulevard 30 bld. 1, 121205 Moscow, Russia}
\email{kazarian@mccme.ru}

\author[S.~Shadrin]{Sergey~Shadrin}
\address{S.~S.: Korteweg-de Vries Institute for Mathematics, University of Amsterdam, Postbus 94248, 1090 GE Amsterdam, The Netherlands}
\email{S.Shadrin@uva.nl}	

\begin{abstract}
	We study the $n$-point differentials corresponding to Kadomtsev--Petviashvili tau functions of hypergeometric type (also known as Orlov--Scherbin partition functions), with an emphasis on their $\hbar^2$-deformations and expansions. 
	
	Under the naturally required analytic assumptions, we prove certain higher loop equations that, in particular, contain the standard linear and quadratic loop equations, and thus imply the blobbed topological recursion. We also distinguish two large families of the Orlov--Scherbin partition functions that do satisfy the natural analytic assumptions, and for these families we prove  in addition the so-called projection property and thus the full statement of the Chekhov--Eynard--Orantin topological recursion.
	
	A particular feature of our argument is that it clarifies completely the role of $\hbar^2$-deformations of the Orlov--Scherbin parameters for the partition functions, whose necessity was known from a variety of earlier obtained results in this direction but never properly understood in the context of topological recursion.
	
	As special cases of the results of this paper one recovers new and uniform proofs of the topological recursion to all previously studied cases of enumerative problems related to weighted double Hurwitz numbers. By virtue of topological recursion and the Grothendieck--Riemann--Roch formula this, in turn, gives new and uniform proofs of almost all ELSV-type formulas discussed in the literature.
\end{abstract}
\dedicatory{To the memory of Sergey Mironovich Natanzon}

\keywords{Hurwitz numbers, spectral curve topological recursion,  KP tau functions, ELSV-type formulas}
\subjclass{37K10, 14H81, 81T45, 53D45, 14H10}

\maketitle	
	

	
\tableofcontents
\section{Introduction}

There exists ample literature studying generating functions of various Hurwitz-type numbers in relation to the so-called \emph{spectral curve topological recursion}.
Up until now this has mostly been done on a case-by-case basis (with a notable, but still restricted, exception of the effort due to Alexandrov--Chapuy--Eynard--Harnad~\cite{ACEH-1,ACEH-2} which does not cover e.g. the important examples of the $r$-spin Hurwitz numbers and the Ooguri-Vafa partition function), with complicated separate proofs of topological recursion for many various cases of Hurwitz-type problems.

In the present paper, which is based on the results of our previous paper~\cite{BDKS20}, we give a new unified approach, which allows us to prove topological recursion for weighted Hurwitz numbers of very general type.
The result of the present paper covers \emph{all} known results regarding topological recursion for Hurwitz-type enumerative problems, and substantially extends them, all in a neat and uniform way, revealing the underlying general structure and the reasons behind these results.

In order to guide the reader who does not need an introduction into these topics through the paper, let us immediately link the main results and applications:
\begin{itemize}
	\item Theorem~\ref{thm:Blobbed} proves blobbed topological recursion / the loop equations for general $\hbar^2$-deformed hypergeometric KP tau functions, under natural analytic assumptions. It is the most general statement that can be obtained in this context. 
	\item Theorems~\ref{thm:Family1TR} and~\ref{thm:Family2TR} provide two enormously large families of cases where we can prove the topological recursion. This simultaneously resolves a huge number of open questions on topological recursion in the individual cases (see Section~\ref{sec:moreTR} for an example of this) and provides new proofs and conceptual framework for the tens of results that can now be seen as special cases of our theorems (see Section~\ref{sec:newprooftr}). 
	\item Application: a direct corollary of  Theorems~\ref{thm:Family1TR} and~\ref{thm:Family2TR} and prior computations with the Grothendieck--Riemann--Roch formula is a new and uniform proof of almost all known ELSV-type formulas (as the original Ekedahl--Lando--Shapiro--Vainshtein formula, Gopakumar--Mari\~no--Vafa formula, Zvonkine's $r$-spin formula, etc. etc.). We discuss this application in Section~\ref{sec:ELSV}.
\end{itemize}

\subsection{Topological recursion} Let $\Sigma$ be a Riemann surface, $x$ and $y$ two functions on $\Sigma$ such that $dx$ is meromorphic and has isolated simple zeros $p_1,\dots,p_N\in \Sigma$, $y$ is holomorphic near $p_i$ and 
$dy|_{p_i}\not=0$, $i=1,\dots,N$.  Let $B$ be a meromorphic symmetric bi-differential on $\Sigma\times\Sigma$ with the only pole being the order $2$ pole on the diagonal with the bi-residue equal to $1$.

The theory of topological recursion, due to Chekhov, Eynard, and Orantin~\cite{CE1,CE2,EynardOrantin,EO-2,EynardSurvey}, associates to this data a sequence of symmetric
differentials meromorphic near the points $p_1,\dots,p_N$. For the so-called \emph{unstable cases} they are given by 
$\omega_{0,1}\coloneqq ydx$, $\omega_{0,2}\coloneqq B$; while in general, for $g\geq 0$, $n\geq 1$, $2g-2+n>0$, these symmetric $n$-differentials  $\omega_{g,n}$ on $\Sigma^n$ are given by an explicit recursive procedure:
\begin{align}\label{eq:toprec}
\omega_{g,n+1}(z_0,z_{\llbracket n \rrbracket}) \coloneqq \frac 12 \sum_{i=1}^N \res_{w\to p_i} \frac{\int_w^{\sigma_i(w)} B(z_0,\cdot)}{\omega_{0,1}(\sigma_i(w))-\omega_{0,1}(w)} \Bigg( \omega_{g-1,n+2}(w, \sigma_i(w), z_{\llbracket n \rrbracket})+
\\ \notag
\sum_{\substack{g_1+g_2=g,\ I_1\sqcup I_2 = \llbracket n \rrbracket \\ (g_1,|I_1|),(g_2,|I_2|)\not=(0,0)
}}
\omega_{g_1,|I_1|+1} (w,z_{I_1})\omega_{g_2,|I_2|+1} (\sigma_i(w),z_{I_2})
\Bigg).
\end{align}
Here $\sigma_i$ is a deck transformation of $x$ near the point $p_i$, $i=1,\dots,N$, by $\llbracket n \rrbracket$ we denote the set $\{1,\dots,n\}$, 
and $z_I$ denotes $\{z_i\}_{i\in I}$ for any $I\subseteq \llbracket n \rrbracket$. If $g=0$, then we assume that $\omega_{g-1,n+2}=0$. Here and everywhere below, if not specified otherwise, a sum of the form $\sum_{I_1\sqcup\dots\sqcup I_k = A}$ is understood as a sum over ordered collections of sets which are allowed to be empty.

This recursive procedure comes from a recursion for the computation of cumulants in the matrix models theory (see also \cite{AMM06a,AMM06}), and in a large number of applications it can be used to fully replace and eliminate the underlying matrix models. For instance, the theory of topological recursion is the base for the so-called remodeling approach in the type B topological string theory~\cite{BKMP-remodellingB}. Our main motivation to study the theory of topological recursion is the fact that it has appeared to be a universal interface to connect a huge variety of combinatorial and algebraic enumerative questions to the intersection theory of the moduli space of curves and cohomological field theories~\cite{Eynard,DOSS}.

\subsubsection{Terminology} The data $(\Sigma, x, y, B)$ is called the spectral curve data.

\subsubsection{Information contained in $\omega_{g,n}$'s}
The two main questions of the theory of topological recursion are the following:
\begin{enumerate}[a)]
	\item For a given spectral curve data, what do $\omega_{g,n}$'s compute? For instance, one can consider the formal expansions of $\omega_{g,n}$ and ask whether these numbers have any interpretation outside the theory of topological recursion.
	\item Assume that some combinatorial or enumerative geometric problem produces numbers  $h_{g;k_1,\dots,k_n}$ for any $g\geq 0$ and $k_1,\dots,k_n\geq 1$. Is there any spectral curve data $(\Sigma,x,y,B)$, and a choice of a point $p_0\in \Sigma$ and a local coordinate $w$ near $p_0$ such that the expansion of $\omega_{g,n}$ in $(w_1,\dots,w_n)$ near $(p_0,\dots,p_0)\in \Sigma^n$ is given by
	\begin{equation}
		\sum_{k_1,\dots,k_n=1}^\infty h_{g;k_1,\dots,k_n} \prod_{i=1}^n dw_i^{k_i}\quad?
	\end{equation}
\end{enumerate}
Both questions have multiple interesting answers in the literature; for instance, a universal answer to the first question was found in the theory of Dubrovin--Frobenius manifolds, where one can connect the spectral curve data to Dubrovin's superpotential~\cite{DNOPS-1,DNOPS-2}. The non-uniqueness of the answers to these questions allows to establish new instances of the general phenomena of classical mirror symmetry.

\subsubsection{Loop equations and projection property} Let $(\Sigma,x,y,B)$ be a spectral curve data. Let $p_0\in\Sigma$ be a point of $\Sigma$. Consider a system of symmetric $n$-differentials $\omega_{g,n}$, $g\geq 0$, $n\geq 1$, $2g-2+n>0$, defined in a vicinity of $(p_0,\dots,p_0)$ of the $n$th Cartesian power of $\Sigma$.
Define, in the unstable cases $2g-2+n\leq0$, explicitly $\omega_{0,1}=ydx$ and $\omega_{0,2}=B$.

It is proved in~\cite[Theorem 2.2]{BS17}, see also~\cite{BEO13}, that a system of symmetric differentials $\{\omega_{g,n}\}_{g\geq 0, n\geq 1}$ satisfies the topological recursion if and only if it satisfies the following collection of properties discussed below in more details:
\begin{itemize}
\item meromorphy; 
\item linear loop equation;
\item quadratic loop equation;
\item projection property.
\end{itemize}

\begin{definition}
We say that the differentials $\omega_{g,n}$ satisfy the \emph{property of being meromorphic} if for any $g\geq0$, $n\geq 1$, $2g-2+n>0$  the form $\omega_{g,n}$ extends as a global meromorphic $n$-differential on the whole $\Sigma^{n}$.
\end{definition}

Even if we are interested in the power expansion of the forms $\omega_{g,n}$ at one point, the property of being meromorphic 
is essential since the topological recursion involves the local behavior of the analytical extension of these forms near the points $p_i$, $i=1,\dots N$.

\begin{definition}
The differentials $\omega_{g,n}$ satisfy the \emph{linear loop equations} if for any $g,n\geq 0$, $i=1,\dots,N$
\begin{equation}
	\omega_{g,n+1}(w,z_{\llbracket n \rrbracket}) + \omega_{g,n+1}(\sigma_i(w),z_{\llbracket n \rrbracket})
\end{equation}
is holomorphic at $w\to p_i$ and has at least a simple zero at $w=p_i$. We  say that the differentials $\omega_{g,n}$ satisfy the \emph{quadratic loop equations} if for any $g,n\geq 0$, $i=1,\dots,N$, the quadratic differential in $w$
\begin{equation}
	\omega_{g-1,n+2}(w,\sigma_i(w),z_{\llbracket n \rrbracket}) + \sum_{\substack{g_1+g_2=g\\ I_1\sqcup I_2 = \llbracket n \rrbracket
	}} \omega_{g_1,|I_1|+1} (w,z_{I_1})\omega_{g_2,|I_2|+1} (\sigma_i(w),z_{I_2})
\end{equation}
is a holomorphic at $w\to p_i$ and has at least a double zero at $w=p_i$.
\end{definition}

\begin{definition} We say that the differentials $\omega_{g,n}$ satisfy the \emph{projection property}, if for any $g\geq 0$, $n\geq 1$, $2g-2+n>0$
\begin{equation}
	\omega_{g,n}(z_1,\dots,z_n) = \sum_{i_1,\dots,i_n=1}^N \Bigg(\prod_{j=1}^n \res_{w_j\to p_{i_j}} \int^{w_j}_{p_{i_j}} \omega_{0,2}(w_j,z_j)\Bigg) \omega_{g,n}(w_1,\dots,w_n).
\end{equation}
\end{definition}

The linear and quadratic loop equations determine uniquely the principal part of the poles of $\omega_{g,n}$ at $p_i$ considered as a meromorphic $1$-form in its first argument. The projection property allows one to recover the form $\omega_{g,n}$ from the principal parts of its poles. The projection property implies that $\omega_{g,n}$ has no poles other than $p_i$ in each of its arguments. Moreover, in the present paper we are mostly interested in the case when $\Sigma\simeq\C P^1$. In this case the form $B$ satisfying the conditions on 
the spectral curve data is unique and the \emph{projection property is equivalent to the condition that $\omega_{g,n}$ has no poles other than $p_i$ in each of its arguments}.

This equivalent reformulation of topological recursion is used either directly or indirectly in many proofs of the topological recursion, and it was at the roots of the derivation of the topological recursion in the original papers~\cite{CE1,CE2,EynardOrantin}.

\subsubsection{Blobbed topological recursion} Let $(\Sigma,x,y,B)$ be a spectral curve data. Consider a system of the symmetric $n$-differentials $\omega_{g,n}$ on $\Sigma^n$, $g\geq 0$, $n\geq 1$, such that $\omega_{0,1}=ydx$, $\omega_{0,2}=B$. Assume that $\{\omega_{g,n}\}$ satisfy the quadratic and the linear loop equations but not necessarily the projection property.

It appears that in this case the $n$-differentials still retain many interesting properties, in particular, they can be expanded in the sums over the so-called blobbed graphs~\cite{BS17} that generalize the Givental-type sums over graphs known in the topological recursion theory~\cite{Eynard,DOSS}. There is a growing literature on the examples that do not satisfy the topological recursion but still fit into the framework of the blobbed topological recursion, cf.~\cite{BonzomBlobbed,branahl2020blobbed}.

\subsection{KP tau functions of hypergeometric type} Let $s_\lambda(p_1,p_2,\dots)$, $\lambda\vdash n$, $n\geq 0$,  be the Schur functions in the power sum variables $p_i$, $i\geq 1$. Consider two formal power series, $\psi(y)$ and $y(z)$ such that $\psi(0)=0$ and $y(0)=0$. Let $y(z)=\sum_{i=1}^\infty q_iz^i$. The following function
\begin{equation}
	\sum_{\lambda} s_\lambda(p_1,p_2,\dots) s_\lambda(q_1,q_2,\dots) \exp\Bigg(\sum_{(i,j)\in \lambda} \psi(i-j)\Bigg)
\end{equation}
is a tau function of the Kadomtsev--Petviashvili hierarchy~\cite{MiwaJimboDate} (the standard KP variables $t_1,t_2,\dots$ are related to $p_1,p_2,\dots$ by $t_i=p_i/i$, $i=1,2,\dots$). It is usually called the tau function of hypergeometric type, or Orlov--Scherbin tau function~\cite{Kharchev,OrlovScherbin}.

\subsubsection{Parameter $\hbar$} The hypergeometric KP tau functions can also be considered as tau functions of the $\hbar$-KP hierarchy of Takasaki--Takebe~\cite{TakasakiTakebe,NatanzonZabrodin}, where the parameter $\hbar$ is introduced as follows:
\begin{equation}
 	\sum_{\lambda} s_\lambda\big(\frac{p_1}\hbar,\frac{p_2}\hbar,\dots\big) s_\lambda\big(\frac {q_1}{\hbar},\frac{q_2}{\hbar},\dots\big) \exp\Bigg(\sum_{(i,j)\in \lambda} \psi(\hbar(i-j))\Bigg),
\end{equation}
and a natural generalization that often occurs in the literature includes the possibility for the coefficients of $\psi$ and $y$ to depend on $\hbar^2$, cf.~\cite{KazLand,andreev2020genus,DKPSS}. That is, we can consider
\begin{align}\label{def:psihat}
	\hat \psi(\hbar^2,y)&:=\sum_{k+m\geq1}^\infty c_{k,m} y^k\hbar^{2m},\\
	\label{def:yhat}
	\hat y(\hbar^2,z)&:=\sum_{k=1}^\infty\sum_{m=0}^\infty q_{k,m} z^k\hbar^{2m}.
\end{align}
and define $\hat q_i(\hbar)$ via $\hat y(\hbar^2,z)=\sum_{i=1}^\infty \hat q_i(\hbar^2) z^i$. Then for the KP partition function
\begin{equation} \label{eq:ParitionFunctionHBAR}
	Z_{\hat\psi,\hat y}(p_1,p_2,\dots)\coloneqq \sum_{\lambda} s_\lambda(p_1,p_2,\dots) s_\lambda\big(\frac {\hat q_1(\hbar^2)}{\hbar},\frac{\hat q_2(\hbar^2)}{\hbar},\dots\big) \exp\Bigg(\sum_{(i,j)\in \lambda} \hat\psi(\hbar^2, \hbar(i-j))\Bigg),
\end{equation}
its rescaled version $Z_{\hat\psi,\hat y}\big(\frac{p_1}\hbar,\frac{p_2}\hbar,\dots\big)$ is still a tau function of the $\hbar$-KP hierarchy.

\subsubsection{Combinatorial and enumerative meaning}
The functions $Z_{\hat\psi,\hat y}$ given by $\eqref{eq:ParitionFunctionHBAR}$ are expanded in $\hbar$ as
\begin{equation}
	Z_{\hat\psi,\hat y}\big(\frac{p_1}\hbar,\frac{p_2}\hbar,\dots\big)=\exp\left(\sum_{g=0}^\infty 
	\hbar^{2g-2} F_{g}(p_1,p_2,\dots)\right).	
\end{equation}
They are intensively studied for various specific choices of $\hat \psi$ and $\hat y$ since the numbers
\begin{equation}\label{eq:h-small-definition}
	h_{g;k_1,\dots,k_n}\coloneqq 
	\frac{\partial^n F_g}{\prod_{i=1}^n \partial p_{k_i}}\Bigg|_{p_1=p_2=\cdots=0}
\end{equation}
encode a huge variety of enumerative questions: various kinds of (weighted) Hurwitz numbers, enumeration of Grothendieck's dessins d'enfants and more general types of polygonal decompositions of surfaces and constellations, relative Gromov-Witten theory of the projective line, the colored HOMFLY-PT polynomials for the torus knots, etc. There is enormous literature on the interpretations of the expansions of these partition function and the combinatorial equivalences that these interpretations imply, cf.~\cite{HarnadOrlov,harnad2020constellations,ALS} and see further references in Sections~\ref{sec:PreviousQuasi} and~\ref{sec:appl}.

\subsubsection{Terminology} The formal power series
\begin{equation}
	H_{g,n}\coloneqq \sum_{k_1,\dots,k_n=1}^\infty h_{g;k_1,\dots,k_n} \prod_{i=1}^nX_i^{k_i}
\end{equation}
are called the \emph{$n$-point functions} of the tau function $Z_{\hat \psi, \hat y}$.

\subsubsection{Spectral curve}
Let $\psi(y)=\hat\psi(0,y)$ and $y(z)=\hat y(0,z)$.
Define the spectral curve data $(\Sigma,x,y,B)$ associated with the partition function $Z_{\hat \psi, \hat y}$ as follows. We take $\Sigma=\C P^1$ with global affine coordinate~$z$, the functions $x=x(z)$, $y=y(z)$ are given by
\begin{align}
y&=y(z),\\
x&=\log(z)- \psi(y(z)),
\end{align}
and $B(z_1,z_2)=\frac{dz_1\,dz_2}{(z_1-z_2)^2}$. This spectral curve data was first proposed in ~\cite{ACEH-1}.

All these functions are defined in a formal vicinity of the origin, and in order to be able to discuss topological recursion we will impose below certain analytic assumptions implying, in particular, that the $1$-form $dx$ is rational.

\subsubsection{Relations to topological recursion} The hypergeometric KP tau functions also give one of the most striking examples of topological recursion.
Consider the following local change of coordinates at the origin
\begin{equation}\label{eq:X(z)}
X_i=X(z_i),\qquad X(z)=e^{x(z)}=z\,e^{-\psi(y(z))}.
\end{equation}

Applying this change we can regard $H_{g,n}$ as a symmetric function on $\Sigma^n$ defined in a vicinity of the origin. Along with $n$-point functions introduce also the \emph{$n$-point differentials}
\begin{equation}\label{eq:omegagn}
\omega_{g,n}=d_1\dots d_n H_{g,n}+\delta_{g,0}\delta_{n,2}\frac{dX_1\,dX_2}{(X_1-X_2)^2}.
\end{equation}
By definition, these are symmetric $n$-differentials on $\Sigma^{n}$ defined in a vicinity of the origin.
We prove in~\cite{BDKS20} that in the unstable cases $2g-2+n\le0$ they are given explicitly by
\begin{align*}
\omega_{0,1}&=y_1\,dx_1,\\
\omega_{0,2}&=\frac{dz_1\,dz_2}{(z_1-z_2)^2},
\end{align*}
where $x_i=x(z_i)$, $y_i=y(z_i)$, $X_i=X(z_i)$. These relations suggest that there should be relationship between the forms $\omega_{g,n}$ and the forms obtained by the topological recursion with the spectral curve data $(\Sigma,x,y,\omega_{0,2})$. There is a huge number of results (see the references in Section~\ref{sec:appl}) that all together imply the following general principle:

\begin{genprin} If $y(z)$ can be extended as a meromorphic function defined on $\mathbb{C}\mathrm{P}^1$ with a global coordinate $z$, and the 1-form $dx(z)$ with  $x(z)\coloneqq \log z - \psi(y(z))$ is meromorphic on $\mathbb{C}\mathrm{P}^1$, and they satisfy the assumptions of topological recursion, then there exist suitable $\hbar^2$-deformations $\hat\psi(\hbar^2,y)$ and $\hat y(\hbar^2,z)$ of $\psi(y)=\hat\psi(0,y)$ and $y(z)=\hat y(0,z)$ such that the topological recursion applied to the spectral curve data
\begin{equation}	
	(\mathbb{C}\mathrm{P}^1,
	x(z)=\log z - \psi(y(z)),
	y(z),
	B=\tfrac{dz_1\,dz_2}{(z_1-z_2)^2}
)
\end{equation}
returns the $n$-differentials $\omega_{g,n}$, $g\geq 0$, $n\geq 1$, whose power expansion in the variables $X_i=\exp(x(z_i))$, $i=1,\dots,n$ at $z_1=\cdots =z_n=0$ is given by~\eqref{eq:omegagn}
where the $H_{g,n}$ are the $n$-point functions of the partition function $Z_{\hat\psi,\hat y}$.
\end{genprin}

There are many cases when the $\hbar^2$-deformation is not needed and we have $\hat\psi(\hbar^2,y)=\psi(z)$, $\hat y(\hbar^2,z)=y(z)$~\cite{KZ2015,ACEH-2,KPS,BDKLM}, and there are also known cases where the $\hbar^2$-deformations are inevitable~\cite{DKPS,DKPSS}. It is not known whether the General Principle, formulated this way, is actually working in full generality. But it works in all studied examples and can be even extended to the more general cases, for instance, there is a version of topological recursion that does not require the zeros of $dx$ to be simple~\cite{BE13}.

\subsubsection{Goal of the paper} The main goal of the present paper is to push the general principle as far as possible, namely, we explore what can be the most general statement on topological recursion that would fall under the General Principle. To this end we find two families of the data $(\hat \psi, \hat y)$ when the General Principle does work, and these families cover almost all sensible choices for $\psi$ and $y$ that lead to a rational spectral curve.

\subsection{Results of this paper}\label{sec1.3}
We refer to the assumptions formulated in the following definition as the \emph{natural analytic assumptions} on $\hat\psi$ and $\hat y$:
\begin{definition}[Natural analytic assumptions]\label{def:naa}
	Let $\hat\psi(\hbar^2,y)$ and $\hat y(\hbar^2,z)$ be arbitrary formal power series such that $\hat\psi(\hbar^2,0)=0$ and $\hat y(\hbar^2,0)=0$. Let $\psi(y)\coloneqq \hat \psi(0,y)$ and $y(z)\coloneqq \hat y(0,z)$. Assume that the series $\frac{d\psi(y)}{dy}|_{y=y(z)}$ and $\frac{dy(z)}{dz}$ extend analytically as rational functions in~$z$. This assumption implies that the series $X(z)\coloneqq \exp(x)= z\exp(-\psi(y(z)))$ has a non-zero radius of convergence and extends analytically as a function on the Riemann surface $\Sigma=\C P^1$ with affine coordinate~$z$ such that the $1$-form
\begin{equation}
dx=\frac{dX}{X}=(1-z\, y'(z)\,\psi'(y(z)))\frac{dz}{z}
\end{equation}
is meromorphic (i.e.~rational) and has a finite number of zeros (denoted by $N$ and $p_1,\dots,p_N$, respectively). Besides, we assume that $dy$ and $\frac{d\psi(y)}{dy}|_{y=y(z)}$ are regular at the zeros of $dx$; and also that all coefficients of positive powers of $\hbar$ in the series $\hat\psi(\hbar,y(z))$ and $\hat y(\hbar,z)$ are rational functions in~$z$ whose singular points are different from the zeros of~$dx$.
\end{definition}
Here and below when we say that a function or a differential is regular at some point we mean that it is holomorphic in a neighborhood of this point.

 Without these assumptions formulated in Definition~\ref{def:naa} the very concept of topological recursion is not even well defined. Note that even though $dx$ is a rational $1$-form neither $x(z)$ nor $y(z)$ are assumed to be rational or even univalued: they may contain logarithmic summands and the procedure of topological recursion is still applicable.

\subsubsection{Loop equations}
Let $\omega_{g,n}$, $g\geq 0$, $n\geq 1$, be symmetric $n$-differentials associated with the partition function $Z_{\hat\psi,\hat y}$ and defined by~\eqref{eq:omegagn}.

\begin{theorem}\label{th:blobbed} Under the natural analytic assumptions of Definition~\ref{def:naa}, the symmetric differentials $\omega_{g,n}$ for $2g-2+n>0$ can be extended analytically to $\Sigma^n$, $\Sigma=\C P^1$ as global rational forms. Furthermore, $\omega_{g,n}$'s satisfy the quadratic and the linear loop equations at any zero point $p_i$ of $dx$ provided that this zero is simple.
\end{theorem}
In other words, the thus defined $\omega_{g,n}$'s  satisfy the \emph{blobbed topological recursion} in the sense of~\cite{BS17}. This theorem is proved as Theorem~\ref{thm:Blobbed} in the main part of the text.

The idea of the proof is as follows. We derived in~\cite{BDKS20} an explicit closed formula for the $n$-point functions $H_{g,n}$ and thus for $\omega_{g,n}$, see also Proposition~\ref{prop:FormulasWgn}. This formula holds true without any analytic assumption on $(\hat\psi,\hat y)$. However, if the assumptions are satisfied then the very structure of this formula implies immediately both the rationality of $\omega_{g,n}$ and the linear loop equations. In Section~\ref{SecPrel} we derive a similar formula for the combination of $n$-point differentials participating in the quadratic loop equation (Proposition~\ref{LEtheor}). Again, the formula itself is valid without any analytical assumptions but as long as they are satisfied the quadratic loop equation is a straightforward consequence of the very structure of the formula.

\begin{remark}In fact, we prove a much more general statement on a family of loop equations that naturally extends the linear and quadratic ones, see Theorem~\ref{theor:higher}. A different but equivalent form of these higher loop equations is considered in~\cite{DKPS}. 
	
	It is proved in \emph{op.~cit.} that the higher loop equations are formal corollaries of the linear and the quadratic ones in general (but the linear and quadratic loops equations are only proved there for the specific situation of orbifold Hurwitz numbers with completed cycles).
		
In the present paper we do not derive the higher loop equations from the linear and quadratic ones, but rather provide an independent direct proof of all these equations separately for all cases satisfying the natural analytic assumptions.
\end{remark}

\subsubsection{Topological recursion} Taking into account Theorem~\ref{th:blobbed}, the only missing component to prove topological recursion is the projection property. The aforementioned formula for $H_{g,n}$ implies that under the natural analytic assumptions these functions being rational in $z_1,\dots,z_n$ along with the `expected' poles at $p_1,\dots,p_N$ may have extra `unwanted' poles. For example, there might be extra `unwanted' poles at the poles of the rational functions $\frac{d\psi(y)}{dy}|_{y=y(z)}$ and $\frac{dy(z)}{dz}$ or at the infinity in coordinate $z$.  In Section~\ref{sec:proj}, we study two very general families of data (in fact, covering most of the cases which satisfy the natural analytic assumptions) when one can choose the $\hbar^2$-deformations $\hat\psi$ and $\hat y$ of $\psi$ and $y$ such that all these extra poles cancel out. This implies that the projection property for these cases is also satisfied and the topological recursion holds true.

Specifically, consider the following families of the data $\hat\psi$ and $\hat y$:
\begin{align}
	\text{Family I:}   &\quad \hat\psi(\hbar^2,y) =  \cS(\hbar \partial_y)P_1(y)+\log P_2(y) -\log P_3(y);\quad \hat y(\hbar^2,z) = R_1(z)/R_2(z);\\ \notag
	\text{Family II:} &\quad \hat\psi(\hbar^2,y) = \alpha y;\quad \hat y(\hbar^2,z) = R_1(z)/R_2(z)+
	\cS(\hbar z\partial_z)^{-1} (\log R_3(z)- \log R_4(z)),
\end{align}
where $\alpha\neq 0$ is a number and $P_i$, $R_i$ are arbitrary polynomials such that $\hat\psi(0,y)$ and $\hat y(0,z)$ are both non-zero but vanishing at zero, and $\cS(u) \coloneqq \sinh(u/2)/u$.

In all these cases, let $\psi\coloneqq \hat \psi(0,y)$ and $y\coloneqq \hat y(0,z)$ and recall that $z$ is a global coordinate on $\mathbb{C}\mathrm{P}^1$. Let $x(z)\coloneqq \log z-\psi(y(z))$. Note that in all these cases  $dx(z)$ is a meromorphic $1$-form.
The assumptions of topological recursion can be reformulated as some general position requirements for $\psi$ and $y$. Under these assumptions consider the symmetric $n$-differentials $\omega_{g,n}$ produced by the topological recursion. We have:

\begin{theorem}\label{th:toporecIntro} The expansions of the $n$-differentials $\omega_{g,n}$, $g\geq 0$, $n\geq 1$, near the point $z_1=\dots=z_n=0$ in the coordinates $X_i=\exp(x(z_i))$, $i=1,\dots, n$, are given by $d_1\cdots d_n H_{g,n}$, where $H_{g,n}$ are the $n$-point functions of $Z_{\hat \psi, \hat y}$.
\end{theorem}
Moreover, we can relax the general position requirements, passing from topological recursion to the so-called \emph{Bouchard-Eynard recursion}. This is done in Section~\ref{sec:TopoRecAppl}, where we prove the more general Theorems~\ref{thm:Family1TR} and ~\ref{thm:Family2TR}, which imply Theorem~\ref{th:toporecIntro}.

Theorems~\ref{thm:Family1TR} and ~\ref{thm:Family2TR} subsume, as special cases, all results on topological recursion for the hypergeometric KP tau functions (i.e. all results on topological recursion for Hurwitz-type enumerative problems) obtained so far in the literature. We present a survey of these results in Section~\ref{sec:appl}. 

\subsubsection{ELSV-type formulas}
One can use our results on the topological recursion for the Hurwitz-type numbers to uniformly prove various ELSV-type formulas that generalize the classical formula of Ekedahl--Lando--Shapiro--Vainshtein~\cite{ELSV} and relate the various weighted Hurwitz numbers to the intersection theory of the moduli spaces of curves and cohomological field theories. In fact, if it is known that a particular enumerative problem satisfies the spectral curve topological recursion, the techniques of \cite{Eynard} and~\cite{DOSS} allow to rather straightforwardly deduce and prove the respective ELSV-type formula, cf. \cite{LPSZ,FangZong,BDKLM} where it is done for several particular cases. This is discussed in Section~\ref{sec:ELSV}. 

\subsection{Structure of the paper}
Section~\ref{SecPrel} is devoted to proving the higher loop equations (which imply the linear and quadratic ones) for all cases satisfying the natural analytic assumptions. In Section~\ref{sec:proj} we introduce two general families of data and formulate and discuss the theorems stating that for these families the projection property holds. Section~\ref{sec:ProjPropProofs} is devoted to the proofs of the theorems regarding the projection property, which are quite technical and involved. In Section~\ref{sec:TopoRecAppl} we recall the Bouchard-Eynard recursion and state the theorems that topological recursion (for the case of simple zeros of $dx$) and Bouchard-Eynard recursion (in general) hold for our aforementioned families; we also list all known literature on topological recursion for Hurwitz-type enumerative problems and discuss how these cases are subsumed by our theorems; furthermore we discuss the implications for ELSV-type formulas.

\subsection*{Acknowledgments} S.S. was supported by the Netherlands Organization for Scientific Research. All authors were partially supported by International Laboratory of Cluster Geometry NRU HSE, RF Government grant, ag.~\textnumero{}  075-15-2021-608 dated 08.06.2021.

We thank Alexander Alexandrov, Nitin Chidambaram, and Reinier Kramer for useful discussion. We are also grateful to the anonymous referee for useful remarks.

\section{Loop equations}\label{SecPrel}

The goal of this Section is to prove that the differentials of the $n$-point functions of $\hbar^2$-deformed KP tau functions of hypergeometric type satisfy the linear and quadratic loop equations, and, therefore, the blobbed topological recursion, under \emph{natural analytic assumptions}. To this end, we formulate a much more general set of \emph{higher loop equations} for the formal power series expansions of the $n$-point functions and prove all of them in a unified manner.

The \emph{natural analytic assumptions} of Definition~\ref{def:naa} that we have to impose is the minimal set of assumptions needed to be able to pass from the formal power series set-up that is natural for the $n$-point functions of $Z_{\hat\psi,\hat y}$ to the analytic set-up of topological recursion; see Theorem~\ref{th:linloop} and Section \ref{sec1.3} for the precise formulation of these assumptions.
Under this set of natural analytic assumptions, the higher loop equations we introduce below hold and imply the linear and the quadratic ones.

The key technical ingredient in the proof of higher loop equations is Proposition~\ref{LEtheor} which on its own does not require any assumptions on $\hat\psi$ and $\hat y$ and holds in the formal power series set-up.

\subsection{Closed formulas for the differentials of \texorpdfstring{$\hbar^2$}{h2}-deformed \texorpdfstring{$n$}{n}-point functions}

Consider $\hbar^2$-deformed initial data $\hat \psi(\hbar^2,y)$ and $\hat y(\hbar^2,z)$ such that $\hat\psi(\hbar^2,0)=0$ and $\hat y(\hbar^2,0)=0$, let $\psi(y)\coloneqq\hat\psi(0,y)$ and $y(z)\coloneqq\hat y(0,z)$, and let
\begin{equation}
	\label{Curve}
	X(z)\coloneqq z\,e^{-\psi(y(z))},
\end{equation}
which can be considered as a formal local change of variables at $X=z=0$. With this change of variables we have
\begin{align}
	D&:=X\partial_{X}=\frac{1}{Q(z)}z\partial_{z}; \\  \label{eq:Qfunction}
	Q(z)&:=1-z\, y'(z)\,\psi'(y(z)).
\end{align}

From now on let $p_1,\ldots, p_N$ be the zeros of $dx=\frac{dX}{X} = (1-zy'(z)\psi'(y(z)))\frac{dz}z$. We assume that these zeros are simple in the present section.
\begin{lemma}\label{lem:dxzeros}
	Under the natural analytic assumptions of Definition~\ref{def:naa} the set of zeros of $Q(z)$ as a function on $\Sigma\setminus\{\infty\}\cong\mathbb{C}$ exactly coincides with the set of zeros of $dx=dX/X$.
\end{lemma}
\begin{proof}	
	Let us show that $\infty\in\Sigma$ is not a zero of $dx$. 
	Assume the contrary, then $dy=y'(z)dz$ is regular at $\infty$, and thus $y'(z) = O(\frac {1} {z^2})$ for $z\rightarrow\infty$. Then as $\frac{d\psi(y)}{dy}|_{y=y(z)}$ is regular at $\infty$, we have $z y'(z)\psi'(y(z)) = O(\frac{1}{z})$, therefore $Q(\infty)=1$ and  $dx=Q(z)\frac{dz}z$ has a pole at $\infty$; thus we have arrived at a contradiction.
	
	Since $y$ and $\psi$ are defined as regular series (at zero), the expression $(\psi(y(z)))'=y'(z)\,\psi'(y(z))$ is regular at $z=0$, and thus $Q(0)=1$, and therefore $dx=Q(z)\frac{dz}z$ has a simple pole at $z=0$.
	
	Finally, if $dx|_p=0$ and $p\notin \{0,\infty\}$ then $p$ is a zero of $Q(z)$ and vice versa since $dx=Q(z)\frac{dz}z$ and $\frac{dz}z$ is regular and nonvanishing at $p\notin \{0,\infty\}$.  
\end{proof}

Let $H_{g,n} = H_{g,n}(X_1,\dots,X_n)$ be the $n$-point functions of $Z_{\hat\psi,\hat y}$. Define $z_i$ by $X_i=X(z_i)$ and let $D_i=X_i\partial_{X_i}=Q(z_i)^{-1}z_i\partial_{z_i}$, $i=1,\dots,n$.
In~\cite{BDKS20} we obtained explicit closed algebraic formulas for $H_{g,n}$ that we recall in~Section~\ref{sec:proj}. In the current Section, with a goal to analyze the loop equations, we use the explicit closed algebraic formulas for
\begin{align}
	W_{g,n}&\coloneqq D_1\dots D_nH_{g,n}, & (g,n)\ne(0,2);\\
	W_{0,2}&\coloneqq D_1D_2H_{0,2}+\frac{X_1X_2}{(X_1-X_2)^2}. \label{eq:W02}
\end{align}
that are also derived in~\cite{BDKS20}.

\begin{remark} The main case analyzed in~\cite{BDKS20} is the case of undeformed $\psi$ and $y$, that is, the case of $\hat\psi=\psi$ and $\hat y = y$. Note however that all arguments in \emph{op.cit.} work without any change also for this more general case, cf.~\cite[Remark~5.6]{BDKS20}. That is, while the statement of Proposition~\ref{prop:FormulasWgn} below is not proved in \cite{BDKS20}, its proof is completely analogous to the proof of \cite[Theorem~4.8]{BDKS20}, so we do not repeat it here.
\end{remark}

Denote
\begin{equation}
	\cS(z)\coloneqq \frac{e^{z/2}-e^{-z/2}}{z}.
\end{equation}

\begin{proposition} \label{prop:FormulasWgn}
For $n\geq 2$, $(g,n)\neq(0,2)$ we have
\begin{equation}\label{eq:mainprop}
 W_{g,n}=[\hbar^{2g-2+n}]
	U_n\dots U_1
	\sum_{\gamma \in \Gamma_n}\prod_{\{v_k,v_\ell\}\in E_\gamma} w_{k,\ell},
\end{equation}
where the sum is over all connected simple graphs on $n$ labeled vertices,
\begin{equation}\label{eq:wkl}
	w_{k,\ell}=e^{\hbar^2u_ku_\ell\cS(u_k\hbar\,Q_k D_k)\cS(u_\ell\hbar\,Q_\ell D_\ell)
		\frac{z_k z_\ell}{(z_k-z_\ell)^2}}-1
\end{equation}
and $U_{i}$ is the operator acting on a function~$f$ in~$u_i$ and $z_i$ by
\begin{align}\label{eq:Uihbar}
	U_{i} f&=
	\sum_{j,r=0}^\infty D_i^j\left(\frac{L^j_{r,i}
	}{Q_i}
	[u_i^r] \frac{e^{u_i(\cS(u_i\,\hbar\,Q_i\,D_i)\hat y(z_i)-y(z_i))}}{u_i\hbar\,\cS(u_i\,\hbar)}f(u_i,z_i)\right),
\end{align}
where
\begin{align}
	\label{eq:Lr}
	L^j_{r,i} &:= \left.\left([v^j]e^{-v\,\psi(y)}\partial_y^r e^{v\frac{\cS(v\,\hbar\,\partial_y)}{\cS(\hbar\,\partial_y)}\hat \psi(y)}\right)\right|_{y=y(z_i)}\\ \nonumber
	&\phantom{:}=\left.\left([v^j]\left(\partial_y+v\psi'(y)\right)^r e^{v\left(\frac{\cS(v\hbar\partial_y)}{\cS(\hbar\partial_y)}\hat{\psi}(y)-\psi(y)\right)}\right)\right|_{y=y(z_i)}
\end{align}
and $Q_i=Q(z_i)$.
For $(g,n)=(0,2)$ we have
\begin{equation}
	W_{0,2} =\frac{1}{Q_1Q_2}\frac{z_1z_2}{(z_1-z_2)^2}.
\end{equation}
For $n=1$, $g>0$ we have 
\begin{equation}\label{eq:Wg1}
	W_{g,1} = [\hbar^{2g}]\left(\hbar \,U_1 1 + \sum_{j=0}^\infty D_1^{j}L_{0,1}^{j+1}\; D_1y(z_1)\right).
\end{equation}
Finally, for $(g,n)=(0,1)$ we have
\begin{equation}
	W_{0,1} = y(z_1).
\end{equation}
\end{proposition}

\begin{remark}
Note that the first exponential in \eqref{eq:Lr} does not depend on $\hbar$ (see \cite[Equations~(94) and~(95)]{BDKS20}).	
\end{remark}

We use the formulas given in Proposition~\ref{prop:FormulasWgn} throughout this Section in order to derive the loop equations.

\subsection{Linear and quadratic loop equations and the space \texorpdfstring{$\hat{\Xi}$}{Xihat}}\label{Sec:loopeq}

Assume that the series $z=z(X)$ has a positive radius of convergence. Let $\Sigma$ be the Riemann surface of $z=z(X)$.

Abusing notation, we denote by $Q(z)$ the analytic extension of the function $Q(z)$ defined in Equation~\eqref{eq:Qfunction}, and recall that it has a finite number of zeros that we denote by $p_1,\ldots,p_N$. Let $\sigma_i$ be the deck transformation of $X$ at $z\to p_i$ on $\Sigma$. Throughout this section we assume that all zeros of $Q$ (and, therefore, the critical points of the analytic extension of $X$ to a function on $\Sigma$) are simple.

We use the functions
\begin{equation}	
W_{g,n} = D_1\cdots D_n H_{g,n}+\delta_{g,0}\delta_{n,2}X_1X_2/(X_1-X_2)^2
\end{equation}
rather than $H_{g,n}$ or
\begin{equation}	
\omega_{g,n}=W_{g,n} \prod_{i=1}^n dX_i/X_i
\end{equation}
 in order to reformulate the loop equations. Also, abusing the terminology a little bit, we often say that the $n$-point functions $H_{g,n}$ satisfy the loop equations rather than $\omega_{g,n}$.

\begin{definition} \label{def:LLEs}
	We say that $n$-point functions $H_{g,n}$ satisfy the linear loop equations if for
	any $g,n\geq 0$ and for any $i = 1,\ldots,N$ the expression
	\begin{equation}
		\label{eq:oldlle}
		W_{g,n+1}(z,z_{\llbracket n \rrbracket})+W_{g,n+1}(\sigma_i(z),z_{\llbracket n \rrbracket})
	\end{equation}
	is holomorphic in $z$ for $z\to p_i$.
\end{definition}

\begin{definition}	\label{def:QLEs}
	We say that $n$-point functions $H_{g,n}$ satisfy the quadratic loop equations if for any $g,n\geq 0$ and for any $i = 1,\ldots,N$ the expression
	\begin{equation}
		\label{eq:oldqle}
		W_{g-1,n+2}(z,\sigma_i(z),z_{\llbracket n \rrbracket})+\sum\limits_{\substack{g_1+g_2=g\\ I\sqcup J=\llbracket n \rrbracket}} W_{g_1,|I|+1}(z_1,z_{I}) W_{g_2,|J|+1}(\sigma_i(z_1),z_{J})
	\end{equation}
	is holomorphic in $z$ for $z\to p_i$.
\end{definition}

The goal of this Section is to reformulate loop equations in a more convenient way. In particular, this will allow us to immediately prove the linear loop equations for $H_{g,n}$ whose $W_{g,n}$'s are described explicitly in the previous Section.

\subsubsection{The space $\hat\Xi$} The key role in the reformulation of the loop equations is played by the space $\hat\Xi$ that we define now.

\begin{definition}
	Let $\hat{\Xi}$ be a subspace in the space of functions defined in the vicinity of the points $p_1,\dots,p_N\in\Sigma$ that is spanned by
	\begin{equation}
		\label{def:hatxi}
		\Big(X \frac{d}{d X}\Big)^{j} f(z),\quad j\ge0,
	\end{equation}
	for all $f$ regular at $p_1,\dots,p_N$ (here $f$ is also defined only in the vicinity of the points $p_1,\dots,p_N\in\Sigma$ and we do not assume that functions forming $\hat{\Xi}$ extend to the whole $\Sigma$). 
\end{definition}

\begin{proposition}
	\label{Prop:localcoord}
	A function $g$ defined in the vicinity of the points $p_1,\dots,p_N\in\Sigma$ belongs to the space $\hat\Xi$ if and only if it satisfies the following condition for every $i=1,\dots,N$. Let $\nu$ be a local coordinate near $p_i$ such that $\nu^2=x-x(p_i)$. Then the condition is that the Laurent series expansion of $g$ near $p_i$ in the coordinate $\nu$ should have odd principal part.
\end{proposition}
\begin{proof} Let $x=\log(X)$ and recall that $x(\sigma_i(z))=x(z)$ for $z$ in the vicinity of $p_i$. Let $\nu_i$ be a local coordinate near $p_i$ such that $x-x(p_i)=\nu_i^2$. Note that $\sigma_i(\nu_i)=-\nu_i$ and locally near $p_i$ we have $X\partial_X = \partial_x = (2\nu_i)^{-1}\partial_{\nu_i}$.
	
In terms of these local coordinates the condition \eqref{def:hatxi} is reformulated as follows: the space $\hat{\Xi}$ is spanned by locally defined  functions of the form
	\begin{equation}
		\label{eq:nucoord}
		\left(\frac1{\nu_i}\frac d{d\nu_i}\right)^j f(\nu_i),\qquad j\geq 0
	\end{equation}
	where $f(\nu_i)$ is regular at $\nu_i=0$, $i=1,\dots,N$.
	
	Note that the operator $\nu_i^{-1}\partial_{\nu_i}$ preserves the space of Laurent series in $\nu_i$ with odd principal parts, and applied to the function $\nu_i$ this operator generates a basis in space of odd principal parts. This implies the statement of the lemma.
\end{proof}

\begin{notation} Let $f=f(z_1,\dots,z_n)$ be a function on $\Sigma$. We say that $f\in \hat\Xi(z_1)$ if the restriction of $f$ considered as a function of $z_1$ to the vicinity of the points $p_1,\dots,p_N\in \Sigma$ belongs to $\hat\Xi$.
\end{notation}

\subsubsection{Linear loop equations} Now we can reformulate the linear loop equations and immediately prove them for the $n$-point functions whose $W_{g,n}$'s are given in Proposition~\ref{prop:FormulasWgn}.

\begin{proposition}\label{prop:lle} The $n$-point functions $H_{g,n}$ satisfy the linear loop equations if and only if  $W_{g,n+1}(z,z_{\llbracket n \rrbracket}) \in \hat{\Xi}(z)$ for any $g,n\geq 0$.
\end{proposition}
\begin{proof}
	Recall expression~\eqref{eq:oldlle}. It is holomorphic in $z$ near the ramification point $p_i$ if and only if its Laurent series expansion in $\nu = \sqrt{x(z)-x(p_i)}$ has odd principal part (recall that all $W_{g,n+1}$ are meromorphic functions in the vicinity of $p_1,\dots,p_N$). Therefore, by Proposition~\ref{Prop:localcoord}, the linear loop equations are satisfied if and only if $W_{g,n+1}(z,z_{\llbracket n \rrbracket}) \in \hat{\Xi}(z)$ for any $g,n\geq 0$.
\end{proof}

\begin{theorem}\label{th:linloop} Consider formal power series $\hat \psi(\hbar^2,y)$ and $\hat y(\hbar^2,z)$ defined by \eqref{def:psihat} and \eqref{def:yhat}. Under the natural analytic assumptions (Definition~\ref{def:naa}) and the additional condition that the zeros of $Q(z)$ are simple, the $n$-point functions of $Z_{\hat\psi,\hat y}$ satisfy the linear loop equations.

\end{theorem}
\begin{proof}
	Note that natural analytic assumptions mean that $W_{g,n}$ uniquely extend as functions on $\Sigma^n$.
	By Proposition~\ref{prop:lle} it is sufficient to show that $W_{g,n}\in \hat{\Xi}(z_1)$ for any $g\geq 0$, $n\geq 1$. To this end we use the formulas stated in Proposition~\ref{prop:FormulasWgn}. In particular, in the case $n\geq 2$ we use Equation~\eqref{eq:mainprop} (the case $n=1$ is completely analogous).
	
	Following the unfolding of Equation~\eqref{eq:mainprop} given in Equations~\eqref{eq:wkl} and~\eqref{eq:Uihbar} we see that $W_{g,n}$ is equal to a finite sum of the expressions $D_1^j Q_1^{-1} f_j$, where $f_j$ is regular at $z_1\to p_i$, $i=1,\dots, N$ (recall that $p_1,\dots,p_N$ are zeros of $Q(z)$). We can say that these $f_j$ are regular due to the natural analytic assumptions.
	
	Note that for any function $f_j(z_1)$ regular at $p_1$ we have $Q_1^{-1} f_j\in \hat\Xi(z_1)$, due to the fact that $Q_1^{-1}\in\hat\Xi(z_1)$ and moreover $Q_1^{-1}$ has a simple pole at $p_1$ (since we have assumed that the zeros of $Q(z)$ are simple). Note also that $D \Xi \subseteq \Xi$. Therefore, a finite sum of the expressions $D_1^j Q_1^{-1} f_j$ does belong to $\hat\Xi(z_1)$.
\end{proof}


\subsubsection{Quadratic loop equations} In this Section we reformulate the quadratic loop equations in terms of the space $\hat\Xi$, under the assumption that the linear loop equations hold. For convenience, let $p=p_i$, $i=1,\dots,N$, denote an arbitrary critical point, and let $\sigma=\sigma_i$ be the corresponding deck transformation.

Set
\begin{align}\label{eq:cW2}
	\cW_{g,n}^{(2)} &\coloneqq
	W_{g-1,n+1}(z_1,z_1,z_{\llbracket n \rrbracket \setminus 1})+
	\sum_{\substack{g_1+g_2=g\\I\sqcup J=\llbracket n \rrbracket \setminus 1}}
	W_{g_1,|I|+1}(z_1,z_{I})
	W_{g_2,|J|+1}(z_1,z_{J});
\\
\label{eq:tildecW2}
	\widetilde{\cW}_{g,n}^{(2)}(z,w) & \coloneqq
	W_{g-1,n+1}(z,w,z_{\llbracket n \rrbracket \setminus 1})+
	\sum_{\substack{g_1+g_2=g\\I\sqcup J={\llbracket n \rrbracket \setminus 1}}}
	W_{g_1,|I|+1}(z,z_{I})
	W_{g_2,|J|+1}(w,z_{J}).
\end{align}

\begin{remark}\label{rem:Special02}
Here and below we always assume that whenever we have to substitute the same variable twice into $W_{0,2}$ we use $D_1D_2H_{0,2}$ instead (cf. Equation~\eqref{eq:W02}).	
\end{remark}

Define the symmetrizing operator $S_z$ acting on functions in the vicinity of $p\in \Sigma$ as
\begin{equation}
	S_z(f(z))\coloneqq f(z)+f(\sigma(z)).
\end{equation}
Note that
\begin{equation}
	\label{eq:SSW}
	S_{z}S_{w} \widetilde{\cW}_{g,n}^{(2)}(z,w) \bigm|_{z=w=z_1} = S_{z_1}(\cW_{g,n}^{(2)}) + 2\widetilde{\cW}_{g,n}^{(2)}(z_1,\sigma(z_1)).
\end{equation}
Note that in the case $(g,n)=(0,2)$ we have to use the convention stated in Remark~\ref{rem:Special02} on both side of this equation. To efficiently use this equation, we will have to additionally assume that the functions $W_{g,n}(z_1,\dots,z_n)$ have no poles at the diagonals $z_i=z_j$ and at the antidiagonals $z_i=\sigma(z_j)$ in the vicinity of $p$, for all critical points $p$, with an obvious exception of the case $(g,n)=(0,2)$, which is covered by Remark~\ref{rem:Special02}.

\begin{proposition} \label{Prop:newqle} Assume $W_{g,n}(z_1,\dots,z_n)$ have no poles at the diagonals $z_i=z_j$ and at the antidiagonals $z_i=\sigma(z_j)$ in the vicinity of $p$ for all critical points $p$ (with an exception of $W_{0,2}$ that does have a double pole at the diagonal with bi-residue $1$). Assume also that the $n$-point functions $H_{g,n}$ satisfy the linear loop equations.
	
	Under these assumptions,
	the $n$-point functions $H_{g,n}$ satisfy the quadratic loop equations if and only if $\cW_{g,n}^{(2)}\in \hat{\Xi}(z_1)$ for any $g\geq 0$, $n\geq 1$.
\end{proposition}
\begin{proof}
	The assumptions imply that $S_{z}S_{w} \widetilde{\cW}_{g,n}^{(2)}(z,w)$ is holomorphic in $z$ and in $w$ near $p$, and this expression remain holomorphic in $z_1$ after the restriction $z=w=z_1$.
	
	The quadratic loop equations~ \eqref{eq:oldqle} are equivalent to the holomorphicity at $z_1\to p$ of the second term on the right hand side of Equation~\eqref{eq:SSW}, for all $g\geq 0$, $n\geq 1$. Therefore, since the left hand side of Equation~\eqref{eq:SSW} is holomorphic at $z_1\to p$, the quadratic loop equations are equivalent to the holomorphicity of $S_{z_1}(\cW_{g,n}^{(2)})$ at $z_1\to p$, for all $g\geq 0$, $n\geq 1$. Repeating \emph{mutatis mutandis} the arguments in the proof of Proposition \ref{prop:lle} we see that the latter property is equivalent to $\cW_{g,n}^{(2)}\in \hat{\Xi}(z_1)$ for any $g\geq 0$, $n\geq 1$.
\end{proof}

\begin{lemma} \label{lem:ReductionQLEassumptions} 	
	Under the natural analytic assumptions of Definition~\ref{def:naa} together with the additional condition that the zeros of $Q(z)$ are simple, the $n$-point functions of $Z_{\hat\psi,\hat y}$ satisfy the assumptions of Proposition~\ref{Prop:newqle}.
\end{lemma}

\begin{proof} The only thing that we have to check is that $W_{g,n}$'s have no poles on the diagonals or anti-diagonals in the vicinity of zeros of $Q$, $(g,n)\not=(0,2)$. For the anti-diagonals this follows manifestly from the formulas in Proposition~\ref{prop:FormulasWgn}. For the diagonals it follows from the analysis in~\cite{BDKS20}, see~\cite[Remark 1.3 and Section 4]{BDKS20}.
\end{proof}

\begin{theorem} \label{thm:QLEs} Under the natural analytic assumptions of Definition~\ref{def:naa} together with the additional condition that the zeros of $Q(z)$ are simple, the $n$-point functions of $Z_{\hat\psi,\hat y}$ satisfy the quadratic loop equations.
\end{theorem}

\begin{proof} 
	It is a corollary of Lemma~\ref{lem:ReductionQLEassumptions}, Proposition~\ref{Prop:newqle} together with Lemma~\ref{Prop:qle} and Theorem~\ref{theor:higher} which we formulate and prove below (the latter Lemma and Theorem are related to the so-called \emph{higher loop equations} which we discuss below; we get the statement on the quadratic loop equations as a corollary of a more general statement regarding these higher loop equations).
\end{proof}

\subsubsection{Blobbed topological recursion} An immediate reformulation of Theorems~\ref{th:linloop} and~\ref{thm:QLEs} is the following statement:

\begin{theorem} \label{thm:Blobbed}
	Consider the $n$-point functions $H_{g,n}$ of $Z_{\hat\psi,\hat y}$. Under the natural analytic assumptions the corresponding symmetric differentials
	\begin{equation}
		\omega_{g,n} = d_1\cdots d_n H_{g,n}++\delta_{g,0}\delta_{n,2}\frac{dX_1dX_2}{(X_1-X_2)^2}
	\end{equation} satisfy the blobbed topological recursion.
\end{theorem}

\subsection{Higher loop equations}
\label{Sec:higherloop}

The goal of this section is to define certain combinations of functions $W_{g,n}$ that we denote by $\cW^{(r)}_{g,n}$, both abstractly and in terms of the vacuum expectation values, and use these new functions to state a version of the higher loop equations.

\subsubsection{First definition and the main statement}
 	Let
 \begin{equation}\label{eq:Wndef}
 	 W_n\coloneqq \sum_{g=0}^\infty \hbar^{2g-2+n} W_{g,n}.
 \end{equation}
 \begin{definition}	\label{def:defT}
 Define
 \begin{align}	\label{def:T-1}
 	T_{n}(z_1;z_{\llbracket n \rrbracket \setminus 1};u)
 	\coloneqq \sum _{k=1}^{\infty } \frac{1}{k!}
 	\left(\prod _{i=1}^k \bigm\rfloor_{z_{\overline{i}}=z_1}u\,\hbar\, \cS(u \hbar  D_{\overline{i}})\right)
 	W_{k+n-1}(z_{\{\overline 1, \overline 2,\dots,\overline k\}},z_{\llbracket n \rrbracket \setminus 1}),
 \end{align}
 and
 \begin{align}	\label{def:cT-1}
 	\cT_n(z_1;z_{\llbracket n \rrbracket \setminus 1};u):=
 	\frac{\cS(u \hbar D_1)}{\hbar \cS (u \hbar )}
 	\sum_{l=1}^\infty\frac1{l!}
 	\sum_{\bigsqcup_{i=1}^l J_{i}=\llbracket n \rrbracket \setminus 1}
 	\prod _{i=1}^l T_{|J_{i}| +1}(z_1;z_{J_{i}}),
 \end{align}
The inner sum in \eqref{def:cT-1} goes over the set of all possible ordered partitions of the set $\llbracket n \rrbracket \setminus 1=\{2,\ldots ,n\}$ as a disjoint union of $k$ subsets $J_i$ that are allowed to be empty.
These functions involve an additional parameter $u$. Recall that we always apply the convention to remove the singularity in $W_{0,2}$ factors once we have to restrict them to the diagonal, as discussed in Remark~\ref{rem:Special02}.
\end{definition}

 \begin{definition}\label{def:cW(r)FirstDefinition} By $\cW^{(r)}_{g,n}$ we denote the coefficient of $[\hbar^{2g-2+n}u^r]$ in $r! \cT_n(z_1;z_{\llbracket n \rrbracket \setminus 1};u)$.
 \end{definition}

Let us go through a first few examples. Obviously, $\cW_{g,n}^{(0)}=0$ and $\cW_{g,n}^{(1)}=W_{g,n}$. For $r=2$ by a straightforward expansion of the definition we have:

 \begin{lemma}
 	\label{Prop:qle}
 	Definition~\ref{def:cW(r)FirstDefinition} for $\cW_{g,n}^{(2)}$ coincides with the one given in~\eqref{eq:cW2}:
 	\begin{equation}
 		\cW_{g,n}^{(2)}=
 		W_{g-1,n+1}(z_1,z_1,z_{\llbracket n \rrbracket \setminus 1})+
 		\sum_{\substack{g_1+g_2=g\\I\sqcup J=\llbracket n \rrbracket \setminus 1}}
 		W_{g_1,|I|+1}(z_1,z_{I})
 		W_{g_2,|J|+1}(z_1,z_{J}).
 	\end{equation}
 \end{lemma}
One more example is for $r=3$:
 	\begin{align}\label{eq:cW3}
 		\cW_{g,n}^{(3)}=& W_{g-2,n+2}(z_1,z_1,z_1,z_{\llbracket n \rrbracket \setminus 1})\\ \notag
 		& +3\sum_{\substack{g_1+g_2=g-1\\J_1\sqcup J_2=\llbracket n \rrbracket \setminus 1}}
 		W_{g_1,| J_1| +1}(z_1,z_{J_1}) W_{g_2,| J_2| +2}(z_1,z_1,z_{J_2})
 		\\ \notag &
 		+\sum _{\substack{g_1+g_2+g_3=g\\J_1\sqcup J_2\sqcup J_3=\llbracket n \rrbracket \setminus 1}}
 		W_{g_1,| J_1| +1}(z_1,z_{J_1})  W_{g_2,| J_2| +1} (z_1,z_{J_3})
 		W_{g_3,| J_3| +1}(z_1,z_{J_3})\\ \notag &
 		+\frac14\Big(2D_1^2-1\Big)W_{g-1,n} (z_1,z_{\llbracket n \rrbracket \setminus 1}).
 	\end{align}

\begin{definition}
	We say that the symmetric differentials $\omega_{g,n} = W_{g,n}\prod_{i=1}^n dX_i/X_i$ or, abusing the terminology, that the $n$-point functions $H_{g,n}$, related to $W_{g,n}$'s by $W_{g,n}=D_1\cdots D_n H_{g,n}+\delta_{g,0}\delta_{n,2}X_1X_2/(X_1-X_2)^2$, satisfy the higher loop equations if $\cW^{(r)}_{g,n}\in \hat{\Xi}$ for any $g\geq 0$, $n,r\geq 1$.
\end{definition}

In particular, for $r=1$ we obtain the definition of the linear loop equations, see Definitions~\ref{def:LLEs}. For $r=2$ we obtain the property that is according to Proposition~\ref{Prop:newqle} is equivalent to the definition of the quadratic loops equations given in Definition~\ref{def:QLEs}.

 \begin{theorem}  	\label{theor:higher}
 	Under the natural analytic assumptions on $\hat\psi,\,\hat y$, together with the additional condition that the zeros of $Q(z)$ are simple,  the $n$-point functions of $Z_{\hat\psi,\hat y}$ satisfy the higher loop equations, that is, $\cW^{(r)}_{g,n}\in \hat{\Xi}$ for all $g\geq 0$, $n\geq 1$.
 \end{theorem}

This theorem immediately implies Theorem~\ref{thm:QLEs} and, as a consequence, Theorem~\ref{thm:Blobbed}. The proof of Theorem~\ref{theor:higher} is given in Section~\ref{Sec54}.

\subsection{An alternative formula for \texorpdfstring{$\cW^{(r)}_{g,n}$}{Wrgn}}

In order to prove Theorem~\ref{theor:higher} we represent the quantities $\cW^{(r)}_{g,n}$ as vacuum expectation values in the semi-infinite wedge formalism. This and the subsequent sections are heavily based on~\cite{BDKS20}.

\subsubsection{Summary for the operators acting in the Fock space} \label{sec:ReminderOnFock}
Recall some notations related to the action of the Lie algebra $A_\infty$ on the bosonic Fock space (see \cite{BDKS20} for the details) which will be needed below. Denote
\begin{align}
	\cE_0(u)&=\sum_{k\in\Z+\frac12}e^{u\,k}\hat E_{k,k},\\
	\cE_m(u)&=\sum_{k\in\Z+\frac12}e^{u\,(k-\frac{m}{2})}\hat E_{k-m,k},\\
	\cE(u,X)&=\sum_{m\in\Z}X^m\cE_m=\sum_{m\in\Z}\sum_{k\in\Z+\frac12}X^me^{u\,(k-\frac{m}{2})}\hat E_{k-m,k},\\
	J_m& = \sum_{k\in\Z+\frac12}\hat E_{k-m,k}.
\end{align}
Then we have the following basic commutation relations:
\begin{align}
	[J_m,\cE_0(u)]&=u\,m\,\cS(u\,m)\cE_m(u);\\
	\label{eq:JE0omm}
	[\sum_{m=-\infty}^\infty \frac{X^m}{m} J_m,\cE_0(u)]&=\sum_{m=-\infty}^\infty u\,\cS(u\,m)\cE_m(u)=u\,\cS(u\,X\partial_X)\cE(u,X).
\end{align}
 We also have the following identity expressing bosonic realization for the action of $A_\infty$:
\begin{equation}\label{eq:cEdef}
	\cE(u,X)
	=\frac{
		e^{\sum_{i=1}^\infty u\cS(ui) J_{-i}X^{-i}}
		e^{\sum_{i=1}^\infty u\cS(ui) J_{i}X^{i}} -1 }
	{u\cS(u)}.
\end{equation}

We also need the following operator, which is an element of the corresponding Lie group:
\begin{equation}
	\cD(\hbar)
	=\exp\left(\sum_{k\in\Z+\frac12}w(\hbar\,k)\,\hat E_{k,k}\right),
\end{equation}
where the series $w(y)$ is related to $\hat\psi(y)$ by
\begin{equation}
	\hat\psi(y)=w(y+\hbar/2)-w(y-\hbar/2)=\zeta(\hbar\partial_y)w(y),
\end{equation}
where
\begin{equation}
	\zeta(z) = z\cS(z)=e^{z/2}-e^{-z/2}.
\end{equation}
\begin{remark}
	This equality defines $w(y)$ up to a constant. A choice of this constant is not important since scalar matrices from $\cA_\infty$ act trivially in the Fock space.
\end{remark}

Define
\begin{equation}\label{eq:eurom}
	\euro_m\coloneqq \cD(\hbar)^{-1}J_m\cD(\hbar).
\end{equation}
It is proved in~\cite[Proposition 3.1]{BDKS20} that
\begin{equation}
	\euro_m = \sum_{r=0}^\infty\partial_y^r\phi_m(y)|_{y=0}[u^r z^m]\frac{
		e^{\sum_{i=1}^\infty u\hbar\cS(u\hbar i) J_{-i}z^{-i}}
		e^{\sum_{i=1}^\infty u\hbar\cS(u\hbar i) J_{i}z^{i}}}
	{u\hbar \cS(u\hbar)},
\end{equation}
where
\begin{equation}\label{eq:phim}
	\phi_m(y)=e^{w(y+\frac{\hbar m}{2})-w(y-\frac{\hbar m}{2})}=e^{\zeta(m \hbar \partial_y)w(y)}=
	e^{\frac{\zeta(m\hbar\partial_y)}{\zeta(\hbar\partial_y)}\hat \psi(y)}=e^{m\hat{\psi}(y)}L_0(m,y,\hbar).
\end{equation}

Let
\begin{align}
J(X)&\coloneqq \sum_{m=-\infty}^\infty X^mJ_m,\\
\euro(X)&\coloneqq\sum_{m=-\infty}^\infty X^m\euro_m
\end{align}
and recall that
$\displaystyle W_{n} \mathop{=}^{\eqref{eq:Wndef}} \sum_{g=0}^\infty h^{2g-2+n}W_{g,n}$.

From \cite[Section~4]{BDKS20} we have the following vacuum expectation value expression for the \emph{disconnected} function $W^\bullet_n$: 
\begin{equation}
\label{eq:correlator}
		W^\bullet_n
		=\sum\limits_{m_1,\ldots,m_n=-\infty}^\infty X_1^{m_1}\ldots X_n^{m_n}\VEV{\left(\prod_{i=1}^nJ_{m_i}\right)\mathcal{D}(\hbar)e^{\sum_{i=1}^\infty \frac{\hat{q}_i J_{-i}}{i\hbar} }},
\end{equation}		
where $W^\bullet_n$ is related to $W_n$ of \eqref{eq:Wndef} via the inclusion-exclusion formula:
\begin{equation}
	W^\bullet_n = \sum\limits_{l=1}^n \frac{1}{l!}\sum_{\substack{I_1\sqcup\ldots\sqcup I_l=\llbracket n \rrbracket \\ \forall j\, I_j \neq \emptyset}}\prod_{i=1}^l  W_{|I_i|}(X_{I_i}).
\end{equation}
Here the inner sum goes over ordered collections of non-empty non-intersecting sets.

\subsubsection{Higher loop equations via vacuum expectation values}
\label{SubSec:higherloop}

In this section we propose a modification of formula \eqref{eq:correlator} that generates the expressions $\cW_{g,n}^{(r)}$.

\begin{definition}
	\label{def1}
Define
	\begin{align} \label{eq:VEVdefinitionOfcWgn}
		\cW_{n}^\bullet(z_1;z_{\llbracket n \rrbracket \setminus 1};u)&\coloneqq
		\VEV{ \left(\sum_{m=-\infty}^\infty \frac{X_1^m J_m}{m\hbar}\right) \cE_0(\hbar\,u)
			\left(\prod_{i=2}^nJ(X_i)\right)
			\cD(\hbar)e^{\sum_{j=1}^\infty \frac{\hat{q}_j J_{-j}}{j\hbar} }},\\
		\cW_{g,n}^\bullet&\coloneqq [\hbar^{2g-2+n}]\cW_{n}^\bullet.
	\end{align}
\end{definition}

	We define the \emph{connected} functions $\cW_{n}(z_1;z_{\llbracket n \rrbracket \setminus 1};u)$ and $\cW_{g,n}(z_1;z_{\llbracket n \rrbracket \setminus 1};u)$ as follows.
	Let
	\begin{equation}
		\mathcal{A}_k:=
		\begin{cases}
			\left(\cD(\hbar)e^{\sum_{j=1}^\infty \frac{\hat{q}_j J_{-j}}{j\hbar} }\right)^{-1}\left(\sum_{m=-\infty}^\infty \frac{X_1^m J_m}{m\hbar}\right) \cE_0(\hbar\,u)\left(\cD(\hbar)e^{\sum_{j=1}^\infty \frac{\hat{q}_j J_{-j}}{j\hbar} }\right),& \text{if $k=1$};\\
			\left(\cD(\hbar)e^{\sum_{j=1}^\infty \frac{\hat{q}_j J_{-j}}{j\hbar} }\right)^{-1}J(X_k)\left(\cD(\hbar)e^{\sum_{j=1}^\infty \frac{\hat{q}_j J_{-j}}{j\hbar} }\right),& \text{if $k\geq 2$}.
		\end{cases}
	\end{equation}
	Then $\cW_n^\bullet(z_1;z_{\llbracket n \rrbracket \setminus 1};u)= \VEV{\mathcal{A}_1\ldots\mathcal{A}_n}$. We can apply the usual inclusion-exclusion procedure to the $\mathcal{A}$ operators. Namely, let
	\begin{align}\label{eq:inclexclinv}
		\VEV{\mathcal{A}_1\ldots\mathcal{A}_n}^\circ& =\sum\limits_{l=1}^n\frac{1}{l!}\sum_{\substack{I_1\sqcup\ldots\sqcup I_l=\llbracket n \rrbracket \\ \forall j\, I_j \neq \emptyset}}(-1)^{l-1}(l-1)! \prod_{i=1}^l  \VEV{\mathcal{A}_{I_i}};
		\\
		\label{eq:inclexcl}
		\VEV{\mathcal{A}_1\ldots\mathcal{A}_n}& =\sum\limits_{l=1}^n\frac{1}{l!}\sum_{\substack{I_1\sqcup\ldots\sqcup I_l=\llbracket n \rrbracket \\ \forall j\, I_j \neq \emptyset}}\prod_{i=1}^l  \VEV{\mathcal{A}_{I_i}}^\circ,
	\end{align}
	where the inner sums go over ordered collections of non-empty non-intersecting sets, and by $\mathcal{A}_{I}$ for $I=\{i_1<\cdots<i_p\}\subset \llbracket n \rrbracket$ we denote $\mathcal{A}_{i_1}\cdots \mathcal{A}_{i_p}$.
	
		\begin{definition} Let
	\begin{align}
		\cW_n(z_1;z_{\llbracket n \rrbracket \setminus 1};u)& \coloneqq  \VEVc{\mathcal{A}_1\ldots\mathcal{A}_n}, \\
			\cW_{g,n}& \coloneqq [\hbar^{2g-2+n}] \cW_n.
	\end{align}
\end{definition}

Now let us introduce a new definition of $\cW^{(r)}_{g,n}$ (previously defined in Definition~\ref{def:cW(r)FirstDefinition}) and immediately prove that the two definitions are equivalent.

\begin{definition}\label{def:cW(r)VEV} By $\cW^{(r)}_{g,n}$ we denote the coefficient of $[u^r]$ in $r!\cW_{g,n}(z_1;z_{\llbracket n \rrbracket \setminus 1};u)$.
\end{definition}

\begin{theorem} For any $\hat \psi(\hbar^2,y)$ and $\hat y(\hbar^2,z)$ such that $\hat\psi(\hbar^2,0)=0$ and $\hat y(\hbar^2,0)=0$ the two definitions of $\cW^{(r)}_{g,n}$, Definition~\ref{def:cW(r)FirstDefinition} and Definition~\ref{def:cW(r)VEV}, are equivalent.
	
In other words, in this case $\cW_{n}(z_1;z_{\llbracket n \rrbracket \setminus 1};u) = \cT_n(z_1;z_{\llbracket n \rrbracket \setminus 1};u)$.
\end{theorem}



\begin{proof} Apply the commutation relation~\eqref{eq:JE0omm} to the definition~\eqref{eq:VEVdefinitionOfcWgn} of $\cW_{n}^\bullet$. Using the explicit formula for $\cE(u,X)$ given in~\eqref{eq:cEdef} and the fact that the operators $J_i$, $i<0$, annihilate the covacuum, we obtain:
	
	\begin{align}
		\label{eq:proof1}
			\cW_{n}^\bullet&=
			u\cS(u\hbar D_1)\VEV{\cE(u\hbar, X_1)
				\left(\prod_{i=2}^nJ(X_i)\right)
				\cD(\hbar)e^{\sum_{j=1}^\infty \frac{\hat{q}_j J_{-j}}{j\hbar} }}
			\\ \notag &
			=
			\frac{\cS(u\hbar D_1)}{\hbar\cS(u\hbar)}\VEV{\left(e^{\sum_{j=1}^\infty u\hbar\cS(u\hbar j) J_{j}X_1^{j}} -1\right)
				\left(\prod_{i=2}^nJ(X_i)\right)
				\cD(\hbar)e^{\sum_{j=1}^\infty \frac{\hat{q}_j J_{-j}}{j\hbar} }}
			\\ \notag &
			=
			\frac{\cS(u\hbar D_1)}{\hbar\cS(u\hbar)}\sum _{k=1}^{\infty } \frac{1}{k!}
			\left(\prod _{i=1}^k \bigm\rfloor_{z_{\overline{i}}=z_1}\hbar u \cS(u \hbar  D_{\overline{i}})\right)
			\\ \notag
			& \hspace{4cm}
			\VEV{\left(\prod _{i=1}^kJ^+(X_{\bar i})\right)
				\left(\prod_{i=2}^nJ(X_i)\right)
				\cD(\hbar)e^{\sum_{j=1}^\infty \frac{\hat{q}_j J_{-j}}{j\hbar} }},
	\end{align}
	where
	\begin{equation}
	J^+(X_i)\coloneqq\sum_{m=1}^\infty X_i^m J_{m}.
	\end{equation}
	
	Note that the inclusion-exclusion formula can also be applied to correlators
	\begin{equation}
	\VEV{\left(\prod _{i=1}^kJ^+(X_{\bar i})\right)
		\left(\prod_{i=2}^nJ(X_i)\right)
		\cD(\hbar)e^{\sum_{j=1}^\infty \frac{\hat{q}_j J_{-j}}{j\hbar} }}
	\end{equation}
		to get connected versions of said correlators. Indeed, note that if one defines	
	\begin{equation}\label{eq:Bopdef}
		\mathcal{B}^{(k,n)}_m:=
		\begin{cases}
			\left(\cD(\hbar)e^{\sum_{j=1}^\infty \frac{\hat{q}_j J_{-j}}{j\hbar} }\right)^{-1}J^+(X_{\bar m})\left(\cD(\hbar)e^{\sum_{j=1}^\infty \frac{\hat{q}_j J_{-j}}{j\hbar} }\right),& \text{if $1\leq m\leq k$};\\
			\left(\cD(\hbar)e^{\sum_{j=1}^\infty \frac{\hat{q}_j J_{-j}}{j\hbar} }\right)^{-1}J(X_{m-k+1})\left(\cD(\hbar)e^{\sum_{j=1}^\infty \frac{\hat{q}_j J_{-j}}{j\hbar} }\right),& \text{if $k+1\leq m \leq k+n-1$},
		\end{cases}
	\end{equation}
then
\begin{equation}
	\VEV{\left(\prod _{i=1}^kJ^+(X_{\bar i})\right)
		\left(\prod_{i=2}^nJ(X_i)\right)
		\cD(\hbar)e^{\sum_{j=1}^\infty \frac{\hat{q}_j J_{-j}}{j\hbar} }} = \VEV{\mathcal{B}^{(k,n)}_1 \cdots \mathcal{B}^{(k,n)}_{k+n-1}}.
\end{equation}
We can then define the connected correlators $\VEV{\mathcal{B}^{(k,n)}_{i_1}\cdots\mathcal{B}^{(k,n)}_{i_m}}^\circ$, for $1\leq i_1 < \dots < i_m \leq k+n-1$, 
via a formula completely analogous to \eqref{eq:inclexclinv}, so that a formula analogous to \eqref{eq:inclexcl} holds:
\begin{equation}	\label{eq:Binclexcl}
	\VEV{\mathcal{B}^{(k,n)}_1\cdots\mathcal{B}^{(k,n)}_n} =\sum\limits_{l=1}^n\frac{1}{l!}\sum_{\substack{I_1\sqcup\ldots\sqcup I_l=\llbracket k+n-1 \rrbracket \\ \forall j\, I_j \neq \emptyset}}\prod_{i=1}^l  \VEV{\mathcal{B}^{(k,n)}_{I_i}}^\circ.
\end{equation}

Substituting \eqref{eq:Binclexcl} and \eqref{eq:Bopdef} into \eqref{eq:proof1} and splitting each index set $I_m$ as $I_m=V_m\sqcup U_m$ for $V_m\subseteq \{\bar 1,\dots,\bar k\}$, $U_m\subseteq \llbracket n \rrbracket \setminus 1$, we get	
	\begin{align}\label{eq:cWdisc1}		
		\cW_{n}^\bullet& =
		\frac{\cS(u\hbar D_1)}{\hbar\cS(u\hbar)}\sum _{k=1}^{\infty } \frac{1}{k!}
		\left(\prod _{i=1}^k \bigm\rfloor_{z_{\overline{i}}=z_1}\hbar u \cS(u \hbar  D_{\overline{i}})\right) \sum_{l=1}^\infty\frac1{l!}
		\sum_{\substack{\bigsqcup_{m=1}^l U_{m}=\llbracket n \rrbracket \setminus 1 \\ \bigsqcup_{m=1}^l V_{m}=\{\bar 1,\dots,\bar k\} \\ \forall m\, U_m \cup V_m \neq \emptyset}}
		\\ \notag
		& \hspace{2cm}
		\prod_{m=1}^l \VEVc{\left(\prod _{\bar i\in V_m} J^+(X_{\bar i})\right)
			\left(\prod _{i\in U_m} J(X_i)\right)
			\cD(\hbar)e^{\sum_{j=1}^\infty \frac{\hat{q}_j J_{-j}}{j\hbar} }};
	\end{align}
	 the innermost sum goes over ordered collections of sets $U_i$, $V_i$, where each of these sets is allowed to be empty, but for each $i$ the sets $U_i$ and $V_i$ cannot be empty simultaneously. 
	
	Now let us apply formula \eqref{eq:inclexclinv} to the whole expression \eqref{eq:cWdisc1}, obtaining $\cW_{n}$ from $\cW_{k}^\bullet$, $k=1\dots n$. Note that since expanding $\cW_n^\bullet$ via formula \eqref{eq:inclexcl} produces an expression where in each summand exactly one correlator contains $\cA_1$, and since \eqref{eq:inclexclinv} is the inverse formula to formula \eqref{eq:inclexcl},
	this operation (passing from $\cW_n^\bullet$ to $\cW_n$) corresponds precisely to taking the formula \eqref{eq:cWdisc1} for $\cW_{n}^\bullet$ and then dropping all terms with one of the $V_m$'s being empty in the innermost sum. We obtain:
	\begin{align}
	\label{eq:proof2}
	\cW_{n} & =
	\frac{\cS(u\hbar D_1)}{\hbar\cS(u\hbar)}\sum _{k=1}^{\infty } \frac{1}{k!}
	\left(\prod _{i=1}^k \bigm\rfloor_{z_{\overline{i}}=z_1}\hbar u \cS(u \hbar  D_{\overline{i}})\right) \sum_{l=1}^\infty\frac1{l!}
	\sum_{\substack{\bigsqcup_{m=1}^l U_{m}=\llbracket n \rrbracket \setminus 1 \\ \bigsqcup_{m=1}^l V_{m}=\{\bar 1,\dots,\bar k\} \\ \forall m\, V_m \neq \emptyset}}
	\\ \notag
	& \hspace{2cm}
	\prod_{m=1}^l \VEVc{\left(\prod _{\bar i\in V_m} J^+X(_{\bar i})\right)
		\left(\prod _{i\in U_m} J(X_i)\right)
		\cD(\hbar)e^{\sum_{j=1}^\infty \frac{\hat{q}_j J_{-j}}{j\hbar} }}
	\\ \notag
	& =
	\frac{\cS(u\hbar D_1)}{\hbar\cS(u\hbar)} \sum_{l=1}^\infty\frac1{l!}
	\sum_{\substack{\bigsqcup_{i=1}^l U_{i}=\llbracket n \rrbracket \setminus 1 }} 	\prod_{m=1}^l \sum _{k_m=1}^{\infty } \frac{1}{k_m!}
	\left(\prod _{i=1}^{k_m} \bigm\rfloor_{z_{\overline{i}}=z_1}\hbar u \cS(u \hbar  D_{\overline{i}})\right)
	\\ \notag
	& \hspace{2cm}
\VEVc{\left(\prod _{i=1}^{k_m} J^+(X_{\bar i})\right)
		\left(\prod _{i\in U_m} J(X_i)\right)
		\cD(\hbar)e^{\sum_{j=1}^\infty \frac{\hat{q}_j J_{-j}}{j\hbar} }}
\end{align}	
The vacuum expectation value in the last line is equal to $W_{k_m+|U_i|}(z_{\bar 1},\dots,z_{\bar k_m},z_{U_i})$ up to one small adjustment. Recall~\cite[Sections 4 and 6]{BDKS20} that the singular term in $W_2$ is coming from the commutator $[J^+(X_{\bar i}),J^-(X_j)]=X_{\bar i}X_j/(X_{\bar i}-X_j)^2$. Note also that the terms $J^-(X_{\bar j})$, $\bar j=1,\dots,k_m$, are absent in the formula. So, the vacuum expectation in the last line of Equation~\eqref{eq:proof2} is obtained from $W_{k_m+|U_i|}(z_{\bar 1},\dots,z_{\bar k_m},z_{U_i})$ by applying the convention explained in Remark~\ref{rem:Special02}. Exactly the same convention is applied in Definition~\ref{def:defT}. Thus we can identify~\eqref{eq:proof2} with~\eqref{def:cT-1}.
\end{proof}

\subsection{A closed formula for \texorpdfstring{$\cW_{g,n}^{(r)}$}{Wgnr} and a proof of Theorem~\ref{theor:higher}}\label{Sec54}

In this section we obtain an explicit closed algebraic formula for the expression $\cW^{(r)}_{g,n}$ of the same type as in Proposition~\ref{prop:FormulasWgn} for $W_{g,n}$. This formula will imply Theorem~\ref{theor:higher}.

\begin{proposition} For $n\geq 2$, $(g,n)\not=(0,2)$, we have:
	\label{LEtheor}
	\begin{equation}\label{eq:mainprop2}
		\cW_{g,n}=[\hbar^{2g-2+n}]
		U_n\dots U_2\widetilde U_1
		\sum_{\gamma \in \Gamma_n}\prod_{\{v_k,v_\ell\}\in E_\gamma} w_{k,\ell},
	\end{equation}
	where the sum is over all connected simple graphs on $n$ labeled vertices, and we have Equation~\eqref{eq:wkl} for $w_{k,\ell}$, Equation~\eqref{eq:Uihbar} for the $U_i$, $i=2,\dots,n$, and $\widetilde U_1$ is defined as an operator that acts on a function $f=f(u_1,z_1)$ as
	\begin{align}
		\widetilde U_{1} f&\coloneqq
		\sum_{j,r=0}^\infty D_1^j\left(\frac{\widetilde{L}^j_{r,1}
		}{Q_1}
		[u_1^r] \frac{e^{u_1(\cS(u_1\hbar Q_1D_1)\hat y(z_1)-y(z_1))}}{u_1\hbar\cS(u_1\hbar)}f\right),
		\\
		\widetilde{L}^j_{r,1} &:= \left.\left([v^j]e^{-v\,\psi(y)}\partial_y^r e^{v\frac{\cS(v\,\hbar\,\partial_y)}{\cS(\hbar\,\partial_y)}\hat \psi(y)}u\cS(vu\hbar)e^{u y}\right)\right|_{y=y(z_1)}.
	\end{align}
For $(g,n)=(0,2)$ we have 
\begin{equation}
	\cW_{0,2} = \frac{ue^{u\,y(z_1)}}{Q_1Q_2}\frac{z_1z_2}{(z_1-z_2)^2}.
\end{equation}
For $n=1$, $g>0$ we have 
\begin{equation}
	\cW_{g,1} = [\hbar^{2g}]\left(\hbar \,\widetilde{U}_1 1 + \sum_{j=0}^\infty D_1^{j}\widetilde{L}_{0,1}^{j+1}\; D_1y(z_1)
	\right) .
\end{equation}
Finally, for $(g,n)=(0,1)$ we have
\begin{equation}	
	\cW_{0,1} = e^{u\,y(z_1)}-1.
\end{equation}
\end{proposition}

\begin{proof}
	We compute $\cW_{g,n}$ in a closed form applying the techniques developed in~\cite{BDKS20}. We have
	\begin{equation}
		\begin{aligned}
			\cW_{g,n}^\bullet&=[\hbar^{2g-2+n}]
			\VEV{\left(\sum_{m=1}^\infty \frac{X^m}{m\hbar } J_m\right)\cE_0(u\hbar)
				\left(\prod_{i=2}^nJ(X_i)\right)
				\cD(\hbar)e^{\sum_{i=1}^\infty \frac{\hat{q}_i J_{-i}}{i\hbar} }}
			\\&
			=[\hbar^{2g-2+n}]
			\VEV{\left(\sum_{m=1}^\infty \frac{X^m}{m\hbar } J_m\right)\cE_0(u\hbar)\cD(\hbar)
				\left(\prod_{i=2}^n\euro(X_i)\right)
				e^{\sum_{i=1}^\infty \frac{\hat{q}_i J_{-i}}{i\hbar} }}
		\end{aligned}
	\end{equation}
	We consider the Lie algebra element $\cE_0(\hbar\,u)$ as the first order term in $\epsilon$ of the Lie group element $e^{\epsilon  \cE_0(\hbar\,u)}$ and subsequently consider all expressions over the ring of dual numbers $\mathbb{C}[\epsilon]/(\epsilon^2)$.
	
	Define $\widetilde \cD(\hbar)\coloneqq e^{\epsilon  \cE_0(u\hbar)}\cD(\hbar)$.
It is straightforward to see that
	\begin{equation}
		\widetilde \cD(\hbar) = \exp\left(\sum_{k\in\Z+\frac12}(w(\hbar k)+\epsilon e^{u \hbar k}) \hat E_{k,k}\right)
		=\exp\left(\sum_{k\in\Z+\frac12}\widetilde w(\hbar k) \hat E_{k,k}\right),
	\end{equation}
	where
	\begin{equation}
		\widetilde w(y)\coloneqq w(y)+\epsilon e^{u y}=\frac1{\zeta(\hbar\partial_y)}\widetilde\psi(y)
	\end{equation}
	for
	\begin{equation}
		\widetilde \psi(y)=\hat\psi(y)+\epsilon \zeta(\hbar\partial_y)e^{u y}.
	\end{equation}
	
	Apply~\eqref{eq:eurom}--\eqref{eq:phim} to $\widetilde \cD(\hbar)$.  We have:
	\begin{align}
			\widetilde\euro_m&=\widetilde\cD(\hbar)^{-1}J_m\widetilde\cD(\hbar)\\ \notag
			&=\sum_{r=0}^\infty\partial_y^r\widetilde\phi_m(y)|_{y=0}[u^r z^m]\frac{
				e^{\sum_{i=1}^\infty u\hbar\cS(u\hbar i) J_{-i}z^{-i}}
				e^{\sum_{i=1}^\infty u\hbar\cS(u\hbar i) J_{i}z^{i}}}
			{u\hbar\cS(u\hbar)},
	\end{align}
	where
	\begin{align}
		\widetilde\phi_m(y)
		& =e^{\frac{\zeta(m\hbar\partial_y)}{\zeta(\hbar\partial_y)}\widetilde \psi(y)}
		=e^{\frac{\zeta(m\hbar\partial_y)}{\zeta(\hbar\partial_y)}(\hat\psi(y)+\epsilon \zeta(\hbar\partial_y)e^{u y})}
		\\ \notag
		&=\phi_m(y) e^{\epsilon\zeta(mu\hbar)e^{u y}}
		=\phi_m(y)(1+\epsilon m u \hbar \cS(mu\hbar)e^{u y})
	\end{align}

	Now we can follow directly the arguments in~\cite[Section 4]{BDKS20}. We obtain:
	\begin{align}
			\cW_{g,n}^\bullet&=[\hbar^{2g-2+n}\epsilon^1]
			\VEV{\left(\sum_{m=1}^\infty \frac{X^m}{m\hbar} \widetilde\euro_m\right)
				\left(\prod_{i=2}^n\euro(X_i)\right)
				e^{\sum_{i=1}^\infty \frac{\hat{q}_i J_{-i}}{i\hbar} }}
			\\ \notag &
			= [\hbar^{2g-2+n}]U_n\dots U_2 \widetilde U_1
			\prod_{1\le k<\ell\le n}
			e^{\hbar^2u_ku_\ell\cS(u_k\hbar\,z_k\partial_{z_k})\cS(u_\ell\hbar\,z_\ell \partial_{z_\ell})
				\frac{z_k z_\ell}{(z_k-z_\ell)^2}}.
	\end{align}
	Passing to the connected functions we obtain Equation~\eqref{eq:mainprop2}.
	
In the case $n=1$ and $(g,n)=(0,2)$ one has to slightly modify this computation, cf.~\cite[Section 6]{BDKS20}.
\end{proof}
\begin{remark}
	Note that Proposition~\ref{LEtheor} is the key technical result (the main ingredient) in the proof of the loop equations and, by extension, of the blobbed topological recursion. Let us stress that Proposition~\ref{LEtheor} itself did not require the natural analytic assumptions (or any other assumptions), and it holds in full generality for all possible formal series $\hat\psi$ and $\hat y$.
\end{remark}

This proposition allows us to prove Theorem~\ref{theor:higher}.
\begin{proof}[Proof of Theorem~\ref{theor:higher}]
	We use the same argument as in the proof of Theorem~\ref{th:linloop}.
	We have to show that $\cW_{g,n}^{(r)}\in \hat{\Xi}(z_1)$ for any $g\geq 0$, $n\geq 1$, $r\geq 1$. To this end we use the formulas stated in Proposition~\ref{LEtheor}. In particular, in the case $n\geq 2$ we use Equation~\eqref{eq:mainprop2} (the case $n=1$ is completely analogous).
	
We see that $\cW_{g,n}^{(r)}$ is equal to a finite sum of the expressions $D_1^j Q_1^{-1} f_j$, where, due to the natural analytic assumptions of Definition~\ref{def:naa}, $f_j$ is regular at $z_1\to p_i$, $i=1,\dots, N$ (recall that $p_1,\dots,p_N$ are zeros of $Q(z)$). This allows us to conclude in the same way as in the proof of Theorem~\ref{th:linloop} that $\cW_{g,n}^{(r)}\in \hat{\Xi}(z_1)$.
\end{proof}

\section{Projection property: statements and discussion}\label{sec:proj}

The spectral curve topological recursion can equivalently be reformulated, see~\cite[Theorem 2.2]{BS17} as the recursion for the so-called \emph{normalized} symmetric $n$-differentials $\omega_{g,n}$ that satisfy the abstract loop equations. Recall that a symmetric $n$-differential $\omega_{g,n}$, $2g-2+n>0$, is called normalized if it satisfies the so-called \emph{projection property}:
\begin{equation}
\cP_1\cdots \cP_n \omega_{g,n} = \omega_{g,n},
\end{equation}
 where $\cP\lambda$ for a $1$-form $\lambda$ is defined as
\begin{equation}
	\cP \lambda (z) \coloneqq \sum_{i=1}^N \res_{w\to p_i} \lambda(w) \int_{p_i}^w \omega_{0,2} (\cdot, z).
\end{equation}
Here $\{p_1,\dots,p_N\}$ are the zeros of $\omega_{0,1}$. So, in order to state the topological recursion for the $\hbar^2$-deformed weighted Hurwitz numbers, we have to analyze when the differentials
 \begin{equation}
 	\omega_{g,n}(z_1,\dots,z_n) \coloneqq \Big(\bigotimes_{i=1}^n \frac{dX_i}{X_i}\Big) W_{g,n}  = \Big(\bigotimes_{i=1}^n d_i \Big) H_{g,n}, \qquad 2g-2+n>0,
 \end{equation}
$X_i = X(z_i)$, $i=1,\dots,n$, satisfy the projection property.

Note that the projection property turns out to be the most restrictive out of all the conditions for the topological recursion in the light of the possible choices of the $\hbar^2$-deformations of the functions $\psi(y)$ and $y(z)$. Recall that the functions $\psi(y)$ and $y(z)$ determine the spectral curve data. We have proved all loop equations for the $\hbar^2$-deformed weighted Hurwitz numbers basically in the full generality (when they make sense, i.e. when what we call \emph{natural analytic assumptions} hold) in Section~\ref{SecPrel}. These arguments are applicable to and imply no restrictions on the arbitrary $\hbar^2$-deformations (for which the natural analytic assumptions still hold). In the meanwhile, once an $\hbar^2$-deformation of $(\psi,y)$ such that the resulting $n$-differentials $\omega_{g,n}$ satisfy the projection property exists, it is automatically unique, since the corresponding $\omega_{g,n}$'s can be reconstructed by the topological recursion.

The goal of this section is to state the projection property for two natural families of the $(\psi,y)$-data, in particular, to identify the necessary $\hbar^2$-deformations of these functions, and to discuss these families in the context of the results already known from the literature. The proofs of the projection property for these families are collected in the next section, Section~\ref{sec:ProjPropProofs}.

\subsection{Two families of \texorpdfstring{$(\psi,y)$}{(psi,y)}-data}
The projection property requires careful 
analysis of the poles arising in our general formulas for $H_{g,n}$ for the respective choices of $\hat{\psi}$ and $\hat{y}$, and we do not know how to choose the correct unique $\hbar^2$-deformations and to prove it in full generality (it is also an open question whether the correct unique choice of $\hbar^2$-deformations always exists). We state and prove it below for two
quite general families of $(\psi,y)$-data, which, in fact, subsume as special cases all Hurwitz-type problems for which the topological recursion
(or at least just the projection property)
has already been studied in the literature (with a little caveat discussed in Remark~\ref{rem:generality}). The results are summarized in Table \ref{tab:families}.

\begin{table}[H]
	\scalebox{0.73}{
		\begin{tabular}{|c||c|c||c|c|}
			\hline
			\Gape[0.2cm]{Family} & $\psi(y)$ & $y(z)$ & $\hat\psi(\hbar^2,y)$ & $\hat y(\hbar^2,z)$\\
			\hline
			\hline
			\Gape[0.5cm]{I} & $P_1(y)+ \log\left(\dfrac{P_2(y)}{P_3(y)}\right)$  
			& $\dfrac{R_1(z)}{R_2(z)}$ & $\cS(\hbar \partial_y)P_1(y)+\log\left(\dfrac{P_2(y)}{P_3(y)}\right)$ 
			& $y(z)$ \\
			\hline
			\Gape[0.5cm]{II} & $\alpha\, y$ & $
			\dfrac{R_1(z)}{R_2(z)}+
			\log\left(\dfrac{R_3(z)}{R_4(z)}\right)$ & $\psi(y)$ & $
			\dfrac{R_1(z)}{R_2(z)}+
			\dfrac{1}{\cS(\hbar z\partial_z)}\,\log\left(\dfrac{R_3(z)}{R_4(z)}\right)$ \\
			\hline
	\end{tabular} }
	\caption{Cases of weighted Hurwitz numbers data $\psi$, $y$ where we can prove the projection property and know the unique $\hbar^2$-deformation ($\hat{\psi}$, $\hat y$). Here $P_i$ and $R_i$ are some polynomials (of degree $\geq 0$)
		such that the natural analytic assumptions of Definition~\ref{def:naa} are satisfied.}
	\label{tab:families}
\end{table}

Family I 
includes, as special cases, weighted Hurwitz problems listed in Table \ref{tab:Hurw}.

\begin{table}[H]
	\renewcommand{\arraystretch}{1.5}
	\scalebox{0.9}{
		\begin{tabular}{|c|c|c|}
			\hline
			Hurwitz numbers & $\psi(y)$ & $\hat\psi(y)$\\
			\hline
			\hline
			usual & $y$ & $\psi(y)$ \\
			r-spin
			& $y^r$ & $\cS(\hbar \partial_y) \psi(y)$ \\
			monotone & $\log\big(1/(1-y)\big)$ & $\psi(y)$ \\
			strictly monotone & $\log(1+y)$ & $\psi(y)$ \\
			hypermaps & $\log\big((1+uy)(1+vy)\big)$ & $\psi(y)$\\
			BMS numbers & $\log\big((1+y)^m\big)$ & $\psi(y)$ \\
			polynomially weighted& $\log\left(\sum_{k=1}^d c_k y^k\right)$ & $\psi(y)$\\
			\hline
	\end{tabular} } \scalebox{0.9}{ \begin{tabular}{|c|c|c|}
			\hline
			Variations & $y(z)$ & $\hat y(z)$ \\
			\hline
			\hline
			simple & $z$ & $y(z)$ \\
			orbifold & $z^q$ & $y(z)$ \\
			\Gape[0.1cm]{\shortstack{polynomial\\ double}} & \raisebox{0.2cm}{$\sum_{k=1}^d s_kz^k$} & \raisebox{0.2cm}{$y(z)$} \\
			\hline
	\end{tabular} }
	\renewcommand{\arraystretch}{1}
	\caption{Types of Hurwitz numbers (known from the literature) belonging to Family I, with their $(\psi,y)$-data and the unique $\hbar^2$-extension $(\hat\psi, \hat y)$.}
	\label{tab:Hurw}
\end{table}

Family II, out of already known examples from the literature, includes just the cases of the extended Ooguri-Vafa partition function for the HOMFLY-PT polynomials of torus knots and of the usual double Hurwitz numbers (the latter case is also included in Family I), see Table \ref{tab:Hurw2}.

\begin{table}[H]
	\renewcommand{\arraystretch}{1}
	\begin{tabular}{|c|c|c|c|c|}
		\hline
		\Gape[0.3cm]{Hurwitz numbers} & $\psi(y)$ & $\hat \psi(y)$ & $y(z)$ &$\hat y(z)$\\
		\hline
		\hline
		\Gape[0.3cm]{usual double} & $y$ & $\psi(y)$ & $\sum_{k=1}^d s_kz^k$ & $y(z)$\\
		\hline
		\Gape[0.3cm]{extended Ooguri--Vafa}
		& $\frac{P}{Q}y$ & $\psi(y)$ & $\log \left(\frac{1-A^{-1}z}{1-Az}\right)$ & 
		$\frac{1}{\cS(\hbar z\partial_z)} y(z)$
		\\
		\hline
	\end{tabular}
	\renewcommand{\arraystretch}{1}
	\caption{Types of Hurwitz numbers (known from the literature) belonging to family II, with their $(\psi,y)$-data and the unique $\hbar^2$-extension $(\hat\psi, \hat y)$.}
	\label{tab:Hurw2}
\end{table}

The precise references to the literature and a survey of known results is given in~Section~\ref{sec:PreviousQuasi} below.

\begin{remark}\label{rem:generality}
	Note that $d\psi/dy$ and $dy/dz$ are rational expressions for both Family I and Family II of Table \ref{tab:families}. 
	In the proofs of the projection property given below we additionally assume \emph{the condition of generality}, meaning
	that 
	all zeros of polynomials $P_2$, $P_3$, $R_2$, $R_3$, $R_4$, and of $dX$, where $X(z)=z\exp(-\psi(y(z))$, are simple. 
	For the cases where 
	this condition does not hold
	(which do appear in applications) it is not clear how to prove the projection property directly, but one can prove the Bouchard-Eynard recursion for these cases by taking the limit from the cases when the condition of generality holds, see Section \ref{sec:TopoRec}; the Bouchard-Eynard recursion in turn implies the projection property, and thus we have an indirect proof of projection property for these cases.  
\end{remark}

\subsection{The spaces \texorpdfstring{$\Xi$}{Xi} and \texorpdfstring{$\Theta$}{Theta}}

The projection property is typically considered together with the linear loop equation. These two properties combined hold if and only if the functions $W_{g,n}$, $2g-2+n>0$, belong to a so-called space $\Xi$, which is a subspace of the space $\hat\Xi$. Let us define the space $\Xi$.

In order to simplify the exposition, we restrict ourselves to the case that we have in our Families I and II, namely, from now on we assume that $z$ (defined by the change of variables $X=z\exp(-\psi(y(z)))$) is a global coordinate on a rational curve $C$, and $dX/X$ is a rational $1$-form on $C$. Recall that $p_1,\dots,p_N$, $N\geq 1$, denote the zeros of $dX/X$ (and, due to Lemma~\ref{lem:dxzeros}, they are distinct from $\infty$ under the natural analytic assumptions of Definition~\ref{def:naa}), and for the rest of the text we use the following notation:
\begin{equation}\label{eq:Qcheckdef}
	\check Q(z) :=	\prod_{j=1}^N (z-p_j).
\end{equation}




\subsubsection{Definitions of the space \texorpdfstring{$\Xi$}{Xi}}
\begin{definition}	\label{def:Xi}
	The space $\Xi$ (that we also denote by $\Xi_n$, $n\geq 1$, when we want to stress the number of variables in the notation) is defined as the linear span of the functions
	\begin{equation}
		\prod_{j=1}^n D_j^{d_j} \frac{z_j}{z_j - p_{i_j}}
	\end{equation}
	defined for each $d_1,\dots,d_n\geq 1$ and $1\leq i_1,\dots,i_n\leq N$.
\end{definition}

\begin{remark} 
	Note that if we consider $\Xi_n$ as a space of functions depending on just one variable $z_j=z$, while all other variables are fixed, then $\Xi_n\subset \hat\Xi(z)$. Indeed, in the local coordinate $\nu$ (see Proposition \ref{Prop:localcoord}) we have $D = X \partial_X= \dfrac2\nu\dfrac{d}{d\nu}$ and $\dfrac{z}{z-p_i} = \dfrac1\nu \,f(\nu)$, where $f$ is regular at $\nu=0$; therefore the space $\Xi_n$ (where the functions are regarded as functions of just one chosen variable) is contained in the space $\hat\Xi$ of functions with odd principal parts. 
\end{remark}

\begin{definition}	\label{def:Xi-2}
	The space $\Xi$ (that we also denote by $\Xi_n$, $n\geq 1$, when we want to stress the number of variables in the notation) is defined as the linear span of the functions
\begin{equation}\label{eq:Xialtdef}
	\prod_{j=1}^n D_j^{d_j} \frac{r_{j}(z_j)}{\check Q(z_j)}
\end{equation}
defined for each $d_1,\dots,d_n\geq 1$ and $r_j \in \mathbb{C}[z_j]_{<N}$.
\end{definition}

By direct inspection one can see that

\begin{lemma}\label{lem:twodefinitionsXi} Definitions~\ref{def:Xi} and~\ref{def:Xi-2} are equivalent.
\end{lemma}

The following statement is proved in many papers on topological recursion, in slightly varying formulations:

\begin{proposition}\label{prop:projXi}
	The functions $W_{g,n}$ satisfy the projection property and the linear loop equations if and only if we have $W_{g,n}\in \Xi_n$ for $2g-2+n>0$.
\end{proposition}

\begin{proof} It is enough to observe that the functions $D^d (z-p_i)^{-1}$, $d\geq 1$, $i=1,\dots,N$ do satisfy the projection property and their principal parts at the points $p_1,\dots,p_N$ form a basis in the space of possible principal parts of functions satisfying the linear loop equations. In the meanwhile, the functions that satisfy the projection property are fully determined by their principal parts at  the points $p_1,\dots,p_N$.
\end{proof}

\subsubsection{Definition of the space \texorpdfstring{$\Theta$}{Theta}}
While the projection property itself is formulated in terms of $W_{g,n}$, it turns out that it is quite hard to check whether $W_{g,n}$ belongs to the space $\Xi$. Recall that $W_{g,n} = D_1\dots D_n H_{g,n}$. It turns out that one can introduce another space of functions, which we call $\Theta$, such that
$W_{g,n} \in \Xi$ if and only if $H_{g,n} \in \Theta$,
and it is much easier to check whether a function belongs to the space $\Theta$.

Specifically, let us define the space $\Theta$ as follows.
\begin{definition}\label{def:Theta}
	The space $\Theta$ (which we denote by $\Theta_n$, $n\geq 1$, when we want to stress the number of variables) is  defined as the linear span of functions of the form
		$\prod_{i=1}^n f_i(z_i)$,
where each $f_i(z_i)$
\begin{itemize}
	\item is a rational function on the Riemann sphere;
	\item has poles only at the points $p_1,\ldots,p_N$;
	\item its principal part at $p_k$, $k=1,\dots,N$, is odd with respect to the corresponding deck transformation, that is, $f_i(z_i)+f_i(\sigma_k(z_i))$ is holomorphic at $z_i\to p_k$.
\end{itemize}
\end{definition}

\begin{proposition}\label{prop:XiThetaequiv}
	Let $W=W(z_1,\ldots,z_n)$ and $H=H(z_1,\ldots,z_n)$ be two functions related by $W=D_1\dots D_n H$. Then $W\in \Xi_n$ if and only if $H\in\Theta_n$.
\end{proposition}
\begin{proof}
	%
	%
	Note that the functions $1$ and  $D_j^{\ell} \frac{z_j}{z_j-p_k}$, $k=1,\dots,N$, $\ell\geq 0$, form a basis of the space of rational functions with no poles in points other than $p_1,\ldots,p_N$ and with odd principal part at these points.
	%
	%
\end{proof}
\begin{proposition}\label{prop:projTheta}	
	The functions $W_{g,n}$ satisfy the projection property and the linear loop equations if and only if we have $H_{g,n}\in \Theta_n$ for $2g-2+n>0$.
\end{proposition}
\begin{proof}
	This proposition directly follows from Propositions \ref{prop:projXi} and \ref{prop:XiThetaequiv}.
\end{proof}

\subsection{The main statements} \label{sec:TheoremsProjectionProperty}

Analogous to the paper~\cite{BDKS20} in our $\hbar^2$-deformed case we have:
\begin{proposition}	\label{prop:Hprop}
	For $n\geq 3$:
\begin{align}	
	\label{eq:HgnU}
	\Hc_{g,n}&= 
	[\hbar^{2g-2+n}] \sum\limits_{\gamma\in\Gamma_n} \prod\limits_{v_i\in\mathcal{I}_\gamma} \overline{U}_i \prod\limits_{\{v_i,v_k\}\in E_\gamma\setminus\mathcal K_\gamma}w_{i,k}
	\\&\quad\times \nonumber
	\prod\limits_{\{v_i,v_k\}\in\mathcal{K}_\gamma}\left( \overline{U}_i w_{i,k} +  \hbar u_k \cS(u_k\hbar Q_kD_k)\frac{z_i}{z_k-z_i}\right)
	+\mathrm{const},
\end{align}
where $\Gamma_n$ is the set of simple graphs on $n$ vertices $v_1,\ldots,v_n$, $E_\gamma$ is the set of edges of a graph $\gamma$, $\mathcal{I}_\gamma$ is the subset of vertices of valency $\geq 2$, and $\mathcal{K}_\gamma$ is the subset of edges with one end $v_i$ of valency $1$ and another end $v_k$, and where 
\begin{align}	
	\overline U_i\, f &=  \sum_{r=0}^\infty\sum_{j=1}^\infty D_i^{j-1}\left(\frac{L^j_{r,i}}{Q_i}
	[u_i^r] \frac{e^{u_i(\cS(u_i\,\hbar\,Q_i\,D_i)\hat{y}(z_i)-y(z_i))}}{u_i\hbar\,\cS(u_i\,\hbar)}f\right),\\
	w_{k,\ell}&:=e^{\hbar^2u_ku_\ell\cS(u_k\hbar\,Q_k D_k)\cS(u_\ell\hbar\,Q_\ell D_\ell)
		\frac{z_k z_\ell}{(z_k-z_\ell)^2}}-1, \\
	L^j_{r,i} &:= \left.\left([v^j]\left(\partial_y+v\psi'(y)\right)^r e^{v\left(\frac{\cS(v\hbar\partial_y)}{\cS(\hbar\partial_y)}\hat{\psi}(y)-\psi(y)\right)}\right)\right|_{y=y(z_i)}.
\end{align}
For $n=2$ and $g>0$ we have:
\begin{align}
	H_{g,2}& =[\hbar^{2g}]\Bigg(\overline{U}_1\overline{U}_2 w_{1,2}
	+ \overline{U}_1\left(\hbar u_1 \cS(u_1\hbar Q_1D_1)\frac{z_2}{z_1-z_2}\right)
	.\\ \notag
	& \qquad +\overline{U}_2\left( \hbar u_2 \cS(u_2\hbar Q_2D_2)\frac{z_1}{z_2-z_1}\right)\Bigg)
	+\mathrm{const}.
\end{align}
 For $n=1$ and $g>0$ we have:
\begin{align} \label{eq:(g,1)}
	H_{g,1}& =[\hbar^{2g}]\Bigg(\hbar\overline{U}_1 1 + \sum_{j=1}^\infty D_1^{j-1}L_{0,1}^{j+1}\; D_1y(z_1) +\int_{0}^{z_1}\dfrac{\hat{y}(z)-y(z)}{z}dz
	\\ \notag & \quad
	+\int_{0}^{z_1}\dfrac{Q(z)}{z}\left.\left(\dfrac{1}{\cS(\hbar\partial_y)}\hat{\psi}(y)-\psi(y)\right)\right|_{y=y(z)}Dy(z) dz\Bigg)
	+\mathrm{const}.
\end{align}
%
In each case the extra constant can be determined from the condition that $H_{g,n}$ vanishes at zero. These constants are not important for the argument below and can be ignored.
\end{proposition}

\begin{remark}
	It is easy to see that formulas \eqref{eq:HgnU}--\eqref{eq:(g,1)} of Proposition~\ref{prop:Hprop} are consistent with formulas \eqref{eq:mainprop}--\eqref{eq:Wg1} of Proposition~\ref{prop:FormulasWgn}; i.e. that the latter ones can be obtained from the former by an application of the $D_1\cdots D_n$ operator.
\end{remark}

\subsubsection{Statements}

\begin{theorem} \label{thm:Xi-Family1}
	Let $\hat\psi =  \cS(\hbar \partial_y)P_1(y)+\log P_2(y) -\log P_3(y)$ and $\hat y = R_1(z)/R_2(z)$ (Family I from Table \ref{tab:families}). Consider the functions $H_{g,n}$ given by \eqref{eq:HgnU}-\eqref{eq:(g,1)} for this choice of $(\hat \psi,\hat y)$. Under the assumption that the polynomials $P_1,P_2,P_3,R_1,R_2$ are in general position, we have $H_{g,n}\in \Theta$, $2g-2+n>0$.
\end{theorem}

\begin{theorem}  \label{thm:Xi-Family2}
	Let $\hat\psi = \alpha y$ and $\hat y = R_1(z)/R_2(z)+
	\cS(\hbar z\partial_z)^{-1} (\log R_3(z)- \log R_4(z))$ (Family II from Table \ref{tab:families}). Consider the functions $H_{g,n}$ given by \eqref{eq:HgnU}-\eqref{eq:(g,1)} for this choice of $(\hat \psi,\hat y)$. Under the assumption that the polynomials $R_1,R_2,R_3,R_4$ are in general position, we have $H_{g,n}\in \Theta$, $2g-2+n>0$.
\end{theorem}

Proofs of these theorems is given in Sections~\ref{sec:ProofProjFamily1} and \ref{sec:ProofProjFamily2}. The precise meaning of the assumption of generality (in general terms explained in Remark~\ref{rem:generality}) for the data $P_1,P_2,P_3,R_1,R_2$ in Theorem~\ref{thm:Xi-Family1} and $R_1,R_2,R_3,R_4$ in Theorem~\ref{thm:Xi-Family2} is also specified there. Taking into account Proposition \ref{prop:projTheta}, we immediately get the following corollary:
\begin{corollary} The functions
	 $W_{g,n}=D_1\dots D_n H_{g,n}$ for $(\hat \psi,\hat y)$ specified in Theorems \ref{thm:Xi-Family1} and~\ref{thm:Xi-Family2} satisfy the projection property together with the linear loop equations.
\end{corollary}

\subsection{Quasi-polynomiality: historical remarks} \label{sec:PreviousQuasi}
In this Section we make a short overview of the previously known cases of quasi-polynomiality or, equivalently, the cases when the corresponding functions $W_{g,n}$ were known to belong to $\Xi$ in the literature.

\subsubsection{Quasi-polynomiality} One of the ways to check whether $W_{g,n}\in\Xi$, $2g-2+n>0$, is to consider the expansion of $W_{g,n}$ in $X_1,\dots,X_n$ at $X_1=\cdots=X_n=0$ and to check whether it can be represented as
\begin{equation}
D_1\cdots D_n\sum_{i_1,\dots,i_n=1}^N A_{g;i_1,\dots,i_n}(D_1,\dots,D_n) \prod_{j=1}^n \xi^{i_j}(X_j),	
\end{equation}
where $A_{g;i_1,\dots,i_n}\in\mathbb{C}[D_1,\dots,D_n]$ are certain polynomials and $\xi^i(X)$, $i=1,\dots,N$, are the expansions of the functions $(z-p_i)^{-1}$ in $X$ at $X=z=0$. Alternatively, using Definition~\ref{def:Xi-2}, one can check whether the expansion of $W_{g,n}$ in $X_1,\dots,X_n$ at $X_1=\cdots=X_n=0$ can be represented as
\begin{equation}
	D_1\cdots D_n\sum_{\alpha_1,\dots,\alpha_n=0}^{N-1} \tilde A_{g;\alpha_1,\dots,\alpha_n}(D_1,\dots,D_n) \prod_{j=1}^n \tilde \xi^{\alpha_j}(X_j),	
\end{equation}
where $\tilde{A}_{g;\alpha_1,\dots,\alpha_n}\in\mathbb{C}[D_1,\dots,D_n]$ are certain polynomials and $\tilde\xi^\alpha(X)$, $\alpha=0,\dots,N-1$, are the expansions of the functions $z^\alpha / \check{Q}(z)$ in $X$ at $X=z=0$.

If we rewrite what this representation of $W_{g,n}$, or, better, $H_{g,n}$, in terms of the properties of the coefficients of its expansion in $X_1,\dots,X_n$, known as \emph{weighted Hurwitz numbers}, we obtain the following reformulation of the the property that $W_{g,n}\in\Xi$ for $2g-2+n>0$. Let
\begin{equation}
	H_{g,n} \coloneqq \sum_{k_1,\dots,k_n=1}^\infty h_{g;k_1,\dots,k_n} \prod_{i=1}^n X_i^{k_i}.
\end{equation}

\begin{proposition} We have $W_{g,n}\in\Xi$ for $2g-2+n>0$ if and only if there exist polynomials $A_{g;i_1,\dots,i_n}\in\mathbb{C}[k_1,\dots,k_n]$, $g\geq 0$, $1\leq i_1,\dots,i_n \leq N$, such that for any $k_1,\dots,k_n\geq 1$
\begin{equation} \label{eq:quasipol-1}
	h_{g;k_1,\dots,k_n} = \sum_{i_1,\dots,i_n=1}^N A_{g;i_1,\dots,i_n}(k_1,\dots,k_n) \prod_{j=1}^n [X_j^{k_j}] \xi^{i_j}(X_j).
\end{equation}
\end{proposition}

\begin{remark}
An alternative way to formulate this proposition is the following: we have $W_{g,n}\in\Xi$ for $2g-2+n>0$ if and only if there exist polynomials $\tilde A_{g;i_1,\dots,i_n}\in\mathbb{C}[k_1,\dots,k_n]$, $g\geq 0$, $0\leq \alpha_1,\dots,\alpha_n \leq N-1$, such that for any $k_1,\dots,k_n\geq 1$
\begin{equation} \label{eq:quasipol-2}
	h_{g;k_1,\dots,k_n} = \sum_{\alpha_1,\dots,\alpha_n=0}^{N-1} \tilde A_{g;\alpha_1,\dots,\alpha_n}(k_1,\dots,k_n) \prod_{j=1}^n [X_j^{k_j}] \tilde\xi^{\alpha_j}(X_j).
\end{equation}
\end{remark}

\begin{definition}
	The property when weighted Hurwitz numbers $H_{g;k_1,\dots,k_n}$ can be expressed in the form~\eqref{eq:quasipol-1} and/or~\eqref{eq:quasipol-2} is called the \emph{quasi-polynomiality} property.
\end{definition}

\begin{remark} If the differentials $\omega_{g,n} = W_{g,n} \prod_{j=1}^n dX_j/X_j$ satisfy in addition the topological recursion, it is proved in~\cite{Eynard,DOSS} that $\deg A=\deg \tilde A = 3g-3+n$. See Section~\ref{sec:TopoRecAppl} for details.
\end{remark}

\subsubsection{Historical remarks}

The first instance of quasi-polynomiality was observed and conjectured for purely combinatorial reasons for the usual simple Hurwitz numbers in~\cite{GouldenJacksonVainshtein}. In our terms, they conjectured that for $\psi=y$ and $y=z$, the coefficient of the expansion of $H_{g,n}$ are expressed as
\begin{equation}
h_{g;k_1,\dots,k_n} = A_g(k_1,\dots,k_n) \prod_{i=1}^n \frac{k_i^{k_i}}{k_i!}	
\end{equation}
for some polynomials $A_g\in \mathbb{C}[k_1,\dots,k_n]$. This conjecture was first proved in~\cite{ELSV} by providing a formula for $h_{g;k_1,\dots,k_n}$ in terms of the intersection numbers on the moduli space of curves $\cM_{g,n}$ (the celebrated \emph{ELSV-formula}). For a long time it was an open question whether an alternative purely combinatorial proof for quasi-polynomiality could be found, and it was settled in two different ways in~\cite{DKOSS} and~\cite{KLS}.

For the usual orbifold Hurwitz numbers ($\psi=y$, $y=z^q$) an ELSV-type formula that proves quasi-polynomiality is derived in~\cite{JohnsonPandTseng}, and two different independent combinatorial proofs are available in~\cite{DLPS,KLS}.
For the usual double Hurwitz numbers ($\psi=y$, $y$ polynomial) the quasi-polynomiality was conjectured in~\cite{Do-Karev-DoubleHurwitz} and proved in~\cite{BDKLM}. For the coefficients of the extended Ooguri-Vafa partition function of the colored HOMFLY polynomials of torus knots ($\psi=\alpha z$, $y =\log(1-A^{-1}z)-\log(1-Az)$) the quasi-polynomiality was proved in~\cite{DPSS}

For the monotone simple Hurwitz numbers ($\psi=-\log(1-y)$, $y=z$)  the quasi-polynomiality is proved in~\cite{GouldenGuayNovak-Polynomiality}, and an alternative proof is available in~\cite{KLS}. For the monotone orbifold Hurwitz numbers ($\psi=-\log(1-y)$, $y=z^q$) the quasi-polynomiality was first conjectured in~\cite{DoKarev} and then proved in~\cite{KLS}.

For the strictly monotone orbifold Hurwitz numbers ($\psi=\log(1+y)$, $y=z^q$)  the quasi-polynomiality was conjectured in~\cite{DoManescu} and proved in~\cite{KLS} (it also follows from~\cite{DOPS-CombLoopEqua}, and for $q=2$ it also follows from~\cite{DumitrescuMulaseSafnukSorkin,Norbury-Lattice,KZ2015}).

For the hypermaps, or, equivalently, dessins d'enfants ($\psi = \log(1+uy)+\log(1+vy)$, $y=z$) the quasi-polynomiality follows from \cite{KZ2015}. For the Bousquet-M\'elou--Schaeffer numbers ($\psi = m\log(1+y)$, $y=z$) the quasi-polynomiality is proved in~\cite{BDS}. For the more general polynomially weighted polynomially double Hurwitz numbers ($\exp(\psi)$ and $y$ are polynomials) the quasi-polynomiality, under some assumptions on general position, follows from~\cite{ACEH-2}.

For the $r$-spin simple Hurwitz numbers ($\psi = y^r$, $y=z$) the quasi-polynomiality was conjectured in~\cite{Zvonkine} as a property of an ELSV-type formula conjectured there (see also~\cite{ShaSpiZvo}). The quasi-polynomiality was proved in a more general case of the $r$-spin orbifold Hurwitz numbers ($\psi = y^r$, $y=z^q$) in~\cite{KLPS19}.

In all cases mentioned above the quasi-polynomiality is either proved directly by  combinatorial arguments in the group algebra of the symmetric group, or via an analysis of the operators whose vacuum expectations give the corresponding weighted Hurwitz numbers in the semi-infinite wedge formalism, or via the structure of the ELSV-type formulas (once these formulas are proved independently) or via the topological recursion (once the topological recursion is proved independently).

\section{Projection property: proofs} \label{sec:ProjPropProofs}

In this Section we prove Theorems~\ref{thm:Xi-Family1} and~\ref{thm:Xi-Family2}, which in particular imply that the projection property holds for Family I and Family II introduced in Section~\ref{sec:proj}.
In other words, we want to show that the functions $H_{g,n}$, $2g-2+n>0$, obtained from \eqref{eq:HgnU}-\eqref{eq:(g,1)} for the respective choices of $\hat{\psi}(\hbar^2,y)$ and $\hat{y}(\hbar^2,z)$, belong to the space $\Theta_n$.

We begin with the a few general observations related to the structure  of the formulas~\eqref{eq:HgnU}-\eqref{eq:(g,1)} for $H_{g,n}$, $2g-2+n> 0$, and relevant for the both families of parameters.  These formulas give manifestly rational functions, whose principal parts at the points $p_1,\dots, p_N$ are odd with respect to the deck transformations. Moreover, they are finite linear combinations of functions of the form $\prod\limits_{i=1}^n f_i(z_i)$, which follows from the fact that all diagonal poles get canceled out (see \cite[Theorem 1.1 and Remark 1.3]{BDKS20}). So, recalling Proposition~\ref{prop:projTheta} and Definition~\ref{def:Theta}, in order to prove the projection property we just have to show that for $2g-2+n>0$ functions $H_{g,n}$  have no poles other than $p_1,\dots, p_N$ in each variable $z_1,\dots,z_n$.

Consider a particular $H_{g,n}(z_1,\dots,z_n)$. From the shape of the formula it is clear that its possible poles in the variable $z_1$ in addition to $p_1,\dots,p_N$ are either at the diagonals $z_1-z_i=0$, $i\not=1$ (but it is known from~\cite{BDKS20} that these functions have no poles at the diagonals $z_i-z_j=0$), or at $\infty$, or at the special points related to the specific form of the operator $\overline U_1$ for Family I and Family II. A bit more special case is the case of $H_{g,1}$, where we have to analyze some extra terms as well.

Note that it is in fact sufficient to analyze the pole structure just for $H_{g,1}$, $g\geq 1$, since this case subsumes the corresponding analysis of the pole structure for $H_{g,n}$, $n\geq 2$. Indeed, the factors of the form $w_{k,\ell}$ and $\hbar u_k \cS(u_k\hbar Q_kD_k)\frac{z_i}{z_k-z_i}$ do not contribute any poles to the resulting expressions, as all diagonal poles get canceled and these factors are regular at infinity as well. Therefore, the possible extra poles
can only occur at the special points of $\overline U_i$, which enters the formula for $H_{g,1}$ in exactly the same way as formulas for $H_{g,n}$ for other values of $n$. The argument for the $n=1$ case includes analysis of the singularities of $\overline U_1$, and once we show that it has no poles outside $p_1,\dots,p_N$, it immediately implies the same statement for any $n\geq 2$ as well.

\subsection{Technical lemmata} We begin with two technical lemmata that conceptually explain the origin of the shape of $\hbar^2$-deformations given in Table~\ref{tab:families}.

\begin{lemma}\label{lem:cvy}
	Let $A\in \mathbb{C}$ be a fixed number. Consider the $\hbar$-series expansion
	\begin{equation}\label{eq:vlog}
		e^{v\left(\frac{\cS(v\,\hbar\,\partial_y)}{\cS(\hbar\,\partial_y)}-1\right)\log\left(y-A\right) } = \sum_{k=0}^\infty \hbar^{2k} c_{2k}(v,y).
	\end{equation}	
	Note that $c_0(v,y) = 1$. For the coefficients $c_{2k}(v,y)$, $k\geq 1$, we have
	\begin{align}
		\label{eq:cvy}
		c_{2k}(v,y) &= \dfrac{(v+1)v(v-1)\cdots (v-2k+1)}{(y-A)^{2k}} \, p_{2k}(v),
	\end{align}
	where $p_{2k}(v)$, $k\geq 1$, are some polynomials in $v$.
\end{lemma}
\begin{proof} Note that
\begin{equation} \label{eq:OperatorUnderExp}
	v\left(\frac{\cS(v\,\hbar\,\partial_y)}{\cS(\hbar\,\partial_y)}-1\right) = v \sum_{i=1}^\infty r_i(v) \hbar^{2i} \partial_y^{2i}
\end{equation}
for some polynomials $r_i(v)\in \mathbb{C}[v^2]_{\leq 2i}$, $i\geq 1$.	Applying this expansion in~\eqref{eq:vlog} we see that
	\begin{equation}
		c_{2k}(v,y) = \dfrac{1}{(y-A)^{2k}} \, \tilde{c}_{2k}(v),
	\end{equation}
	where $\tilde{c}_{2k}(v)$ is some polynomial in $v$, $k\geq 0$. In order to prove~\eqref{eq:cvy}, it is sufficient to show that $c_{2k}(v,y)$ vanishes for $v\in\{-1,0,1,\dots,2k-1\}$, $k\geq 1$.
	
	Let $v=m\in \mathbb{Z}_{\geq -1}$. For $m=-1, 0, 1$ the coefficients $c_{2k}(m,y)$, $k\geq 1$, vanish since the whole operator~\eqref{eq:OperatorUnderExp}  vanishes. For $m\geq 2$
	\begin{align}\label{eq:cvycomp}
		v\left(\frac{\cS(v\,\hbar\,\partial_y)}{\cS(\hbar\,\partial_y)}-1\right)\log\left(y-A\right)
		& = \left(\frac{e^{m\hbar \partial_y/2}-e^{-m\hbar \partial_y/2}}{e^{\hbar \partial_y/2}-e^{-\hbar \partial_y/2}}-m\right)\log\left(y-A\right)\\ \nonumber
		& = \left(-m+\sum_{i=0}^{m-1}e^{(2i-m+1)\hbar \partial_y/2}\right)\log\left(y-A\right)\\ \nonumber
		&=-m\log(y-A)+\sum_{i=0}^{m-1} \log\left(y-A + \tfrac{2i-m+1}{2}\hbar\right).
	\end{align}
	Therefore,
	\begin{align}
		e^{m\left(\frac{\cS(m\,\hbar\,\partial_y)}{\cS(\hbar\,\partial_y)}-1\right)\log\left(y-A\right) } = \prod_{i=0}^{m-1}\dfrac{y-A + \frac{2i-m+1}{2}\hbar}{y-A}.
	\end{align}
	The latter expression is an even polynomial in $\hbar$ of degree $\leq m$, so its coefficient in front of $\hbar^{2k}$ for $2k > m$ vanishes, which implies that $c_{2k}(m,y)=0$ for $2 \leq m\leq 2k-1$.
\end{proof}


\begin{lemma}\label{lem:cuzdz}
	Let $A\in \mathbb{C}$ be a fixed number. Consider the series expansion in $\hbar$
	\begin{equation}
		e^{u\left(\frac{\cS(u\,\hbar\,z \partial_z)}{\cS(\hbar\,z \partial_z)}-1\right)\log\left(z-A\right) }
		=\sum_{k=0}^\infty \hbar^{2k} c_{2k}(u,z).
	\end{equation}	
	Note that $c_0(u,z) = 1$. For $k\geq 1$, the coefficients $c_{2k}(u,z)$ represented as Laurent series in $z-A$ as
	\begin{align}
		c_{2k}(u,z) &= \sum_{l=-\infty}^\infty \dfrac{d_{2k,l}(u)}{(z-A)^l},
	\end{align}
	where for $l>0$ the coefficients $d_{2k,l}(u)$ are represented as
	\begin{equation}
		d_{2k,l}(u) = (u+1)u(u-1)\cdots(u-l+1)\cdot p_{2k,l}(u)
	\end{equation}
	for some polynomials $p_{2k,l}(u)\in\mathbb{C}[u]$.
\end{lemma}
\begin{proof}
	The proof is similar to the proof of the previous lemma. We have to show that for $k>0$ the polynomials $d_{2k,l}(u)$ vanish for $u\in\{-1,0,1,\dots,l-1\}$. Vanishing at $u=-1,0,1$ is evident. Let $u=m\in\mathbb{Z}_{\geq 2}$. Analogously to \eqref{eq:cvycomp} we then have:
	\begin{align}
		u\left(\frac{\cS(u\,\hbar\,z \partial_z)}{\cS(\hbar\,z \partial_z)}-1\right)\log\left(z-A\right) & = \left(-m+\sum_{i=0}^{m-1}e^{(2i-m+1)\hbar z\partial_z/2}\right)\log\left(z-A\right) \\ \nonumber
		&= -m\log(z-A)+\sum_{i=0}^{m-1} \log\left(z e^{\frac{2i-m+1}{2}\hbar}-A\right).
	\end{align}
	Thus
	\begin{equation}
		e^{u\left(\frac{\cS(u\,\hbar\,z \partial_z)}{\cS(\hbar\,z \partial_z)}-1\right)\log\left(z-A\right) } = \prod_{i=0}^{m-1}\dfrac{z e^{\frac{2i-m+1}{2}\hbar}-A}{z-A}.
	\end{equation}
	The Laurent expansion of the latter expression in $z-A$ clearly does not contain any terms $(z-A)^{-l}$ for $l>m$.
\end{proof}

\subsection{Family I: proof} \label{sec:ProofProjFamily1}
In this section we present the proof of Theorem~\ref{thm:Xi-Family1}.

Recall that for this family
	\begin{align} \label{eq:FamI-psi}
		\psi(y)&=P_1(y)+ \log\left(\dfrac{P_2(y)}{P_3(y)}\right)
		, & \hat{\psi}(\hbar^2,y)&=\cS(\hbar \partial_y)P_1(y)+\log\left(\dfrac{P_2(y)}{P_3(y)}\right)
		,\\ \label{eq:FamI-y}
		y(z)&=\dfrac{R_1(z)}{R_2(z)}, & \hat y(\hbar^2,z)&=\dfrac{R_1(z)}{R_2(z)}.
	\end{align}
	Denote $d_i\coloneqq\deg P_i$, $i=1,2,3$, and $e_j\coloneqq \deg R_j$, $j=1,2$. 
	According to Remark~\ref{rem:generality}, we assume that the zeros of polynomials $P_2$, $P_3$, $R_2$, and also of $dX$ are all simple.
	
\subsubsection{Specialization of the formulas for $H_{g,1}$ and extra notation}	As we discussed above, it is sufficient to restrict the analysis of singularities to the case of $n=1$.
For $\hat \psi$ and $\hat y$ as in~\eqref{eq:FamI-psi},~\eqref{eq:FamI-y}, the operator $\overline{U}_1$ takes the form
	\begin{align}	\label{eq:UiI}
		\overline{U}_{1} f&=\sum_{r=0}^\infty\sum_{j=1}^\infty D_1^{j-1}\left(\frac{L^j_{r,1}}{Q_1}
		[u_1^r] \frac{e^{u_1\left(\cS(u_1\,\hbar\,Q_1\,D_1)-1\right)\frac{R_1(z)}{R_2(z)} }}{u_1\hbar\,\cS(u_1\,\hbar)}f(u_1,z_1)\right),
		\\
		L^j_{r,1}&= [v^j]
		(\partial_y+v\,\psi'(y))^r e^{v\left(\cS(v\,\hbar\,\partial_y)-1\right)P_1(y)
		}e^{v\left(\frac{\cS(v\,\hbar\,\partial_y)}{\cS(\hbar\,\partial_y)}-1\right)\log\left(\frac{P_2(y)}{P_3(y)}
			\right)}\Big|_{y=\frac{R_1(z_1)}{R_2(z_1)}}, \label{eq:LrI}\\
		\label{eq:ProofFam1-Q1}
		Q_1 &= 1-z_1\left(P_1'+\dfrac{P_2'\left(\frac{R_1(z_1)}{R_2(z_1)}\right)}{P_2\left(\frac{R_1(z_1)}{R_2(z_1)}\right)}-\dfrac{P_3'\left(\frac{R_1(z_1)}{R_2(z_1)}\right)}{P_3\left(\frac{R_1(z_1)}{R_2(z_1)}\right)}\right)\dfrac{R_1'(z_1) R_2(z_1)-R_1(z_1)R_2'(z_1)}{R_2^2(z_1)},
	\end{align}
and
\begin{align} \label{eq:Hg1I}
	H_{g,1}
	&=[\hbar^{2g}]\left(\hbar\overline{U}_1 1 + \sum_{j=1}^\infty D_1^{j-1}L^{j+1}_{0,1}D_1\dfrac{R_1(z_1)}{R_2(z_1)}\right)
	\\ \notag
	&\quad +\Bigl([u^{2g}]\frac1{\cS(u)}\Bigr)  \left.\left(\partial_y^{2g-1}\log\left(\dfrac{P_2(y)}{P_3(y)}\right)\right)\right|_{y=\frac{R_1(z_1)}{R_2(z_1)}}
	+\mathrm{const}.
\end{align}
Denote the summands of $H_{g,1}$ by
\begin{align}
	\tau^{(1)}_g&:=[\hbar^{2g}]\hbar\overline{U}_1 1,\\
	\tau^{(2)}_g&:=[\hbar^{2g}]\sum_{j=1}^\infty D_1^{j-1}L^{j+1}_{0,1}D_1\dfrac{R_1(z_1)}{R_2(z_1)},\\
	\tau^{(3)}_g&:=\Bigl([u^{2g}]\frac1{\cS(u)}\Bigr)  \left.\left(\partial_y^{2g-1}\log\left(\dfrac{P_2(y)}{P_3(y)}\right)\right)\right|_{y=\frac{R_1(z_1)}{R_2(z_1)}},
\end{align}
which allows to rewrite $H_{g,1}$ as
\begin{equation}\label{eq:HtauI}
	H_{g,1} = \tau^{(1)}_g+\tau^{(2)}_g+\tau^{(3)}_g +\mathrm{const}.
\end{equation}

We introduce also an extra piece of notation to express $\psi(y(z))$, $\check Q$, and $Q$. We have
	\begin{equation}
		\psi(y(z)) = \dfrac{\Pt_1(z)}{\big(R_2(z)\big)^{d_1}}+ \log\left(\dfrac{\Pt_2(z)}{\Pt_3(z)}\right) 
	\end{equation}
	for some polynomials $\Pt_i(z)$. Then, using Definition \eqref{eq:Qcheckdef} and Equation \eqref{eq:ProofFam1-Q1}, we have:
	\begin{align} \label{eq:QcheckI}
		\check Q(z)&= R_2^{d_1+1} \Pt_2 \Pt_3+z \left(d_1  R_2' \Pt_1 \Pt_2 \Pt_3 -  R_2 \Pt_1' \Pt_2 \Pt_3   +  R_2^{d_1+1}\Pt_2 \Pt_3' -  R_2^{d_1+1}\Pt_2' \Pt_3 \right), \\ 
		Q(z) &= \dfrac{\check Q(z)}{R_2^{d_1+1} \Pt_2 \Pt_3}.
	\end{align}
	Here we have slightly abused the notation: by \eqref{eq:Qcheckdef}  $\check Q(z)$ is expected to be a monic polynomial, which is not necessarily the case in \eqref{eq:QcheckI}.

\subsubsection{Part 1 of the proof: $\tau^{(1)}_g$}
\label{prp:1}
Let us prove that $\tau^{(1)}_g\in\Theta$. Recall \eqref{eq:UiI}. Note that the expression
	\begin{equation}
		\left(\frac{\cS(u_1\,\hbar\,\partial_y)}{\cS(\hbar\,\partial_y)}-1\right) \log\left(\frac{P_2(y)}{P_3(y)}\right)
	\end{equation}
is a series in $\hbar$ with coefficients given by rational functions of $y$. So,  $\tau^{(1)}_g$ is manifestly a rational function in $z_1$, and its set of possible poles includes $p_1,\dots,p_N$ (the zeros of $\check Q$),  the zeros of $\Pt_2$, $\Pt_3$, $R_2$, and at $z=\infty$. In a sequence of lemmata below we show that $\tau^{(1)}_g$ has no poles at the zeros of $\Pt_2$, $\Pt_3$, $R_2$, and $z=\infty$. Since the poles at $p_1,\dots,p_N$ are odd with respect to the deck transformations (they are generated by the iterative application of the operator $D_1$ to a function with the first order poles at $p_1,\dots,p_N$, cf.~the proof of Proposition~\ref{prop:XiThetaequiv}), we conclude that $\tau^{(1)}_g\in \Theta$.

\begin{lemma} \label{lem:R2-Fam1} $\tau_g^{(1)}$ has no poles at the zeros of $R_2(z)$.
\end{lemma}

\begin{proof}

Let $B$ be a zero of $R_2(z)$. 
		Note that for $z_1\rightarrow B$ we have $y(z_1)=R_1(z_1)/R_2(z_1) \rightarrow \infty$, moreover, if $B$ is a simple zero of $R_2$ then it is a simple pole of $y(z_1)$. Let us count the order of the pole of
\begin{align} \label{eq:R2-expr1-orderpole-B}
& \sum_{r=0}^\infty\sum_{j=1}^\infty D_1^{j-1}\Bigg([v^j]
(\partial_y+v\,\psi'(y))^r e^{v\left(\cS(v\,\hbar\,\partial_y)-1\right)P_1(y)
}e^{v\left(\frac{\cS(v\,\hbar\,\partial_y)}{\cS(\hbar\,\partial_y)}-1\right)\log\left(\frac{P_2(y)}{P_3(y)}
	\right)}\Big|_{y=\frac{R_1(z_1)}{R_2(z_1)}}\times \\
\notag & \hspace{3cm} \frac{1}{Q_1}
	[u_1^r] \frac{e^{u_1\left(\cS(u_1\,\hbar\,Q_1\,D_1)-1\right)\frac{R_1(z_1)}{R_2(z_1)} }}{u_1\hbar\,\cS(u_1\,\hbar)}\Bigg)
\end{align}		
at $z_1=B$. To this end, two immediate observations are in order:
\begin{itemize}
	\item Firstly, note that $e^{v\left(\frac{\cS(v\,\hbar\,\partial_y)}{\cS(\hbar\,\partial_y)}-1\right)\log\left(\frac{P_2(y)}{P_3(y)}
		\right)}$ does not contribute to the pole at infinity in $y$, and, therefore, to the pole in $z$ at $z=B$, and can be safely ignored in this computation. Indeed, let us factorize $P_2(y)$ and $P_3(y)$ and then decompose $\log\left(\frac{P_2(y)}{P_3(y)}\right)$ as a sum with plus and minus signs of logarithms of linear functions. Then note that the operator $v\left(\frac{\cS(v\,\hbar\,\partial_y)}{\cS(\hbar\,\partial_y)}-1\right)$ contains at least one differentiation over $y$, therefore we obtain with necessity a regular function at infinity in $y$.
	\item Secondly, note that $Q_1^{-1}$ has zero of order $d_1+1$ at $z_1=B$ and each application of $D_1=Q_1^{-1}z_1\partial_{z_1}$ decreases the degree of the pole in $z_1$ at $B$ by $d_1$. The total effect of the factor $Q_1^{-1}$ in the middle of the formula and $D_1^{j-1}$ is the decrease of the order of pole by $d_1j +1$.
\end{itemize}
Therefore, the order of the pole of~\eqref{eq:R2-expr1-orderpole-B} is equal to the order of pole at $z_1=B$ of
\begin{align} \label{eq:R2-expr2-orderpole-B}
	& (z_1-B)\sum_{r=0}^\infty\Bigg(
	(\partial_y+v\,\psi'(y))^r e^{v\left(\cS(v\,\hbar\,\partial_y)-1\right)P_1(y)
	}\Big|_{y=\frac{R_1(z_1)}{R_2(z_1)}}\times \\
	\notag & \hspace{3cm}
	[u_1^r] \frac{e^{u_1\left(\cS(u_1\,\hbar\,Q_1\,D_1)-1\right)\frac{R_1(z_1)}{R_2(z_1)} }}{u_1\hbar\,\cS(u_1\,\hbar)}\Bigg) \Big|'_{v=(z_1-B)^{d_1}}\,,
\end{align}	
where by $|'$ we mean that we only select the terms with $\deg v \geq 1$ before the substitution $v=(z_1-B)^{d_1}$.
Note also that
\begin{itemize}
	\item Since $y(z_1)$ has a simple pole at $z_1=B$, each $\partial_y$ decreases the order of pole in the resulting expression by $1$.
	\item Multiplication by $\psi'(y)$ increases the order of pole by $d_1-1$.
\end{itemize}		
Taking into account these two observations and that each $v$ factor decreases the order of pole by $d_1$, we see that each application of the operator $\partial_y + v\psi'(y)$ decreases the order of pole in the resulting expression by $1$.   Therefore, the order of the pole of~\eqref{eq:R2-expr2-orderpole-B} is equal to the order of pole at $z_1=B$ of
\begin{align} \label{eq:R2-expr3-orderpole-B}
	& (z_1-B)\Bigg(
	 e^{v\left(\cS(v\,\hbar\,\partial_y)-1\right)P_1(y)
	}\Big|_{y=\frac{R_1(z_1)}{R_2(z_1)}}
	\frac{e^{u_1\left(\cS(u_1\,\hbar\,Q_1\,D_1)-1\right)\frac{R_1(z_1)}{R_2(z_1)} }}{u_1\hbar\,\cS(u_1\,\hbar)}\Bigg) \Bigg|''_{\substack{v=(z_1-B)^{d_1}\\ u_1 = z_1-B}}\,,
\end{align}	
where by $|''$ we mean that we only select the terms with $\deg v \geq 1$ and regular in $u_1$ before the substitutions $v=(z_1-B)^{d_1}$, $u=z_1-B$.
In this expression, in the first exponent in $\cS(v\,\hbar\,\partial_y)-1$ each $v\,\hbar\,\partial_y$ does not increase the order of the pole at $z_1=B$ (in fact, it even decreases it by $d_1+1$); since $vP_1(y)$ has no pole at $z_1=B$ this means that the whole first exponential is regular. 
In the second exponent, in $\left(\cS(u_1\,\hbar\,z_1\,\partial_{z_1})-1\right)$ each $u_1\,\hbar\,z_1\,\partial_{z_1}$ preserves the order of the pole at $z_1=B$; since $u_1 R_1(z_1) / R_2(z_1)$ has no pole at $z_1=B$ this means that the whole second exponential is regular. 
Finally, $(z_1-B)/ (u_1\hbar\,\cS(u_1\,\hbar))$ is also regular at $z_1=B$ in this expression.

Thus, \eqref{eq:R2-expr3-orderpole-B} is regular at $z_1=B$, and therefore~\eqref{eq:R2-expr1-orderpole-B} is regular at $z_1=B$ as well.
\end{proof}

\begin{lemma} \label{lem:Fam1PfSecondLemma} $\tau_g^{(1)}$ is regular at the zeros of $\Pt_2$ that are not zeros of  $R_2$.
\end{lemma}

\begin{proof}
		Let $B$ be a zero of $\Pt_2$ which is not a zero of $R_2$. Note that in this case we can write $\tau^{(1)}_g$ 
		as
		\begin{equation}\label{eq:WgnIreg}
		 \tau^{(1)}_g = [\hbar^{2g}]\sum_{r=0}^\infty\sum_{j=1}^\infty D_1^{j-1}\left(\frac{L^j_{r,1}}{Q_1} \cdot
			\mathrm{reg}_r\right),
		\end{equation}
		where $L_r$ is given by \eqref{eq:LrI} and $\mathrm{reg}_r$ is some expression regular in $z_1$ at $z_1=B$.
		
		From the definition of $\Pt_2$ and the conditions of generality, and since $B$ is not a zero of $R_2$, there exists exactly one root $A$ of $P_2$ such that $B$ is a root of $R_1(z)-A\, R_2(z)$. Then note that
		\begin{equation}
			e^{v\left(\cS(v\,\hbar\,\partial_y)-1\right)\frac{P_1(y)}{P_2(y)}}e^{v\left(\frac{\cS(v\,\hbar\,\partial_y)}{\cS(\hbar\,\partial_y)}-1\right)\log\left(\frac{P_3(y)}{P_4(y)}\right)} = e^{v\left(\frac{\cS(v\,\hbar\,\partial_y)}{\cS(\hbar\,\partial_y)}-1\right)\log\left(y-A\right)} \cdot \mathrm{reg},
		\end{equation}
		where $\mathrm{reg}$ is regular in $y$ at $A$, and the pole in $z_1$ at $B$ in the whole expression for 
		$\tau^{(1)}_g$
		can only come from the $e^{v\left(\frac{\cS(v\,\hbar\,\partial_y)}{\cS(\hbar\,\partial_y)}-1\right)\log\left(y-A\right)}$ part (after the substitution $y=y(z_1)$). According to Lemma \ref{lem:cvy}, we can then rewrite
		\begin{align}\label{eq:Lv1}	
			& (\partial_y+v\,\psi'(y))^r e^{v\left(\cS(v\,\hbar\,\partial_y)-1\right)P_1(y)
			}e^{v\left(\frac{\cS(v\,\hbar\,\partial_y)}{\cS(\hbar\,\partial_y)}-1\right)\log\left(\frac{P_2(y)}{P_3(y)}
				\right)} \\ \notag
			& = (\partial_y+v\,\psi'(y))^r\left(\sum_{k=1}^\infty\hbar^{2k}v(v-1)(v-2)\cdots(v-2k+1) \dfrac{p_{k}(v)}{(y-A)^{2k}}+\widetilde{\mathrm{reg}}\right),
		\end{align}
		where $p_k(v)$ are some polynomials in $v$ and $\widetilde{\mathrm{reg}}$ is regular in $y$ at $A$. Note that
		\begin{equation}
			\psi'(y) = \dfrac{1}{y-A}+\mathrm{reg}^\psi_A,
		\end{equation}
		where $\mathrm{reg}^\psi_A$ is regular in $y$ at $A$. Thus
		\begin{equation}
			(\partial_y+v\,\psi'(y))\dfrac{1}{(y-A)^{2k}}= (v-2k) \dfrac{1}{(y-A)^{2k+1}} + O\left(\dfrac{1}{(y-A)^{2k}}\right).
		\end{equation}
		This means that we can rewrite \eqref{eq:Lv1} as
		\begin{equation}
		 \sum_{k=1}^\infty\hbar^{2k}\sum_{l=1}^\infty v(v-1)(v-2)\cdots(v-l+1) \dfrac{\tilde{p}_{k,l}(v)}{(y-A)^{l}}+\widetilde{\widetilde{\mathrm{reg}}},
		\end{equation}
		where $\tilde{p}_{k,l}(v)$ are some polynomials in $v$ and $\widetilde{\widetilde{\mathrm{reg}}}$ is regular in $y$ at $A$. Now let us plug $y=y(z_1)$ into this expression and substitute it into \eqref{eq:WgnIreg}. Note that for $z_1\rightarrow B$
		\begin{equation}
			\dfrac{1}{(y(z_1)-A)^l} = 	\dfrac{1}{\left(\frac{R_1(z_1)}{R_2(z_1)}-A\right)^l}=\dfrac{C}{(z_1-B)^l}+O\left(\dfrac{1}{(z_1-B)^{l-1}}\right)
		\end{equation}
		for some constant $C$ and $Q_1^{-1}$ has a simple zero at $z_1=B$.
		This means, taking into account that $1/Q_1 \sim (z_1-B)$ for $z_1\rightarrow B$, that Equation~\eqref{eq:WgnIreg} can be rewritten as
		\begin{align}
			&\tau^{(1)}_g
			= 
			 [\hbar^{2g}]\sum_{r=0}^\infty\sum_{j=1}^\infty D_1^{j-1}
			[v^j]\, \left(v\sum_{k=2}^\infty\hbar^{k}\sum_{l=1}^\infty (v-1)(v-2)\cdots(v-l+1) \dfrac{\eta_{r,k,l}(v)}{(z_1-B)^{l-1}}+\widetilde{\mathrm{reg}}_r\right),
		\end{align}
		where $\eta_{r,k,l}$ are some expressions polynomial in $v$
		and $\widetilde{\mathrm{reg}}_r$ is regular in $z_1$ at $B$.
		Taking the sum over $j$ we can rewrite this expression as
		\begin{equation}\label{eq:WgnID}
			\tau^{(1)}_g =
			[\hbar^{2g}]\sum_{r=0}^\infty \left(\sum_{k=2}^\infty\hbar^{k}\sum_{l=1}^\infty \eta_{r,k,l}(D_1)\cdot (D_1-1)(D_1-2)\cdots(D_1-l+1) \dfrac{1}{(z_1-B)^{l-1}}+\widetilde{\mathrm{reg}}_r\right).
		\end{equation}
		%
		Now note that
		\begin{align}
		D_1 = \left(-(z_1-B) + O\left((z_1-B)^2\right)\right)	 \partial_{z_1}
		\end{align}
		at $z_1\rightarrow B$, and, therefore, by an easy inductive argument,
		\begin{align}
			(D_1-1)(D_1-2)\cdots(D_1-l+1) \dfrac{1}{(z_1-B)^{l-1}}
		\end{align}
		is regular at $z_1\rightarrow B$ for any $l\geq 1$. This implies that \eqref{eq:WgnID} is regular at $z_1\rightarrow B$.
\end{proof}
		
\begin{lemma} $\tau^{(1)}_g$ is regular at the zeros of $\Pt_3$ that are not zeros of  $R_2$.
\end{lemma}

\begin{proof} The proof of this lemma is completely analogous to the proof of Lemma~\ref{lem:Fam1PfSecondLemma}, up to adjustment of a few signs in the computation. 
\end{proof}
	
\begin{lemma} \label{lem:Fam1-tau1-infty}
	$\tau^{(1)}_g$ is regular at $z=\infty$.
\end{lemma}

\begin{proof}
	
	Note that the following calculations of degrees are correct due to the condition of generality which we imposed on all polynomials in the consideration.
	Let us consider two cases. First, assume that $e_1\leq e_2$. Then the degrees of the poles of $\check Q$ and $Q$ at $z\rightarrow \infty$ are given by
	\begin{align}
		\deg \check{Q}(z) &= e_1(d_2+d_3)+e_2(1+d_1);\\
		\deg Q(z) &= 0.
	\end{align}
	In this case it is then clear that all parts of $\tau^{(1)}_g$ are regular at $z_1\rightarrow \infty$.
	
	Now let $e_1 > e_2$.
	The degrees of the poles of $\check Q$ and $Q$ at $z\rightarrow \infty$ are given by
	\begin{align}\label{eq:QcdegI}
		\deg \check{Q}(z) &= e_1(d_1+d_2+d_3)+e_2 ;\\\label{eq:QdegI}
		\deg Q(z) &= (e_1-e_2)d_1.
	\end{align}
	We recall that $\tau^{(1)}_g$ is equal to the coefficient of $\hbar^{2g-1}$ in the expansion of
	\begin{align} \label{eq:infinity-expr1-orderpole-B}
		& \sum_{r=0}^\infty\sum_{j=1}^\infty D_1^{j-1}\Bigg([v^j]
		(\partial_y+v\,\psi'(y))^r e^{v\left(\cS(v\,\hbar\,\partial_y)-1\right)P_1(y)
		}e^{v\left(\frac{\cS(v\,\hbar\,\partial_y)}{\cS(\hbar\,\partial_y)}-1\right)\log\left(\frac{P_2(y)}{P_3(y)}
			\right)}\Big|_{y=\frac{R_1(z_1)}{R_2(z_1)}}\times \\
		\notag & \hspace{3cm} \frac{1}{Q_1}
		[u_1^r] \frac{e^{u_1\left(\cS(u_1\,\hbar\,Q_1\,D_1)-1\right)\frac{R_1(z_1)}{R_2(z_1)} }}{u_1\hbar\,\cS(u_1\,\hbar)}\Bigg)
	\end{align}	
 and count the order of pole of this expression at $z\to \infty$.
The count follows exactly the same scheme as in the proof of Lemma~\ref{lem:R2-Fam1}. We begin with a few observations:
\begin{itemize}
	\item The factor $e^{v\left(\frac{\cS(v\,\hbar\,\partial_y)}{\cS(\hbar\,\partial_y)}-1\right)\log\left(\frac{P_2(y)}{P_3(y)}
		\right)}|_{y=\frac{R_1(z_1)}{R_2(z_1)}}$ has no pole at $z_1\to\infty$.
	\item The factor $Q_1^{-1}$ has zero of order $d_1(e_1-e_2)$ at $z_1\to \infty$ and each application of $D_1=Q_1^{-1}z_1\partial_{z_1}$ decreases the degree of the pole at $z_1\to \infty$ by $d_1(e_1-e_2)$. The total effect of the factor $Q_1^{-1}$ in the middle of the formula and $D_1^{j-1}$ is the decrease of the order of pole by $jd_1(e_1-e_2)$.
\end{itemize}
Hence, the order of pole of~\eqref{eq:infinity-expr1-orderpole-B} at $z_1\to \infty$ is equal to the order of pole at $z_1\to \infty$ of the following expression:
\begin{align} \label{eq:infinity-expr2-orderpole-B}
	& \sum_{r=0}^\infty\Bigg(
	(\partial_y+v\,\psi'(y))^r e^{v\left(\cS(v\,\hbar\,\partial_y)-1\right)P_1(y)
	}\Big|_{y=\frac{R_1(z_1)}{R_2(z_1)}}
	[u_1^r] \frac{e^{u_1\left(\cS(u_1\,\hbar\,Q_1\,D_1)-1\right)\frac{R_1(z_1)}{R_2(z_1)} }}{u_1\hbar\,\cS(u_1\,\hbar)}\Bigg)\Big|'_{v=z_1^{d_1(e_2-e_1)}},
\end{align}	
where by $|'$ we mean that we only select the terms with $\deg v \geq 1$ before the substitution $v=z_1^{d_1(e_2-e_1)}$.
Note also that
\begin{itemize}
	\item Since $y(z_1)$ has a pole of order $e_1-e_2$ at $z_1\to\infty$, each $\partial_y$ decreases the order of pole in the resulting expression by $e_1-e_2$.
	\item Multiplication by $\psi'(y)$ increases the order of pole by $(e_1-e_1)(d_1-1)$.
\end{itemize}	
Taking into account these two observations and that each $v$ factor decreases the order of pole by $d_1(e_1-e_2)$, we see that each application of the operator $\partial_y + v\psi'(y)$ decreases the order of pole in the resulting expression by $e_1-e_2$. Hence, the order of pole of~\eqref{eq:infinity-expr2-orderpole-B} at $z_1\to \infty$ is equal to the order of pole at $z_1\to \infty$ of the following expression:
\begin{align} \label{eq:infinity-expr3-orderpole-B}
	& \Bigg(
	e^{v\left(\cS(v\,\hbar\,\partial_y)-1\right)P_1(y)
	}\Big|_{y=\frac{R_1(z_1)}{R_2(z_1)}}
	\frac{e^{u_1\left(\cS(u_1\,\hbar\,Q_1\,D_1)-1\right)\frac{R_1(z_1)}{R_2(z_1)} }}{u_1\hbar\,\cS(u_1\,\hbar)}\Bigg)\Bigg|''_{\substack{v=z_1^{d_1(e_2-e_1)} \\ u_1=z_1^{e_2-e_1}}},
\end{align}	
where by $|''$ we mean that we only select the terms with $\deg v \geq 1$ and regular in $u_1$ before the substitutions $v=z_1^{d_1(e_2-e_1)}$, $u=z_1^{e_2-e_1}$.

In this expression, in the first exponent in $\cS(v\,\hbar\,\partial_y)-1$ each $v\,\hbar\,\partial_y$ does not increase the order of the pole at $z_1=B$ (in fact, it even decreases it by $(d_1+1)(e_1-e_2)$); since $vP_1(y)$ has no pole at $z_1=B$ this means that the whole first exponential is regular.
In the second exponent, in $\left(\cS(u_1\,\hbar\,z_1\,\partial_{z_1})-1\right)$ each $u_1\,\hbar\,z_1\,\partial_{z_1}$ does not increase the order of the pole at $z_1=B$ (in fact, it even decreases it by $e_1-e_2$); since $u_1 R_1(z_1) / R_2(z_1)$ has no pole at $z_1=B$ this means that the whole second exponential is regular.
Thus the only possible pole at $z_1\to \infty$ in this expression can occur from the pole in $u_1$ at $u_1=0$ but this pole is removed by the condition that we only select the non-negative powers of $u_1$.

So,  \eqref{eq:infinity-expr3-orderpole-B} is regular at $z_1\to\infty$, and, therefore, $\tau^{(1)}_g$ is regular at $z_1\to\infty$ as well.
\end{proof}	

\subsubsection{Part 2 of the proof: $\tau^{(2)}_{g}+\tau^{(3)}_{g}$} \label{prp:2}
	Now let us show that $\tau^{(2)}_{g}+\tau^{(3)}_{g}\in\Theta$.
	We proceed analogously to Part 1 of the proof presented in Section~\ref{prp:1}. In a sequence of lemmata below we show that there are no poles at the zeros of $R_2$, $\Pt_2$, $\Pt_3$, and $\infty$. As in the case of $\tau^{(1)}_g$ in Section~\ref{prp:1}, this is sufficient
	in order to conclude that $\tau^{(2)}_{g}+\tau^{(3)}_{g}\in\Theta$.
	
	Combining this with the result of Part 1 of the proof given in Section~\ref{prp:1},  we obtain that $H_{g,1}=\tau^{(1)}_{g}+ \tau^{(2)}_{g}+\tau^{(3)}_{g}+\mathrm{const}\in\Theta$. With the explanation in the beginning of Section~\ref{sec:ProjPropProofs}, this completes the proof of Theorem~\ref{thm:Xi-Family1}.
	
\begin{lemma}\label{lem:R2-Fam1-tau23} $\tau_g^{(2)}+\tau_g^{(3)}$ is regular at the zeros of $R_2(z)$.
\end{lemma}

\begin{proof}
		Let $B$ be a zero of $R_2(z_1)$.
		Note that  $\tau^{(3)}_{g}$ is regular at $B$, so we only have to check that $\tau^{(2)}_{g}$ is regular at $z_1=B$ as well. To this end we have to count the order of the pole of
		\begin{equation}
			\sum_{j=1}^\infty D_1^{j-1} [v^{j+1}]
			e^{v\left(\cS(v\,\hbar\,\partial_y)-1\right)P_1(y)
			}e^{v\left(\frac{\cS(v\,\hbar\,\partial_y)}{\cS(\hbar\,\partial_y)}-1\right)\log\left(\frac{P_2(y)}{P_3(y)}
				\right)}\Big|_{y=\frac{R_1(z_1)}{R_2(z_1)}} D_1\dfrac{R_1(z_1)}{R_2(z_1)}
		\end{equation}
		at $z_1=B$. Repeating \emph{mutatis mutandis} the argument in the proof of Lemma~\ref{lem:R2-Fam1} (up to a slight difference in that we do not have a $1/Q_1$ factor but instead have an extra $D_1$ factor now), we see that the order of pole of this expression at $z_1=B$ is equal to the order of pole at $z_1=B$ of the following expression:
		\begin{align}
		& \Bigg(\frac{1}{v}	e^{v\left(\cS(v\,\hbar\,\partial_y)-1\right)P_1(y)
			}\Big|_{y=\frac{R_1(z_1)}{R_2(z_1)}} \dfrac{R_1(z_1)}{R_2(z_1)} \Bigg)
			\Big|'_{v=(z_1-B)^{d_1}} \\ \notag
			& = \Bigg(\frac{1}{v}\left(	e^{v\left(\cS(v\,\hbar\,\partial_y)-1\right)P_1(y)}-1\right) \Bigg)
				\Big|_{v=(z_1-B)^{d_1}}
			\dfrac{R_1(z_1)}{R_2(z_1)} \,.
		\end{align}
		In the first line, again, prime stands for taking only terms with degree $\geq 2$ in the $v$-expansion before the substitution.
		 In the resulting expression, in the first exponent $vP_1(y)$ has no pole at $z_1=B$, and in the subsequent application of $\cS(v\,\hbar\,\partial_y)-1$ each $v\,\hbar\,\partial_y$ increases the order of zero by $d_1+1$. Note that $v\,\hbar\,\partial_y$ is applied at least twice. Thus the first factor has a zero of order at least $2(d_1+1)-d_1=d_1+2$ at $z_1=B$. The second factor, $R_1(z_1)/R_2(z_1)$, has a simple pole at $z_1=B$. Hence the whole expression is regular at $z_1=B$.
\end{proof}

\begin{lemma} \label{lem:Fam1PfSecondLemma-tau23} $\tau_g^{(2)}+\tau_g^{(3)}$ is regular at the zeros of $\Pt_2$ that are not zeros of  $R_2$.
\end{lemma}

\begin{proof}
\ Let B be a zero of $\Pt_2$ which is not a zero of $R_2$. 
Let us 
show that $\tau^{(2)}_{g}+\tau^{(3)}_{g}$ when added together do not have a pole at $B$.

Note that the expression $\tau^{(2)}_{g}+\tau^{(3)}_{g}$ has a pole at $z_1=B$ if and only if $D_1(\tau^{(2)}_{g}+\tau^{(3)}_{g})$ has a pole there. Indeed, as	$Q_1 \sim (z_1-B)^{-1}$ for $z_1\rightarrow B$, the operator $D_1=Q_1^{-1}z_1\partial_{z_1}$ preserves the degree of the pole at $B$ for any function.  We have:
\begin{equation}
	D_1(\tau^{(2)}_{g}+\tau^{(3)}_{g}) = [\hbar^{2g}]\sum_{j=0}^\infty D_1^{j}L^{j+1}_{0,1}D_1\dfrac{R_1(z_1)}{R_2(z_1)}.
\end{equation}

From that point the proof becomes analogous to the proof of Lemma~\ref{lem:Fam1PfSecondLemma}. There exists exactly one root $A$ of $P_2$ such that $B$ is a root of $R_1(z)-A\, R_2(z)$. Lemma \ref{lem:cvy} implies that
\begin{align}
	& e^{v\left(\cS(v\,\hbar\,\partial_y)-1\right)P_1(y)
	}e^{v\left(\frac{\cS(v\,\hbar\,\partial_y)}{\cS(\hbar\,\partial_y)}-1\right)\log\left(\frac{P_2(y)}{P_3(y)}
		\right)}\Big|_{y=\frac{R_1(z_1)}{R_2(z_1)}}
	\\ \notag
	& =\sum_{k=1}^\infty\hbar^{2k}v(v-1)(v-2)\cdots(v-2k+1) \dfrac{p_{k}(v)}{(y-A)^{2k}}+\mathrm{reg},
\end{align}
where $\mathrm{reg}$ is regular in $y$ at $y=A$. We also have
\begin{equation}
D_1\frac{R_1(z_1)}{R_2(z_1)} = O(z_1-B).
\end{equation}
Thus
\begin{align} \label{eq:Fam1Tau23SecondLemma-LastExpr}
	&
	D_1(\tau^{(2)}_{g}+\tau^{(3)}_{g}) =
	\\
	&
	\notag
	\sum_{j=0}^\infty D_1^{j}[v^{j+1}]\left(v(v-1)(v-2)\cdots(v-2g+1) \dfrac{\tilde{p}_{g}(v)}{(z_1-B)^{2g-1}}+\mathrm{reg_1}\right)\,\mathrm{reg_2},
\end{align}
where $\tilde{p}_g$, $\mathrm{reg_1}$ and $\mathrm{reg_2}$ are some expressions regular in $z_1$ at $z_1=B$.

By the same argument as in the proof of Lemma~\ref{lem:Fam1PfSecondLemma}, the right hand side of~\eqref{eq:Fam1Tau23SecondLemma-LastExpr} is regular at $z_1=B$.
\end{proof}

\begin{lemma}  $\tau_g^{(2)}+\tau_g^{(3)}$ is regular at the zeros of $\Pt_3$ that are not zeros of  $R_2$.
\end{lemma}

\begin{proof} The proof of this lemma is completely analogous to the proof of Lemmata~\ref{lem:Fam1PfSecondLemma} and~\ref{lem:Fam1PfSecondLemma-tau23}, up to adjustment  of a few signs in the computation.
\end{proof}

\begin{lemma}\label{lem:fam1tau23inf} $\tau^{(2)}_{g}+\tau^{(3)}_{g}$ is regular at $z_1\to\infty$. 	
\end{lemma}

\begin{proof} It is evident that $\tau^{(3)}_{g}$ is regular at infinity. Let us check this for $\tau^{(2)}_{g}$. As in the proof of Lemma~\ref{lem:Fam1-tau1-infty}, we have to consider two cases: $e_1\leq e_2$ and $e_1>e_2$. In the first case, all parts of the expression are manifestly regular.
	
Note again that the following calculations of degrees are correct due to the condition of generality which we imposed on all polynomials in the consideration.
Let $e_1>e_2$.
Recall from \eqref{eq:QcdegI}-\eqref{eq:QdegI} that for $z_1\rightarrow \infty$
\begin{align}
	\deg \check{Q}(z) &= e_1(d_1+d_2+d_3)+e_2; \\
	\deg Q(z) &= (e_1-e_2)d_1.
\end{align}	

We recall the expression for $\tau^{(2)}_g$, which is equal to the coefficient of $\hbar^{2g}$ in
\begin{equation}
	\sum_{j=1}^\infty D_1^{j-1} [v^{j+1}]
	e^{v\left(\cS(v\,\hbar\,\partial_y)-1\right)P_1(y)
	}e^{v\left(\frac{\cS(v\,\hbar\,\partial_y)}{\cS(\hbar\,\partial_y)}-1\right)\log\left(\frac{P_2(y)}{P_3(y)}
		\right)}\Big|_{y=\frac{R_1(z_1)}{R_2(z_1)}} D_1\dfrac{R_1(z_1)}{R_2(z_1)},
\end{equation}
 and count the order of pole of the latter expression at $z\to \infty$.
The count follows exactly the same scheme as in the proofs of  Lemmata~\ref{lem:Fam1-tau1-infty}, \ref{lem:R2-Fam1}, and~\ref{lem:R2-Fam1-tau23}. Repeating the same argument here, we see that the order of the pole of this expression at $z_1=\infty$ is equal to the order of the pole at $z_1=\infty$ of the following expression:
\begin{align}
	& \Bigg(\frac{1}{v}	e^{v\left(\cS(v\,\hbar\,\partial_y)-1\right)P_1(y)
	}\Big|_{y=\frac{R_1(z_1)}{R_2(z_1)}} \dfrac{R_1(z_1)}{R_2(z_1)} \Bigg)
	\Big|'_{v=z_1^{d_1(e_2-e_1)}} \\ \notag
	& = \Bigg(\frac{1}{v}	\left(e^{v\left(\cS(v\,\hbar\,\partial_y)-1\right)P_1(y)}-1\right) \Bigg)
	\Big|_{v=z_1^{d_1(e_2-e_1)}}
	\dfrac{R_1(z_1)}{R_2(z_1)} \,,
\end{align}
where prime in the first line once again means that we keep only terms with degree $\geq 2$ in the $v$-expansion prior to the substitution.
Following similar reasoning to Lemma~\ref{lem:R2-Fam1-tau23}, we see that this expression is regular.



	This completes the proof for the statement that $\tau^{(2)}_{g}+\tau^{(3)}_{g}\in \Theta$.
%
%
\end{proof}

\begin{remark}
	Note that the precise form of $\hat{\psi}$ was crucial in the argument above. Namely, any other $\hbar^2$-deformation of $\psi$ would not have got canceled in the term
\begin{equation} [\hbar^{2g}]\int_{0}^{z_1}\dfrac{Q(z)}{z}\left.\left(\dfrac{1}{\cS(\hbar\partial_y)}\hat{\psi}(y)-\psi(y)\right)\right|_{y=y(z)}Dy(z) dz,
\end{equation}
	and this would lead to a pole at infinity in the whole expression. Moreover, in Lemmata~\ref{lem:R2-Fam1-tau23} and~\ref{lem:fam1tau23inf} we used the fact that $\cS(v\,\hbar\,\partial_y)-1$ is proportional to $v^2$, while any other $\hbar^2$--deformation of $\psi$ would not have canceled the factor $\cS(\hbar\partial_y)$ in $\frac{\cS(v\,\hbar\,\partial_y)}{\cS(\hbar\partial_y)}-1$, and thus this expression would not be proportional to $v$, which leads to unwanted poles as well.
\end{remark}

\subsection{Family II: proof}  \label{sec:ProofProjFamily2}
In this section we present the proof of Theorem~\ref{thm:Xi-Family2}. The logic of the proof is exactly the same as in the proof of Theorem~\ref{thm:Xi-Family1} presented in Section~\ref{sec:ProofProjFamily1}.

	Recall that for Family II we have
	\begin{align}\label{eq:FamII-psi}
		\psi(y)&=\alpha\,y, & \hat{\psi}(\hbar^2,y)&=\alpha\,y,\\ \label{eq:FamII-y}
		y(z)&=\dfrac{R_1(z)}{R_2(z)}+ \log\left(\dfrac{R_3(z)}{R_4(z)}\right), & \hat y(\hbar^2,z)&=\dfrac{R_1(z)}{R_2(z)}+\dfrac{1}{\cS(\hbar z\partial_z)}\,\log\left(\dfrac{R_3(z)}{R_4(z)}\right).
	\end{align}
	Denote $e_i\coloneqq\deg R_i$, $i=1,2,3,4$. 
	According to Remark~\ref{rem:generality}, we assume that the poles of zeros of polynomials $R_2$, $R_3$, $R_4$ and also of $dX$ are all simple.	

\subsubsection{Specialization of the formulas for $H_{g,1}$}	As we discussed above, it is sufficient to restrict the analysis of singularities to the case of $n=1$.
For $\hat \psi$ and $\hat y$ as in~\eqref{eq:FamII-psi},~\eqref{eq:FamII-y}, we have $L_r(v,y,\hbar)=(\alpha\,v)^r$ and the operator $\overline{U}_1$ takes the form
\begin{align}	\label{eq:UiII}
	\overline{U}_{1} f&=
	\alpha\sum_{r=1}^\infty (\alpha D_1)^{r-1}
	[u_1^r] \frac{e^{u_1\left(\frac{\cS(u_1\,\hbar\,Q_1\,D_1)}{\cS(\hbar\,Q_1\,D_1)}-1\right)\log\left(\frac{R_3(z_1)}{R_4(z_1)}\right) }e^{u_1\left(\cS(u_1\,\hbar\,Q_1\,D_1)-1\right)\frac{R_1(z_1)}{R_2(z_1)} }
	}{Q_1\cdot u_1\hbar\,\cS(u_1\,\hbar)}f(u_1,z_1),
\end{align}
where $Q_1 \coloneqq Q(z_1)$ for 
\begin{align}
	\label{eq:QII}
	Q(z) &= \dfrac{\check Q(z)}{R_2^2 R_3 R_4},\\
	\label{eq:QcheckII}
	\check Q(z)&= R_2^2 R_3 R_4+\alpha z \left(R_1 R_2' R_3 R_4 - R_1' R_2 R_3 R_4   + R_2^2 R_3 R_4' - R_2^2 R_3' R_4 \right).
\end{align}
Here we have slightly abused the notation: by \eqref{eq:Qcheckdef} we defined $\check Q(z)$ to be a monic polynomial, which is not necessarily the case in \eqref{eq:QcheckII}.

Since $L_{0,1}^{j+1}$ vanishes, we have
\begin{align} \label{eq:Hg1II}	
	H_{g,1}& =[\hbar^{2g}]\Bigg(\hbar\overline{U}_1 1 +\int_{0}^{z_1}\dfrac{\hat{y}(z)-y(z)}{z}dz
	+\mathrm{const}\Bigg) \\ \nonumber
	&=[\hbar^{2g}]\hbar\overline{U}_1 1 + \left([u^{2g}]\dfrac{1}{\cS(u)}\right) (z_1\partial_{z_1})^{2g-1}\log\left(\dfrac{R_3(z_1)}{R_4(z_1)}\right) +\mathrm{const}.
\end{align}
Denote the summands of $H_{g,1}$ by
\begin{align}
	\sigma^{(1)}_g&:=[\hbar^{2g}]\hbar\overline{U}_1 1,\\
	\sigma^{(2)}_g&:=\left([u^{2g}]\dfrac{1}{\cS(u)}\right) (z_1\partial_{z_1})^{2g-1}\log\left(\dfrac{R_3(z_1)}{R_4(z_1)}\right),
\end{align}
which allows to rewrite $H_{g,1}$ as
\begin{equation}\label{eq:HtauII}
	H_{g,1} = \sigma^{(1)}_g+\sigma^{(2)}_g +\mathrm{const}.
\end{equation}

	\subsubsection{Proof of Theorem~\ref{thm:Xi-Family2}}
	\label{prp:1II}
	As we explained in the beginning of Section~\ref{sec:ProjPropProofs}, it is sufficient to prove that $\sigma^{(1)}_g+\sigma^{(2)}_g\in\Theta$. Recall \eqref{eq:UiII}. Similarly to the case of $\tau^{(1)}_g$ of Section~\ref{prp:1}, we note that the expression
	\begin{equation}		
		\left(\frac{\cS(u_1\,\hbar\,Q_1\,D_1)}{\cS(\hbar\,Q_1\,D_1)}-1\right) \log\left(\frac{R_3(z_1)}{R_4(z_1)}\right)
	\end{equation}
	is a series in $\hbar$ with coefficients given by rational functions of $z_1$. So,  $\sigma^{(1)}_g$ is manifestly a rational function in $z_1$, and its set of possible poles includes $p_1,\dots,p_N$ (the zeros of $\check Q$),  the zeros of $R_2$, $R_3$, $R_4$, and $z\to\infty$. The same applies to $\sigma^{(2)}_g$, though, in fact, it is manifestly regular at the zeros of $R_2$. In a sequence of lemmata below we show that $\sigma^{(1)}_g+\sigma^{(2)}_g$ is regular at the zeros of $R_2$, $R_3$, $R_4$, and at $z\to\infty$. Since the poles at $p_1,\dots,p_N$ are odd with respect to the deck transformations (they are generated by the iterative application of the operator $D_1$ to a function with the first order poles at $p_1,\dots,p_N$, cf.~the proof of Proposition~\ref{prop:XiThetaequiv}), we conclude that $\sigma^{(1)}_g+\sigma^{(2)}_g\in \Theta$, and this completes the proof of Theorem~\ref{thm:Xi-Family2}.
	
	\begin{lemma} \label{lem:R2-Fam2} $\sigma_g^{(1)}+\sigma^{(2)}_g$ is regular at the zeros of $R_2(z)$.
	\end{lemma}	
	\begin{proof}	
		Since $\sigma^{(2)}_g$	is regular at the zeros of $R_2(z)$, we only have to prove the regularity of $\sigma_g^{(1)}$. The proof follows the same lines as the the proof of Lemma~\ref{lem:R2-Fam1}.
		
		Let $B$ be a zero of $R_2$. We recall the expression for $\sigma_g^{(1)}$:
		\begin{equation}	\label{eq:R2-expr1-orderpole-B-II}		
		[\hbar^{2g}]\alpha\sum_{r=1}^\infty (\alpha D_1)^{r-1}\left(\frac{1}{Q_1}
		[u_1^r] \frac{e^{u_1\left(\frac{\cS(u_1\,\hbar\,Q_1\,D_1)}{\cS(\hbar\,Q_1\,D_1)}-1\right)\log\left(\frac{R_3(z_1)}{R_4(z_1)}\right) }e^{u_1\left(\cS(u_1\,\hbar\,Q_1\,D_1)-1\right)\frac{R_1(z_1)}{R_2(z_1)} }
		}{u_1\,\cS(u_1\,\hbar)}\right).
	\end{equation}
		and count its the order of the pole of at $z_1=B$.
		 Two observations are in order:
		\begin{itemize}
			\item Firstly, note that $e^{u_1\left(\frac{\cS(u_1\,\hbar\,Q_1\,D_1)}{\cS(\hbar\,Q_1\,D_1)}-1\right)\log\left(\frac{R_3(z_1)}{R_4(z_1)}\right)}$ does not contribute to the pole at $B$ in $z_1$ and can be safely ignored in this computation.
			\item Secondly, note that $Q_1^{-1}$ has zero of order $2$ at $z_1=B$ and each application of $D_1=Q_1^{-1}z_1\partial_{z_1}$ decreases the degree of the pole in $z_1$ at $B$ by $1$. The total effect of the factor $Q_1^{-1}$ and of $D_1^{r-1}$ is the decrease of the order of pole by $r +1$.
		\end{itemize}

		Therefore, the order of pole of~\eqref{eq:R2-expr1-orderpole-B-II} is equal to the order of pole at $z_1=B$ of
		\begin{align} \label{eq:R2-expr2-orderpole-B-II}	
		[\hbar^{2g}]	\frac{e^{u_1\left(\cS(u_1\,\hbar\,Q_1\,D_1)-1\right)\frac{R_1(z_1)}{R_2(z_1)} }}
			{\cS(u_1\,\hbar)}\bigg|'_{u_1=z_1-B}.
		\end{align}	
		where by $|'$ we mean that we only select the terms with $\deg u_1 \geq 2$ before the substitution $u_1=(z_1-B)$.
		Recall that $Q_1D_1=z_1\partial_{z_1}$. In the exponent, in $\left(\cS(u_1\,\hbar\,z_1\,\partial_{z_1})-1\right)$ each $u_1\,\hbar\,z_1\,\partial_{z_1}$ preserves the order of the pole at $z_1=B$. Since $u_1 R_1(z_1) / R_2(z_1)$ has no pole at $z_1=B$ this means that the whole exponential in the numerator is regular. Since $1/\cS(u_1\,\hbar)$ after the substitution $u_1=z_1-B$ is also regular at $z_1=B$, the whole expression is regular at $z_1=B$ as well.
	\end{proof}

\begin{lemma} \label{lem:R3-Fam2} $\sigma_g^{(1)}+\sigma^{(2)}_g$ is regular at  the zeros of $R_3(z)$.
\end{lemma}	
\begin{proof}
	 Let $B$ be a zero of $R_3$.
	Analogously to the similar cases for the first family, Lemmata~\ref{lem:Fam1PfSecondLemma-tau23} and \ref{lem:Fam1PfSecondLemma}, let us 
	show that $\sigma^{(1)}_{g}+\sigma^{(2)}_{g}$ when added together (i.e., in this case, just the whole $H_{g,1}$) do not have a pole at $B$.
	
	Note that $H_{g,1}$ has a pole at $z_1=B$ if and only if $D_1H_{g,1}$ has a pole there. Indeed, since $Q_1$ has a simple pole at $z_1=B$, the operator $D_1=Q_1^{-1}z_1\partial_{z_1}$ preserves the degree of the pole at $B$ for any function. In our case, we have the following formula for $D_1H_{g,1}$:
	\begin{equation}
		D_1H_{g,1} = [\hbar^{2g}]
		\sum_{r=0}^\infty (\alpha D_1)^{r}\left(
		[u_1^r] \frac{e^{u_1\left(\frac{\cS(u_1\,\hbar\,Q_1\,D_1)}{\cS(\hbar\,Q_1\,D_1)}-1\right)\log\left(\frac{R_3(z_1)}{R_4(z_1)}\right) }e^{u_1\left(\cS(u_1\,\hbar\,Q_1\,D_1)-1\right)\frac{R_1(z_1)}{R_2(z_1)} }
		}{Q_1 \cdot u_1\cS(u_1\,\hbar)}\right).
	\end{equation}
	Note that
	\begin{equation}
		\log R_3(z_1) - \log R_4(z_1) = \log(z_1-B) + \mathrm{reg},
	\end{equation}
	where $\mathrm{reg}$ is regular in $z_1$ at $B$.
	Using this and taking also into account Remark \ref{rem:generality} we can represent $D_1H_{g,1}$ as the coefficient of $\hbar^{2g}$ in the following expression:
	\begin{equation} \label{eq:IIWgnreg}
\sum_{r=0}^\infty (\alpha D_1)^{r}[u_1^r]\left((z_1-B)e^{u_1\left(\frac{\cS(u_1\,\hbar\,z_1\partial z_1)}{\cS(\hbar\,z_1 \partial z_1)}-1\right)\log\left(z_1-B\right) } 	\cdot \widetilde{\mathrm{reg}}\right),
	\end{equation}
	where $\widetilde{\mathrm{reg}}=\widetilde{\mathrm{reg}}(\hbar)$ is regular at $z_1=B$. Lemma \ref{lem:cuzdz} implies that we can rewrite Expression \eqref{eq:IIWgnreg} as
	\begin{align}	
\sum_{r=0}^\infty (\alpha D_1)^{r}[u_1^r]\left(\sum_{k=2}^\infty\hbar^k\sum_{l=1}^{\infty}(u_1-1)(u_1-2)\cdots(u_1-k+1) \dfrac{p_{k,l}(u_1)}{(z-B)^{l-1}}+\widetilde{\widetilde{\mathrm{reg}}}\right),
	\end{align}
	where all $p_{k,l}$'s are polynomials in $u_1$,
	and $\widetilde{\widetilde{\mathrm{reg}}}$ is regular in $z_1$ at $B$. Taking the sum over $r$ we can rewrite this expression as
	\begin{equation}\label{eq:alphad}
\sum_{k=2}^\infty\hbar^k\sum_{l=1}^{\infty}p_{k,l}(\alpha D_1)\cdot(\alpha D_1-1)(\alpha D_1-2)\cdots(\alpha D_1-l+1) \dfrac{1}{(z_1-B)^{l-1}}+\widetilde{\widetilde{\mathrm{reg}}},
	\end{equation}
	Using Equations~\eqref{eq:QII} and~\eqref{eq:QcheckII}, we see that
	\begin{align}
		&\alpha D_1 = \dfrac{\alpha z_1}{Q_1} \partial_{z_1} =\left(-(z_1-B) + O\left((z_1-B)^2\right)\right)	 \partial_{z_1}
	\end{align}
	for $z_1\rightarrow B$. Substituting this in place of $\alpha D_1$ in \eqref{eq:alphad} in the last bracket before the polar term and applying it to the polar term, we get:
	\begin{align}
		&(\alpha D_1-l+1) \dfrac{1}{(z-B)^{l-1}} = \left(\left(-(z_1-B)+ O\left((z_1-B)^2\right)\right)\partial_{z_1} -l+1\right)\dfrac{1}{(z-B)^{l-1}} \\ \nonumber
		&= (-(-l+1)-l+1)\dfrac{1}{(z-B)^{l-1}}+\dfrac{O\left((z_1-B)^2\right)}{(z-B)^l} = O\left(\dfrac{1}{(z-B)^{l-2}}\right).
	\end{align}
	Repeating the same computation for the factors $(\alpha D_1 -i )$, $i=l-2,l-3,\dots,1$, in~\eqref{eq:alphad}, we see the whole polar part in $z_1$ at $B$ gets canceled. Thus, Expression~\eqref{eq:alphad} is regular at $z_1=B$, and therefore $D_1H_{g,1}$ and $H_{g,1}$ are regular at $B$ as well.
\end{proof}

\begin{lemma} \label{lem:Fam2R4}
	$\sigma_g^{(1)}+\sigma^{(2)}_g$ has no poles at the zeros of $R_4(z)$.
\end{lemma}
\begin{proof}
	The proof of this lemma is completely analogous to the proof of Lemma~\ref{lem:R3-Fam2}, up to adjustment of a few signs in the computation.
\end{proof}
	
\begin{lemma} \label{lem:Fam2-infty}
	$\sigma^{(1)}_g+\sigma^{(2)}_g$ is regular at $z_1\to\infty$.
\end{lemma}
\begin{proof}
	Since $\sigma^{(2)}_g$ is regular at infinity, we only have to prove the regularity of $\sigma^{(1)}_g$.
	
	Note again that the following calculations of degrees are correct due to the condition of generality which we imposed on all polynomials in the consideration.
	Let us consider two cases. First, assume that $e_1\leq e_2$. Then the degrees of the poles of $\check Q$ and $Q$ at $z\rightarrow \infty$ are given by
	\begin{align}
		\deg \check{Q}(z) &=2e_2+e_3+e4;\\
		\deg Q(z) &= 0.
	\end{align}
	In this case it is then clear that all factors of $\sigma^{(1)}_g$ are regular at $z_1\rightarrow \infty$.
	
	Now let $e_1 > e_2$.
	The degrees of the poles of $\check Q$ and $Q$ at $z\rightarrow \infty$ are given by
	\begin{align}\label{eq:QcdegII}
		\deg \check{Q}(z) &= e_1+e_2+e_3+e_4 ;\\\label{eq:QdegII}
		\deg Q(z) &= e_1-e_2.
	\end{align}
	We recall the expression for $\sigma^{(1)}_g$, which is the coefficient of $\hbar^{2g}$ of
	\begin{align} \label{eq:infinity-expr1-orderpole-B-II}
		\alpha\sum_{r=1}^\infty (\alpha D_1)^{r-1}\left(\frac{1}{Q_1}
		[u_1^r] \frac{e^{u_1\left(\frac{\cS(u_1\,\hbar\,Q_1\,D_1)}{\cS(\hbar\,Q_1\,D_1)}-1\right)\log\left(\frac{R_3(z_1)}{R_4(z_1)}\right) }e^{u_1\left(\cS(u_1\,\hbar\,Q_1\,D_1)-1\right)\frac{R_1(z_1)}{R_2(z_1)} }
		}{u_1\cS(u_1\,\hbar)}\right),
	\end{align}	
	and count its order of pole of this expression at $z_1\to \infty$.
	The count follows exactly the same scheme as in the proof of Lemma~\ref{lem:R2-Fam2}.

	Two observations are in order:
	\begin{itemize}
		\item Firstly, note that $e^{u_1\left(\frac{\cS(u_1\,\hbar\,Q_1\,D_1)}{\cS(\hbar\,Q_1\,D_1)}-1\right)\log\left(\frac{R_3(z_1)}{R_4(z_1)}\right)}$ does not contribute to the pole at infinity and can be safely ignored in this computation.
		\item Secondly, note that $Q_1^{-1}$ has zero of order $e_1-e_2$ at $z_1=\infty$ and each application of $D_1=Q_1^{-1}z_1\partial_{z_1}$ decreases the degree of the pole in $z_1$ at infinity by $e_1-e_2$. The total effect of the factor $Q_1^{-1}$ and of $D_1^{r-1}$ is the decrease of the order of pole at infinity by $r(e_1-e_2)$.
	\end{itemize}

	Therefore, the order of pole of~\eqref{eq:infinity-expr1-orderpole-B-II} is equal to the order of pole at $z_1=\infty$ of
	\begin{align} \label{eq:infinity-expr2-orderpole-B-II}	
		&\left(\frac{e^{u_1\left(\cS(u_1\,\hbar\,Q_1\,D_1)-1\right)\frac{R_1(z_1)}{R_2(z_1)} }}
		{u_1\cS(u_1\,\hbar)}\right)\Big|'_{u_1=z_1^{e_2-e_1}}
		=\left(\frac{e^{u_1\left(\cS(u_1\,\hbar\,Q_1\,D_1)-1\right)\frac{R_1(z_1)}{R_2(z_1)} }-1}
		{u_1\cS(u_1\,\hbar)}\right)\Big|_{u_1=z_1^{e_2-e_1}}
	\end{align}
	where
	 by $|'$ we mean that we only select the terms with degree $\geq 1$ in the $u_1$-expansion before the substitution $u_1=z_1^{e_2-e_1}$.
	Recall that $Q_1D_1=z_1\partial_{z_1}$. In the exponent, in $\left(\cS(u_1\,\hbar\,z_1\,\partial_{z_1})-1\right)$ each $u_1\,\hbar\,z_1\,\partial_{z_1}$ has a zero at $z_1=\infty$ of degree $e_1-e_2$. Since $u_1 R_1(z_1) / R_2(z_1)$ has no pole at $z_1=\infty$ this means that the whole exponential is regular, and that the numerator on the right hand side of this formula has a zero of order $3(e_1-e_2)$ at $z_1\to\infty$. Note also that $(u_1\cS(u_1\,\hbar))^{-1}$ has a pole of order $e_2-e_1$ at $z_1\to \infty$. Therefore, the whole expression~\eqref{eq:infinity-expr2-orderpole-B-II} is regular, and this proves that $\sigma_g^{(1)}$ is regular at $z_1\to\infty$.
\end{proof}

\begin{remark}
	Note that the precise form of $\hat{y}$ was crucial in these computations. Namely, for any other $h^2$- deformation of $y$ we would not have cancellation of the pole in the arguments in Lemmata~\ref{lem:R3-Fam2} and~\ref{lem:Fam2R4}.
\end{remark}

\section{Topological recursion and its applications} \label{sec:TopoRecAppl}
We are ready to prove the General Principle discussed in the Introduction restricted to the cases of Families I and II of Table \ref{tab:families}. Even when restricted to these cases the statement is still quite general, and in particular covers all known cases of topological recursion for Hurwitz-type problems, and much more (see Section \ref{sec:appl} below).

In order to formulate the statements in the most general form, we first need to introduce the Bouchard-Eynard recursion, which extends the topological recursion \eqref{eq:toprec} to the cases when the zeros of $dx$ are not necessarily simple.

\subsection{Topological recursion}\label{sec:TopoRec}
\begin{definition}
	The \emph{Bouchard-Eynard recursion} \cite{BE13} is defined as follows. Let spectral curve $(\Sigma, x, y, B)$ be as in the definition of topological recursion, just without the restriction on the simplicity of zeros of $dx$. Let $p_1,\dots,p_N$ be these zeros, and let $m_a\in\mathbb{Z}_{\geq 2}$ be such that $m_a-1$ is the multiplicity of the zero $p_a$, $a=1,\dots,N$. For $\zeta\neq p_a$ being a point in the vicinity of $p_a$, denote by $\zeta_1, \zeta_2,\dots,\zeta_{m_a}$ the points in $x^{-1}(x(\zeta))$ which fall into the vicinity of $p_a$, with $\zeta_1=\zeta$. We regard $\zeta_i$ as a function in $\zeta_1$.
		
	
	The  Bouchard-Eynard  recursion reads:
	\begin{align} \label{eq:BErec}
		& \omega_{g,n+1}(z_{[n]},z_{n+1}):=
		\\ \notag
		& 
		-\sum_{a=1}^N\, \res_{\zeta_1\rightarrow p_a}\,
		\sum_{I\subset [m_a]\setminus\{1\}} \frac{\int\limits_{p_a}^{\zeta_1} \omega_{0,2}(\cdot, z_{n+1})}{\prod\limits_{i\in I} (\omega_{0,1}(\zeta_i)-\omega_{0,1}(\zeta_1))}
		\sum_{\substack{
				J\,\vdash\, I\cup \{1\} \\
				\sqcup_{i=1}^{\ell(J)} N_i = [n] \\
				\!\!\! \!
				\sum_{i=1}^{\ell(J)} g_i = g +\ell(J) -|I|-1
		}}^{\text{no}\, (0,1)}
		\!\!\! \!
		\prod_{i=1}^{\ell(J)} \omega_{g_i,|J_i|+|N_i|}(\zeta_{J_i},z_{N_i}),
	\end{align}
	where we forbid in the second sum the choices where we have $(g_i,|J_i|+|N_i|)=(0,1)$, by $[m_a]$ (resp., $[n]$) we denote the set $\{1,\dots,m_a\}$ (resp., $\{1,\dots,n\}$), and the notation of the form $\zeta_{J}$ means all variables $\zeta_i$, $i\in J$. Though it might not be obvious at the first glance, the right hand side of equation~\eqref{eq:BErec} does not depend on the way we label $\zeta_2,\dots,\zeta_m$, and it is a symmetric function of $z_1,\dots,z_{n+1}$.
\end{definition}

 \begin{remark}
 	Note that when all zeros of $dx$ are simple, Bouchard-Eynard recursion coincides with the topological recursion.
 \end{remark}

\begin{remark}\label{rem:multiple-poles} This definition is compatible with taking limits of spectral curves~\cite{BBCKS}. However, once a critical point of $x$ tends to a pole in the limit, one has to extend this definition and add to the sum of residues the contributions from the multiple poles of $x$ as well.  
\end{remark}

Now we are ready to formulate the main theorems regarding the topological recursion for the general weighted Hurwitz numbers.

\begin{theorem}[Family~I from Table~\ref{tab:families}] \label{thm:Family1TR}
	Let
	\begin{align}
		\hat\psi(\hbar^2,y) &=  \cS(\hbar \partial_y)P_1(y)+\log P_2(y) -\log P_3(y),\\
		\hat y (\hbar^2,z) &= R_1(z)/R_2(z),
	\end{align}
	where $P_i$, $i=1,2,3$, and $R_i$, $i=1,2$, are some polynomials such that 
	the natural analytic assumptions of Definition~\ref{def:naa} are satisfied.
	
	Consider the spectral curve
\begin{equation}	
	(\mathbb{C}\mathrm{P}^1,
	x(z)=\log z - \psi(y(z)),
	y(z),
	B=dz_1dz_2/(z_1-z_2)^2
	).
\end{equation}

Then topological recursion \eqref{eq:toprec} (or, in general, when zeros or poles of $dx$ are not necessarily simple, the Bouchard-Eynard recursion \eqref{eq:BErec}, including possibly the terms coming from Remark~\ref{rem:multiple-poles}) applied on this spectral curve
returns the $n$-differentials $\omega_{g,n}$, $g\geq 0$, $n\geq 1$, whose genus expansion in the variables $X_i=\exp(x(z_i))$, $i=1,\dots,n$ at $z_1=\cdots =z_n=0$ is given by
\begin{equation}
	\omega_{g,n} = d_1\cdots d_n H_{g,n},
\end{equation}
where the $H_{g,n}$ are the $n$-point functions of a tau function $Z_{\hat\psi,\hat y}$	of \eqref{eq:ParitionFunctionHBAR} for these particular $\hat \psi$, $\hat y$:
\begin{equation}
	H_{g,n}(z_1,\dots,z_n):=[\hbar^{2g-2+n}]\sum_{k_1,\dots,k_n=1}^\infty\frac{\partial^n \log Z_{\hat\psi,\hat y}}{\partial p_{k_1}\dots\partial p_{k_n}}\Bigm|_{p_1=p_2=\cdots=0}X_1^{k_1}\dots X_n^{k_n},
\end{equation}
where $X_i:=\exp(x(z_i))$.

In other words, $\omega_{g,n}$'s produced by topological recursion procedure \eqref{eq:toprec} (or, for the case of non-simple zeros or poles of $dx$, the Bouchard-Eynard recursion \eqref{eq:BErec}, including possibly the terms coming from Remark~\ref{rem:multiple-poles}) for this spectral curve are the generating functions for the respective weighted Hurwitz numbers
\begin{equation}
	h_{g;k_1,\dots,k_n}\coloneqq [\hbar^{2g-2+n}] \frac{\partial^n \log Z_{\hat\psi,\hat y}}{\prod_{i=1}^n \partial p_{k_i}}\Bigg|_{p_1=p_2=\cdots=0}
\end{equation}
in the following sense:
\begin{equation}\label{eq:TRtoprove}
	\omega_{g,n} = \sum_{k_1,\dots,k_n=1}^\infty h_{g;k_1,\dots,k_n} \prod_{i=1}^n \left(k_i X_i^{k_i-1} d X_i\right).
\end{equation}

 \end{theorem}

\begin{theorem}[Family II from Table \ref{tab:families}]  \label{thm:Family2TR}
	The statement of Theorem \ref{thm:Family1TR} holds for 	
	\begin{align}
		\hat\psi(\hbar^2,y) &=  \alpha y,\\
		\hat y (\hbar^2,z) &= R_1(z)/R_2(z)+
		\cS(\hbar z\partial_z)^{-1} (\log R_3(z)- \log R_4(z)),
	\end{align}
where $\alpha\neq 0$ is a number and $R_i$, $i=1,2,3,4$, are some polynomials 
such that the natural analytic assumptions of Definition~\ref{def:naa} are satisfied. 
\end{theorem}
\begin{proof}[Proof of Theorems \ref{thm:Family1TR} and \ref{thm:Family2TR}]
	Note that Theorems \ref{th:linloop} and \ref{theor:higher
	} hold for all $\hbar$-deformed KP tau-functions of hypergeometric type satisfying the \emph{natural analytic assumptions} of Definition~\ref{def:naa}, which are required in Theorems \ref{thm:Family1TR} and \ref{thm:Family2TR}, and satisfying the additional condition that zeros of $dx$ are simple. Thus the linear and quadratic loop equations hold for the cases with simple zeros of $dx$.
	
	If the respective generality conditions of Theorems \ref{thm:Xi-Family1} and \ref{thm:Xi-Family2} are satisfied (in which case the zeros of $dx$ are simple and we are in the realm of the ordinary topological recursion), then those theorems say that we also have the projection property. Thus in that case \cite{BS17} implies that topological recursion holds.
	
	If the generality conditions of Theorems \ref{thm:Xi-Family1} and \ref{thm:Xi-Family2} are not satisfied (and, in particular, if the condition on the simplicity of zeros of $dx$ is not satisfied),
	we need to use the Bouchard-Eynard recursion. However, this case does not give rise to too many additional complications. In fact, the fact that the Bouchard-Eynard recursion holds for the case when 
	the generality conditions are not satisfied 
	follows from the statement that the ordinary topological recursion holds for all cases when 
	the generality condition is satisfied (this situation is similar to~\cite[Proposition 2.1]{BDS}).

Namely, consider the coefficients of the polynomials $P_i$, $R_i$ participating in the formulations of Theorems~\ref{thm:Family1TR} and~\ref{thm:Family2TR} as parameters of the corresponding families of the spectral curve data. The right hand side of~\eqref{eq:TRtoprove} depends analytically on these parameters. Besides, the two sides agree on an open dense subset of parameters satisfying generality condition. Therefore, it is sufficient to prove that the left hand side of~\eqref{eq:TRtoprove} depends also analytically on these parameters. But this follows from~\cite{BBCKS} which establishes the analyticity of topological/Bouchard-Eynard recursion differentials (where the understanding of the Bouchard-Eynard recursion includes the possible extra terms coming from the multiple poles mentioned in Remark~\ref{rem:multiple-poles}) under certain conditions which are satisfied in the present case (which is checked in~\cite[Section~6.3.1]{BBCKS}).
\end{proof}

\begin{remark}
	Note that the natural analytic assumptions are satisfied for all cases from Tables~\ref{tab:Hurw} and~\ref{tab:Hurw2}, and thus Theorems~\ref{thm:Family1TR} and~\ref{thm:Family2TR} respectively are applicable there.	
Let us also provide a couple of non-examples where the natural analytic assumptions are not satisfied in a not completely trivial way.
For Family~I:
	\begin{equation}
	\psi(y)=\log\left(\frac1 {1-y}\right),\;\; y(z)=\frac{z(1-z^2)}{1-z^2-z^3};
\end{equation}
	for Family~II:
	\begin{equation}
	\psi(y) = \frac{y}{2},\;\; y(z) = \log(1+z^2)
\end{equation}
(with the respective $\hbar$-deformations $\hat \psi$ and $\hat y$ specified by Table~\ref{tab:families}).
For both of these cases everything is satisfied apart from the fact that $d\psi/dy|_{y=y(z)}$ for the first case and $dy$ for the second case respectively are not regular at $\infty$ which is a zero of $dx$ in both of these cases.
\end{remark}

\subsection{Applications}\label{sec:appl}
\subsubsection{New proof of topological recursion for the previously known cases of enumerative Hurwitz-type problems}\label{sec:newprooftr}
Theorems \ref{thm:Family1TR} and \ref{thm:Family2TR}, in particular, provide a uniform, new and independent proof of topological recursion for the following Hurwitz-type problems for which there already exist known proofs:
\begin{itemize}
	\item for the usual simple Hurwitz numbers ($\psi=y,\, y=z$) \cite{BM08,BEMS,EMS11,MZ10,Eynard,DKOSS};
	\item for the usual orbifold Hurwitz numbers ($\psi=y$, $y=z^q$) \cite{bouchard2014,DLN16,DLPS};
	\item for the usual polynomially double Hurwitz numbers ($\psi=y$, $y$ polynomial) \cite{Do-Karev-DoubleHurwitz,BDKLM};
	\item for the coefficients of the extended Ooguri-Vafa partition function of the colored HOMFLY polynomials of torus knots ($\psi=\alpha z$, $y =\log(1-A^{-1}z)-\log(1-Az)$) \cite{DPSS,DKPSS};
	\item for the monotone simple Hurwitz numbers ($\psi=-\log(1-y)$, $y=z$)~\cite{Do-Dyer-Mathews,DKPS-0};
	\item for the monotone orbifold Hurwitz numbers ($\psi=-\log(1-y)$, $y=z^q$) \cite{DoKarev,KPS};
	\item for the enumeration of maps / dessins d'enfants / ribbon graphs ($\psi = \log(1+y)$, $y=z^2$)~\cite{DumitrescuMulaseSafnukSorkin,AndersenChekhovNorburyPenner,Norbury-Lattice};
	\item for the enumeration of more general bi-colored maps or hypermaps, in various parametrizations  (for instance, $\psi = \log(1+y)$, $y=z^q$, or $\psi = \log(1+uy)+\log(1+vy)$, $y=z$, etc.) \cite{Eynard-CombMaps,KZ2015,DOPS-CombLoopEqua,DNOPS-1};
\item for the Bousquet-M\'elou--Schaeffer numbers ($\psi = m\log(1+y)$, $y=z$) \cite{BDS} (the interesting feature of this case is that it does have a multiple pole that contributes non-trivially to the recursion along the lines of Remark~\ref{rem:multiple-poles});
\item for the general polynomially weighted polynomially double Hurwitz numbers (that is, $\exp(\psi)$ and $y$ are polynomials, with some assumptions on general position) \cite{ACEH-1,ACEH-2};
\item for the $r$-spin Hurwitz numbers ($\psi = y^r$, $y=z$) \cite{ShaSpiZvo,BKLPS,DKPS};
\item for the $r$-spin orbifold Hurwitz numbers ($\psi = y^r$, $y=z^q$) \cite{MulShaSpi,KLPS19,BKLPS,DKPS}
\end{itemize}
 (we supply the references to various proofs as well as to the references where the corresponding statements on topological recursion were first conjectured and discussed).

\begin{remark}
This list of previously known cases of topological recursion substantially intersects with Section \ref{sec:PreviousQuasi}, as in many 
cases where there exists an already known proof of the projection property, there also exists an already known proof of topological recursion.
\end{remark}

\begin{remark}
	It is important to stress that the present paper neatly combines all known topological recursion results for Hurwitz-type numbers into one framework, including the ones which require non-trivial $\hbar^2$-deformation, like the coefficients of the extended Ooguri-Vafa partition function and $r$-spin Hurwitz numbers, which previously were thought to be outliers (cf.~\cite[Section 1.4]{DKPS},~\cite[Introduction]{DKPSS}, and~\cite{ACEH-1,ACEH-2}.
\end{remark}

\subsubsection{An example of a natural Hurwitz-type enumerative problem where the topological recursion was not previously known}\label{sec:moreTR}

Note that Theorems \ref{thm:Family1TR} and \ref{thm:Family2TR} (when taken together) are quite a bit more general than all the previously known special cases combined.

For instance, to give just one particular example that is natural from combinatorial and geometric viewpoints, Theorems \ref{thm:Family1TR} and \ref{thm:Family2TR} imply that topological recursion holds for the usual double Hurwitz numbers corresponding to
\begin{equation}
	\hat\psi(\hbar^2,y)=\psi(y)=y,\quad \hat y(\hbar^2,z) = y(z) = \dfrac{z}{1-z}
\end{equation}
with the spectral curve data
\begin{equation}
	\left(\mathbb{C}\mathrm{P}^1, \
	x(z)=\log z - \dfrac{z}{1-z}, \
	y(z)=\dfrac{z}{1-z}, \
	B=\frac{dz_1dz_2}{(z_1-z_2)^2}
	\right).
\end{equation}
In this case we resolve the following enumerative problem: the count of all ramified coverings of the Riemann sphere by a genus-$g$ surface with a fixed monodromy at infinity, and arbitrary monodromy at zero, and an arbitrary number of additional simple ramifications. This appears to be a very natural problem in the general context of Hurwitz-type problems, but, to our knowledge, it has not been studied in the literature up until now.

\subsubsection{ELSV-type formulas}\label{sec:ELSV}
We can also use Theorems~\ref{thm:Family1TR} and~\ref{thm:Family2TR} to uniformly prove various ELSV-type formulas that generalize the classical formula of Ekedahl--Lando--Shapiro--Vainshtein~\cite{ELSV} and relate the combinatorially defined weighted Hurwitz numbers to the intersection theory of the moduli spaces of curves and cohomological field theories. Indeed, in general, if it is given that a particular enumerative problem satisfies the spectral curve topological recursion, the techniques of \cite{Eynard} and~\cite{DOSS} allow to immediately deduce and prove the respective ELSV-type formula, cf. \cite{LPSZ,FangZong,BDKLM}.

In particular, the present paper (suitably complemented by particular computations identifying the tautological cohomology classes in the moduli spaces of curves and/or the corresponding cohomological field theories results for the respective cases --- note, however, that this is not the tricky part, as the general techniques from \cite{Eynard,DOSS} work in a fairly straightforward way) gives new proofs of the following ELSV-type formulas:
\begin{itemize}
	\item the original Ekedahl--Lando--Shapiro--Vainshtein formula \cite{ELSV}; the respective (topological recursion) $\Rightarrow$ (ELSV formula) part is proved in \cite{Eynard,DKOSS};
	\item the Johnson--Pandharipande--Tseng formula \cite{JohnsonPandTseng}; the respective (topological recursion) $\Rightarrow$ (ELSV-type formula) part is proved in \cite{DLPS};
	\item the ELSV-type formula for monotone (orbifold) Hurwitz numbers \cite{ALS,DoKarev}; the respective (topological recursion) $\Rightarrow$ (ELSV-type formula) part is proved in \cite{ALS,DoKarev};
	\item the ELSV-type formula for double Hurwitz numbers \cite{BDKLM};
	\item the ($q$-orbifold) $r$-spin Zvonkine formula \cite{Zvonkine,KLPS19,BKLPS,DKPS}; the respective (topological recursion) $\Rightarrow$ (ELSV-type formula) part is proved in \cite{LPSZ};
	\item the formula relating the extended Ooguri-Vafa partition function of the colored HOMFLY polynomials of torus knots with the open Gromov-Witten invariants of the resolved conifold \cite{DiaconescuShendeVafa,FangZong,DKPSS}; the respective (topological recursion) $\Rightarrow$ (ELSV-type formula) part is proved in \cite{FangZong}.
		\item The Mari\~no--Vafa formula (also called Gopakumar--Mari\~no--Vafa formula in the literature) for the special cubic Hodge integrals conjectured in~\cite{MarinoVafa-conj} and proved in~\cite{LiuLiuZhou,OkounkovPandharipande-GMV}, see an explanation in~\cite[Lemma 2.26]{kramer2021kp}. It corresponds to the case $\psi=ay$, $y=\log(1-z)$ (cf.~\cite[Equation 2.1]{OkounkovPandharipande-GMV}). The respective (topological recursion) $\Rightarrow$ (ELSV-type formula) part is proved in~\cite{Eynard}.
\end{itemize}

Since, as mentioned in Section~\ref{sec:moreTR}, the present paper proves topological recursion also for quite a few previously unknown cases of Hurwitz-type problems, we also naturally get implications for new ELSV-type formulas. Using the general techniques from \cite{Eynard,DOSS} we can routinely and absolutely  straightforwardly produce new ELSV-type formulas. The interesting open question here is to identify the cohomological field theories (possibly with a non-flat unit) that might occur this way and simultaneously also emerge naturally in different mathematical contexts. This, however, goes beyond the scope of the present paper.

\bibliographystyle{alphaurl}
\bibliography{top_rec}
\end{document}